\documentclass[longbibliography]{amsart}
\usepackage[foot]{amsaddr}
\usepackage[unicode=true,pdfusetitle, bookmarks=true,bookmarksnumbered=false,bookmarksopen=false, breaklinks=false,pdfborder={0 0 0},backref=false,colorlinks=false]{hyperref}
\hypersetup{colorlinks,linkcolor=myurlcolor,citecolor=myurlcolor,urlcolor=myurlcolor}
\usepackage{graphics,epstopdf,graphicx,amsthm,amsmath,amssymb,mathptmx,braket,colortbl,color,bm,framed,mathrsfs}
\usepackage[up]{subfigure}
\usepackage{tikz}
\usepackage{multirow}
\usepackage{enumerate}
\usepackage{float}
\usepackage{cite}
\usepackage{fontawesome}
\usepackage{cleveref}
\usepackage{orcidlink}

\usepackage{arydshln}

\usepackage{tabularx}
\usepackage{stmaryrd}

\usepackage{booktabs} % Required for better horizontal rules in tables
\let\hbar\undefined
\usepackage{fouriernc} % Use the New Century Schoolbook font

\usepackage{algorithm}
\usepackage{algpseudocode}

\definecolor{myurlcolor}{rgb}{0,0,0.9}

\newcommand{\proj}[1]{| #1\rangle\!\langle #1 |}

\newcommand{\iinner}[2]{\langle #1 | #2\rangle}

\DeclareMathOperator{\trace}{Tr}
\DeclareMathOperator*{\Expect}{\mathop{\mathbb{E}}}
\newcommand{\Ptr}[2]{\trace_{#1}\Pa{#2}}
\newcommand{\Tr}[1]{\Ptr{}{#1}}

\newcommand{\Pa}[1]{\left[#1\right]}

\newcommand{\norm}[1]{\left\lVert #1 \right\rVert}

\theoremstyle{plain}
\newtheorem{thm}{Theorem}
\newtheorem{lem}[thm]{Lemma}
\newtheorem{prop}[thm]{Proposition}
\newtheorem{cor}[thm]{Corollary}

\newtheorem{Def}[thm]{Definition}
\newtheorem{Rem}[thm]{Remark}

\newtheorem{Examp}[thm]{Example}

\newcommand*{\myproofname}{Proof}
\newenvironment{mproof}[1][\myproofname]{\begin{proof}[#1]}{\end{proof}}
\def\ot{\otimes}

\def\R{\mathbb{R}}
\def\Z{\mathbb{Z}}

\newcommand{\CDD}{\mathcal D}

\newcommand{\CEE}{\mathcal E}

\newcommand{\BFF}{\mathbb F}
\newcommand{\CFF}{\mathcal F}

\newcommand{\CNN}{\mathcal N}

\newcommand{\Bvv}{\mathbf{v}}

\newcommand{\Bxx}{\mathbf{x}}

\newcommand{\Byy}{\mathbf{y}}

\newcommand{\bmu}{{\boldsymbol{\mu}}}

\newcommand{\be}{\begin{equation}}
\newcommand{\ee}{\end{equation}}

\renewcommand{\ge}{\geqslant}
\renewcommand{\geq}{\geqslant}
\renewcommand{\leq}{\leqslant}
\renewcommand{\le}{\leqslant}

\newcommand{\wt}{\operatorname{wt}}

\newcommand{\Dec}{\mathsf{Dec}}

\DeclareMathAlphabet{\mathcal}{OMS}{cmsy}{m}{n}

\makeatother
\title{Multivariate Decoded Quantum Interferometry for Weighted Optimization}
\author{Kaifeng Bu$^{1,2}$\,}
\email{ bu.115@osu.edu (K.Bu, corresponding author)}
 \address{$^1$\textnormal{Department of Mathematics, The Ohio State University, Columbus, Ohio 43210, USA}}
\address{$^2$\textnormal{Department of Physics, Harvard University, Cambridge, Massachusetts 02138, USA}}

\author{Weichen Gu$^1$}

\author{Xiang Li$^1$}

\begin{document}
\begin{abstract}

Decoded Quantum Interferometry (DQI) is a recently introduced quantum algorithm that reduces discrete optimization to decoding
with potential advantages over the best-known polynomial-time classical algorithms for certain Max-LINSAT problems. 
In its original formulation, however, DQI treats all constraints uniformly and cannot exploit the weight structure present in most optimization problems of interest. In this work, we develop multivariate Decoded Quantum Interferometry (multivariate DQI) for weighted optimization problems, focusing on the weighted Max-LINSAT problem over a prime field. Grouping constraints into $N$ blocks by distinct weights, we introduce multivariate DQI states built from $N$-variable polynomials of bounded total degree, and derive a closed-form asymptotic expression for both their optimal expectation value and their concentration behavior. 
We give an explicit preparation circuit using a single decoder call, and extend the analysis to imperfect decoding.
We also show that, for certain weighted OPI problems, multivariate DQI outperforms a natural weighted analogue of Prange's algorithm, which serves as the weighted counterpart of the classical benchmark used in the unweighted setting.
Finally, we extend the ideas to Hamiltonian DQI, obtaining approximate Gibbs states for commuting Pauli Hamiltonians with block structure.

\end{abstract}

\maketitle

\tableofcontents

\section{Introduction}

Quantum optimization---the task of using quantum algorithms to identify optimal or near-optimal solutions from an exponentially large space of feasible configurations---has emerged as one of the most promising pathways toward achieving practical quantum advantage~\cite{abbas2024challenges,leng2025sub,pirnay2024principle,huang2025vast}. Several families of quantum optimization algorithms have been extensively investigated over the past two decades, including amplitude-amplification-based methods such as Grover's algorithm~\cite{grover1996fast}, adiabatic and annealing-based quantum algorithms~\cite{farhi2000quantum,Roland2002QuantumSearch,aharonov2008adiabatic,Das2008QuantumAnnealing,albash2018adiabatic}, and variational methods such as
Variational quantum eigensolver (VQE)~\cite{peruzzo2014variational,mcclean2016theory,kandala2017hardware,cao2019quantum,cerezo2021variational,tilly2022variational}  and Quantum Approximate Optimization Algorithm (QAOA)~\cite{farhi2014quantum,ZhouQAOA20,herrman2022multi,vijendran2023expressive,shi2022multiangle,zhao2025symmetry}. Although these methods have shown encouraging empirical performance, their theoretical guarantees remain limited, and in many regimes classical algorithms can match or even surpass their performance. Understanding when quantum optimization algorithms provide a provable advantage therefore remains a central open problem.

A significant step in this direction was recently achieved by Jordan et al., who introduced \emph{Decoded Quantum Interferometry} (DQI)~\cite{jordan2024optimization}. DQI is a quantum algorithm that reduces a class of discrete optimization problems to the decoding of a classical linear code,
based on quantum Fourier transform.
Such ideas have also been used in other works~\cite{regev2009lattices,aharonov2003adiabatic,yamakawa2024verifiable,chen2022quantum}. 
The original DQI prepares a superposition whose amplitudes are specified by a suitably chosen univariate polynomial of the objective function; after applying a syndrome decoder, the resulting measurement statistics concentrate on high-quality solutions.
These features point to a scalable framework with the potential to deliver speedups for certain problem classes~\cite{khattar2025verifiable}. The computational complexity of DQI has since been studied in a variety of settings \cite{marwaha2025complexity,sabater2025towards,Parekh2025no,hillel2025optimization,gu2025algebraic,Anschuetz2025decoded,wang2025kernelized,rosmanis2026nearly,gao2026hidden,kramer2026tight}, and its robustness to noise and imperfect operations has also begun to be explored \cite{bu2025decoded}, further highlighting its promise for both near-term and fault-tolerant quantum devices. The DQI paradigm has also been extended beyond optimization to the preparation of polynomial functions of Hamiltonians as approximations to Gibbs states, giving rise to \emph{Hamiltonian DQI}~\cite{schmidhuber2025hamiltonian,bu2026hamiltonian}.

The Max-LINSAT problem over a prime field $\mathbb{F}_p$ is the task of maximizing the number of linear constraints $\mathbf{b}_i \cdot \mathbf{x} \in L_i$ that are satisfied, for a matrix $B \in \mathbb{F}_p^{m \times n}$ and target sets $L_i \subseteq \mathbb{F}_p$. For this problem, DQI achieves an asymptotic satisfaction ratio that strictly exceeds what is known to be achievable in polynomial time classically, in a well-defined parameter regime. One instance of this advantage arises in the Optimal Polynomial Intersection (OPI) problem, where the original DQI provably beats Prange's algorithm, the best known polynomial-time classical benchmark in certain regimes~\cite{jordan2024optimization,khattar2025verifiable}. 

In many applications of combinatorial optimization, constraints are not equally important. Instead, each constraint $i$ is assigned a positive weight $c_i>0$, and the objective is to maximize the weighted sum
\[
F_{\mathbf{c}}(\mathbf{x}) = \sum_{i=1}^m c_i\, f_i(\mathbf{b}_i \cdot \mathbf{x}).
\]
The original DQI analysis applies only to the uniform-weight case $c_i \equiv  1$, where the amplitudes of the DQI state are determined by a single univariate polynomial in the scalar objective
$F(\mathbf{x})=\sum_{i=1}^m f_i(\mathbf{b}_i\cdot \mathbf{x}).
$ This uniform-weight structure also leads naturally to Dicke states in the ancilla registers, which is a key ingredient in the efficient quantum preparation procedure.
When the weights take several distinct values, however, this uniform treatment is generally suboptimal: clauses with larger weights should be prioritized.
Hence, finding the optimal univariate polynomial for the weighted objective becomes a genuinely new optimization problem.

In this work, we develop a multivariate DQI framework for weighted optimization problems, with a focus on weighted Max-LINSAT problems with block structure. Specifically, we partition the index set $[m]$ as
$[m]=S_1\sqcup \cdots \sqcup S_N,$
and assume that the constraints in block $S_t$ all carry the same weight $g_t$, that is, $c_j=g_t$ for every $j\in S_t$. For each block, we define the partial objective
$F_t(\mathbf{x})=\sum_{i\in S_t} f_i(\mathbf{b}_i\cdot \mathbf{x}),$
and then the objective function becomes 
\[
F_{\mathbf{g}}(\mathbf{x}) = \sum^N_{t=1}g_t F_t(\mathbf{x}).
\]
We  consider multivariate DQI states of the form
\[
|P(F_1,  \dots,F_N)\rangle
=
\sum_{\mathbf{x}\in\mathbb{F}_p^n}
P\bigl(F_1(\mathbf{x}), \dots,F_N(\mathbf{x})\bigr)\,|\mathbf{x}\rangle,
\]
where $P$ is a suitably normalized polynomial in $N$ variables of  total degree $l$.
Our goal is to determine the optimal expected weighted objective value achievable by such multivariate DQI states.

\subsection{Summary of main results}
We list our main results as follows. 

\textbf{1. Performance of multivariate DQI (Section~\ref{sec:perform_MDQI}).} Our main result, Theorem~\ref{thm:pf_semicircle_general}, determines the asymptotically optimal expectation value achievable by any multivariate DQI state of total degree at most \(l\). Under the minimum-distance condition \(2l+1<d^\perp\), we prove that
\[
\frac{\langle s_{\mathbf{g}} \rangle}{m} = \Bigl(\tfrac{2r}{p} - 1\Bigr)\frac{1}{m}\sum_{t=1}^N g_t m_t + \frac{2\sqrt{r(p-r)}}{p}\, \Gamma_{\mathbf{g},\mathbf{m},\kappa}(\mu) + O_{\kappa}\!\left(\frac{1}{m}\sum_t g_t m_t^{3/4}\right),
\]
where \(\Gamma_{\mathbf{g},\mathbf{m},\kappa}(\mu)\) is an explicit variational functional that describes how the decoding budget is optimally distributed across the \(N\) weight blocks. For the two-block case \(N=2\), we further analyze three scaling regimes in Proposition~\ref{prop:pf_two_point_crossover}. In addition, we show that measurements of the multivariate DQI state are concentrated on a set of solutions whose objective values achieve the asymptotically optimal expectation value. 
In addition, we compare multivariate DQI with the original (univariate) DQI framework and show that multivariate DQI states can outperform univariate DQI states for non-uniform weight profiles (see Section~\ref{subsec:compare-univariate}).

\textbf{2. Preparation algorithm (Section~\ref{sec:gen_MDQI}).} We provide an explicit quantum algorithm that prepares the multivariate DQI state, extending the standard DQI circuit to a blockwise fashion. The algorithm uses $N$ weight registers, $N$ error registers, and a single syndrome register, and exploits the product form of the optimal coefficients $w_{\mathbf{j}} = \prod_t a^{(t)}_{j_t}$. A single call to a syndrome decoder for the dual code $C^\perp = \{\mathbf{y} \in \mathbb{F}_p^m : B^\top \mathbf{y} = 0\}$ suffices, exactly as in the unweighted case. We describe the $\mathbb{F}_2$ protocol in detail and then extend it to general prime $p$ via an appropriate use of the quantum Fourier transform over $\mathbb{F}_p$.

\textbf{3. Robustness to imperfect decoding (Section~\ref{sec:no_min_dist}).} The condition $2l + 1 < d^\perp$ is a sufficient condition for bounded-distance decoding, but is not always satisfied in practice. In Theorem~\ref{thm:f2_nd_average}, we give a performance guarantee for the multivariate DQI state prepared with an \emph{imperfect} decoder, bounding the loss from the ideal value by blockwise effective decoder failure rates $\widetilde{\gamma}_{\max}$ defined on block-weight layers of the error space.

\textbf{4. Weighted OPI: comparison with a natural weighted Prange benchmark (Section~\ref{sec:opi}).} We apply our framework to the weighted OPI problem with two balanced weight blocks, in which the dual code is a Reed--Solomon code and a Berlekamp--Massey decoder can be used inside DQI. 
We compare multivariate DQI with a natural weighted analogue of Prange's algorithm, obtained by prioritizing the heavier block.
We show in Proposition~\ref{prop:weighted_opi_dqi_beats_prange} that for every $g \geq 1$ and every $x = n/p \in (0,1)$,
\[
R^{\mathrm{DQI}}_g(x) > R^{\mathrm{Pr}}_g(x),
\]
with the symmetric statement holding for $0 < g \leq 1$. Thus, the strict asymptotic advantage of multivariate DQI over weighted Prange's algorithm in balanced OPI 
persists across the entire non-trivial regime of weight imbalances, extending well beyond the case of uniform weights.

\textbf{5. Block-structured  Hamiltonian DQI (Section~\ref{sec:HDQI}).} 
The block-polynomial viewpoint is not limited to classical objective functions. 
It also applies naturally to commuting Pauli Hamiltonians with block structure. 
Specifically, for a commuting Pauli Hamiltonian
\[
H_{\mathbf g}=\sum_{t=1}^N g_t H_t,\qquad 
H_t=\sum_{a\in[m_t]}P_{t,a},
\]
we prove that for any real polynomial $P$ of degree $l$, the state $\rho_{P}(H_{\mathbf{g}}) =\frac{P(H_{\mathbf{g}})^2}{ \mathrm{Tr}[P(H_{\mathbf{g}})^2]}$ can be prepared with a single call to a weight-$l$ decoder for the classical linear code associated with $H_{\mathbf{g}}$ (see Theorem~\ref{thm:block_pauli_poly_state}).
This leads to a simplified Hamiltonian DQI protocol for block-structured commuting Pauli Hamiltonians: the reference state is described by linear combinations of product Dicke states, rather than by a general matrix-product-state representation.

\subsection{Organization}
In Section~\ref{sec:preliminary},  we introduce the basic notation and background. We study the performance of multivariate DQI states in Section~\ref{sec:perform_MDQI}, and provide their preparation procedure in Section~\ref{sec:gen_MDQI}. In addition, we discuss imperfect decoding in Section~\ref{sec:no_min_dist}, and study the potential advantage in weighted OPI in Section~\ref{sec:opi}. Finally, in Section~\ref{sec:HDQI},  we extend the ideas  to block-structured Hamiltonian DQI.  The detailed proofs of the main results are collected in the appendix.

\section{Preliminary}\label{sec:preliminary} 
\begin{Def}[Weighted Max-LINSAT Problem]\label{Def:1}
Let $\mathbb{F}_p$ be a finite field and let $B \in \mathbb{F}_p^{m \times n}$. 
For each $i=1,\ldots,m$, let $c_i\in \R$ and let $L_i \subset \mathbb{F}_p$ be an arbitrary subset of $\mathbb{F}_p$.
The weighted Max-LINSAT problem is to find $\mathbf{x} \in \mathbb{F}_p^n$ that maximizes the function
\begin{align}\label{def:targ_eq1}
    F_c(\Bxx) = \sum_{i=1}^m c_i f_i\left(\mathbf{b}_i \cdot \mathbf{x} \right),
\end{align}
where $\mathbf{b}_i$ denotes the $i$-th row of the matrix $B$, $\mathbf{b}_i \cdot \mathbf{x}=\sum_{j=1}^n B_{ij} x_j $, and  each function $f_i$ is defined as follows
\begin{align}\label{eq:def1}
f_i(y) = \left\{ 
\begin{aligned}
&1, && \text{ if } y\in L_i,\\
&-1, && \text{ if } y\not\in L_i.
\end{aligned}\right.    
\end{align}
\end{Def} 

 Without loss of generality, we assume $c_i>0$ throughout the paper.
 Indeed,
if $c_i<0$ for some index $i$, we may replace $c_i$ by $|c_i|$ and simultaneously replace the set $L_i$ by its complement $\mathbb{F}_p \setminus L_i$.
By \eqref{eq:def1}, this flips the sign of $f_i$ and therefore leaves the objective function unchanged.

In the Boolean domain where $p=2$, we have $|L_i|=1$ for every index $i$. Let $v_i$ denote the unique element in $L_i$. The function $f_i: \mathbb{F}_2 \to \{1, -1\}$ is 
\begin{align} \label{eq:vec_f2}
f_i(y) = \left\{ 
\begin{aligned}
&1, && \text{ if } y= v_i;\\
&-1, && \text{ if } y\neq v_i.
\end{aligned}\right.    
\end{align}
For \(p=2\), Definition~1 becomes the weighted Max-XORSAT problem, whose goal is to
find a vector $\mathbf{x} \in \mathbb{F}_2^n$ that maximizes the objective function:
\begin{align}\label{260202eq3}
F_c(\mathbf{x}) = \sum_{i=1}^m c_i f_i(\mathbf{b}_i \cdot \mathbf{x})
=\sum_{i=1}^m c_i (-1)^{\mathbf{b}_i \cdot \mathbf{x}+v_i}.
\end{align}

Note that, throughout the paper,  we adopt the following standard conventions for asymptotic notation.
For two nonnegative  functions $f(x)$ and $g(x)$, 
$f = O(g)$ means there exists an absolute constant $C>0$ such that 
$f(x)\leq Cg(x)$ for all sufficiently large $x$. 
$f = o(g)$ means $f(x)/g(x)\to 0$ as $x\to \infty$. 
$f = \Theta(g)$ means $f=O(g)$ and $g=O(f)$.
$f = O_d(g)$ means there exists a constant $C_d
>0$ that may depend on the parameter $d$ such that $f(x)\leq C_dg(x)$ for all sufficiently large $x$.

\section{Performance of Multivariate DQI States}\label{sec:perform_MDQI}
Assume that the weights $c_i$ are restricted to $N$ distinct positive values. Specifically,  let $g_1, \dots, g_N$ be a set of given constants, and let $[m] = S_1 \sqcup \dots \sqcup S_N$ be a partition of the index set $[m] = \{1, \dots, m\}$ with size  $m_t := |S_t|$.
Then the weight $c_j=g_t$ for $j\in S_t$.
Hence, the target function in \eqref{def:targ_eq1} can be written as follows
\begin{align}
    F_{\mathbf{g}}(\Bxx)
    =\sum^N_{t=1}g_tF_t(\mathbf{x} ),
\end{align}
where $F_t(\mathbf{x} )=\sum_{i\in S_t}f_i\left(\mathbf{b}_i \cdot \mathbf{x} \right)$
for each $t\in [N]$.

Here, we will use a  multivariate variant of the DQI algorithm  based  on the following code
\begin{align}
  C^\perp = \{ \mathbf{y} \in \mathbb{F}_p^m : B^\top \mathbf{y} = \mathbf{0} \} .
\end{align}
If the above code $C^\perp$ can be efficiently decoded for up to $l$ errors, i.e., there exists a reversible decoder  $\Dec_B^{(l)}$ such that 
\begin{align}\label{eq:decoder}
\Dec_B^{(l)}:\ket{\mathbf y}\ket{B^\top \mathbf y}
\to \ket{0}\ket{B^\top \mathbf y}, \quad\forall |\mathbf y|\leq l, 
\end{align}
 we can efficiently  generate the following 
multivariate DQI state
\begin{align}
    \ket{P(F_1, \ldots, F_N)}
    =\sum_{\Bxx \in \mathbb{F}^n_p}
   P(F_1 (\Bxx), \ldots, F_N(\Bxx))\ket{\mathbf x},
\end{align}
where $P$ is an appropriately normalized multivariate polynomial 
in $N$ variables with  a total degree at most $l$, i.e., $\sum_{\Bxx \in \mathbb{F}^n_p} |P(F_1 (\Bxx), \ldots, F_N(\Bxx))|^2=1$.
The dependence on the weights enters through the block decomposition \(S_1, \ldots, S_N\), and hence through the block objectives \(F_1, \ldots, F_N\). After taking the measurement in the computational basis, the output probability 
of a bit string $\Bxx$ is 
$ |P(F_1 (\Bxx), \ldots, F_N(\Bxx))|^2$. 
The goal is to choose the multivariate polynomial $P$ so that the output probability distribution concentrates on bitstrings $\mathbf{x}$ that correspond to large objective values.

\subsection{Asymptotic performance of the optimal expectation}
Let us first study the asymptotic performance of the expectation value of the weighted Max-LINSAT objective function obtained by measuring the corresponding multivariate DQI state in the limit \(m\to\infty\).

\begin{thm}[Asymptotically Optimal Expectation Value]
\label{thm:pf_semicircle_general}
Let $p$ be a prime, $B \in \mathbb{F}_p^{m \times n}$ be a matrix, and  let  $ f_i:\mathbb{F}_p \to \{\pm 1\}$ be functions with $|f_i^{-1}(+1)|=r$ for all $i$. 
Consider a partition of the index set $[m] = S_1 \sqcup \dots \sqcup S_N$ with subset sizes $m_t := |S_t|$, and let $F_{\mathbf{g}}(\mathbf{x})$ be the weighted Max-LINSAT objective function:
$$F_{\mathbf{g}}(\mathbf{x}) = \sum_{t=1}^N g_t F_t(\mathbf{x}), \quad F_t(\mathbf{x}) = \sum_{i \in S_t} f_i(\mathbf{b}_i \cdot \mathbf{x}),\quad g_t>0, \forall t\in [N].$$
Given an $N$-variable multivariate polynomial $P$ with total degree at most $l$,
let $\langle s_{\mathbf{g}} \rangle$ denote the expectation value obtained upon measuring the corresponding multivariate  DQI state $\ket{P(F_1,\ldots, F_N)}$.
Suppose $2l + 1 < d^\perp$, where $d^\perp$ is the minimum distance of the dual code $C^\perp = \{ \mathbf{d} \in \mathbb{F}_p^m : B^\top \mathbf{d} = \mathbf{0} \}$.
With  fixed ratio $\mu = l/m \in (0, 1/2)$,  for the  optimal choice of multivariate polynomial $P$,
the normalized expectation value satisfies

\begin{equation}
\label{eq:pf_semicircle_general}
\frac{\langle s_{\mathbf g} \rangle}{m}
=
\left(\frac{2r}{p}-1\right)\frac{1}{m}\sum_{t=1}^N g_t m_t
+
\frac{2\sqrt{r(p-r)}}{p}\,\Gamma_{\mathbf g,\mathbf m,\kappa}(\mu)
+
O_{\kappa}\!\left(
\frac{1}{m}\sum_{t=1}^N g_t m_t^{3/4} 
\right),
\end{equation}
where $\kappa=\frac{p-2r}{\sqrt{r(p-r)}}$, and
\begin{align}\label{eq:im_form}
    \Gamma_{\mathbf g,\mathbf m,\kappa}(\mu):=
\sup_{\substack{0\le \alpha_t\le 1\ \forall t\in [N]\\ \sum_{t=1}^N \frac{m_t}{m}\alpha_t\le \mu}}
\sum_{t=1}^N \frac{m_t}{m}g_t\left(\kappa\alpha_t+2\sqrt{\alpha_t(1-\alpha_t)}\right).
\end{align}
Suppose further that $N$ is fixed and $\max_t g_t = O(1)$,  the error term in \eqref{eq:pf_semicircle_general} is $o_\kappa(1)$
as $m \to \infty$.
\end{thm}
For the special case \(N=1\), corresponding to the uniform-weight case, this recovers the result of~\cite{jordan2024optimization}.

\begin{mproof}[Proof sketch]
    The detailed proof of Theorem \ref{thm:pf_semicircle_general} is presented in Appendix~\ref{append:A}, which consists of  three parts: (a) reduce expectation to a matrix eigenvalue problem, (b) upper bound via Collatz–Wielandt, (c) lower bound via explicit product test vector.
    Here we provide a proof sketch.
    
    All functions $f_i$ have the same mean value $\bar f:=\frac1p\sum_{y\in \mathbb F_p} f_i(y)=\frac{2r}{p}-1.$
Let us define the normalization constant 
$\varphi:=\left(\sum_{y\in \mathbb F_p}|f_i(y)-\bar f|^2\right)^{1/2}
=
\sqrt{4r\left(1-\frac{r}{p}\right)},$
and define the rescaled centered functions
$h_i(y):=\frac{\sqrt p}{\varphi}\bigl(f_i(y)-\bar f\bigr)$ for $i\in [m]$.
Then $h_i$ satisfies the following properties:
$h_i(y)^2 = 1+\kappa\,h_i(y)$, with $\kappa=\frac{p-2r}{\sqrt{r(p-r)}}$.
Hence,
\begin{align*}
f_i(y)=\bar f+\frac{\varphi}{\sqrt p}\,h_i(y)
=
\left(\frac{2r}{p}-1\right)+\frac{2\sqrt{r(p-r)}}{p}\,h_i(y).
\end{align*}

For each $t\in [N]$ and each $k\geq 0$, denote 
\begin{align*}
P^{(k)}_{H_t}
:=
\sum_{\mathbf{x} \in \mathbb{F}_p^n}
\left(
\sum_{\substack{s_1,\ldots,s_k\in S_t\\ s_1<\cdots<s_k}}
h_{s_1}\!\left(\mathbf{b}_{s_1} \cdot \mathbf{x} \right)\cdots
h_{s_k}\!\left(\mathbf{b}_{s_k} \cdot \mathbf{x} \right)
\right)
\ket{\mathbf{x}}\bra{\mathbf{x}}.
\end{align*}
Let us consider the block-symmetric polynomial state
\begin{align}
    \ket{\mathbf{P}(\mathbf j)}:
    =\frac{1}{\sqrt{\prod^N_{t=1}\binom{m_t}{j_t}}}\frac{1}{\sqrt{p^n}}\sum_{\Bxx\in \mathbb{F}^n_p}\prod^N_{t=1}P^{(j_t)}_{H_t}\ket{\Bxx},
\end{align}
where each integer $j_t\geq 0$. 
Let $T_l$ denote the set of indices bounded by the total degree $l$, that is, 
$T_l = \left\{(j_1,\ldots,j_N)\in \mathbb Z_{\ge 0}^N: j_t\le m_t\text{ for all }t,\text{ and } j_1+\cdots+j_N\le l\right\}.$

We establish two key properties. 
First, we can show that any multivariate DQI state $ \ket{P(F_1, \ldots, F_N)}$ for a multivariate polynomial $P$ with degree at most $l$ can be 
written as a linear combination of block-symmetric polynomial states, that is, 
$$\ket{P(F_1, \ldots, F_N)}=\sum_{\mathbf j\in T_l}w_{\mathbf j} \ket{\mathbf{P}(\mathbf j)}.$$
Second, the set
$\{\ket{\mathbf P(\mathbf j)} : \sum^N_{t=1}j_t\leq l\}$ is orthonormal.
Hence, finding 
the optimal expected value is equivalent to finding the largest eigenvalue of a specific matrix, where the optimal state is determined by the corresponding eigenvector $w = (w_{\mathbf{j}})_{\mathbf{j} \in T_l}$. 

We provide an analytical method for determining this maximum value in the limit \(m\to\infty\). In particular, we derive an upper bound using the Collatz--Wielandt formula~\cite{HornJohnson2012}, and we construct a matching lower bound by choosing coefficients with a product-form decomposition,
$w_{\mathbf j}=\prod_{t=1}^N a^{(t)}_{j_t}.$
This product ansatz achieves the optimal expectation value in the asymptotic limit.

\end{mproof}

The general formula in Theorem  \ref{thm:pf_semicircle_general}  simplifies considerably under a constant-density scaling assumption, in which each block size $m_t$
grows linearly with $m$.

\begin{cor}[Constant Density Scaling]
Under the same assumptions as Theorem \ref{thm:pf_semicircle_general}, 
suppose further that 
$N$ is fixed and the weights $\set{g_t}_t$ are constants independent of $m$
and that the size of each subset $S_t$ satisfies $m_t = \theta_t m$ for fixed positive constants $\theta_t$. In the limit $m \to \infty$, the optimal expectation value in \eqref{eq:pf_semicircle_general} reduces to
\begin{align}
\lim_{m\to \infty}   \frac{\langle s_{\mathbf g} \rangle}{m} 
=\left(\frac{2r}{p}-1\right)\sum_{t=1}^N g_t \theta_t
+\frac{2\sqrt{r(p-r)}}{p}
\sup_{\substack{0\le \alpha_t\le 1\ \forall t\in [N]\\ \sum_{t=1}^N \theta_t\alpha_t\le \mu}}
\sum_{t=1}^N g_t\theta_t\left(\kappa\alpha_t+2\sqrt{\alpha_t(1-\alpha_t)}\right).
\end{align}

In the specific case of the finite field $\mathbb{F}_2$ (where $p=2$ and $r=1$), the coefficient $\kappa$ vanishes. Under these conditions, the formula simplifies further to
\begin{align}\label{eq:simp_F_2}
\lim_{m\to \infty}   \frac{\langle s_{\mathbf g} \rangle}{m} 
=
\sup_{\substack{0\le \alpha_t\le 1\ \forall t\in [N]\\ \sum_{t=1}^N \theta_t\alpha_t\le \mu}}
2\sum_{t=1}^N g_t\theta_t\sqrt{\alpha_t(1-\alpha_t)}.
\end{align}
\end{cor}

The performance of the optimal expectation value in \eqref{eq:pf_semicircle_general} is characterized by the functional $\Gamma_{\mathbf{g}, \mathbf{m}, \kappa}(\mu)$, which is generally difficult to evaluate for arbitrary $N$.
Here, let us focus on the special case $N=2$. 
In this setting, the target function in \eqref{def:targ_eq1} becomes
\begin{align}
    F_{\mathbf{g}}
    =g_1F_1(\mathbf{x} )+g_2F_2(\mathbf{x} ),
\end{align}
where $F_t(\mathbf{x} )=\sum_{i\in S_t}f_i\left(\mathbf{b}_i \cdot \mathbf{x} \right)$ for $t=1,2$.
Without loss of generality, we set $g_1=1$ and $g_2=g>0$. The optimal expected value in \eqref{eq:pf_semicircle_general} then reduces to
\begin{align}
    \frac{\langle s_{\mathbf g} \rangle}{m}
=
\left(\frac{2r}{p}-1\right)\frac{m_1+gm_2}{m}
+
\frac{2\sqrt{r(p-r)}}{p}\,\Gamma_{\mathbf g,\mathbf m,\kappa}(\mu)
+
O_{\kappa}\!\left(
\frac{m_1^{3/4}+gm_2^{3/4}}{m}
\right),
\end{align}
where $\Gamma_{\mathbf g,\mathbf m,\kappa}(\mu)$ is defined as in \eqref{eq:im_form} for $N=2$.
Under these conditions, the following result refines the asymptotic formula by analyzing how the relative weight $g$
governs which block dominates the expectation value and by identifying three qualitatively distinct regimes.

\begin{prop}[Scaling Regimes for Different Weights]\label{prop:pf_two_point_crossover}
    Under the same assumptions as Theorem \ref{thm:pf_semicircle_general} with $N=2$, $g_1=1$, $g_2=g$, and 
    $\theta_t=m_t/m$. Then the normalized expected value  $\langle s_{\mathbf g}\rangle/m$  has the following regimes
    based on the choice of $g$:

 \begin{enumerate}[(I)]
 \item \textbf{Weak-Weight Regime ($g \ll \theta_1/\theta_2$, i.e., $g\theta_2=o(\theta_1)$):} the expectation value is dominated by the first block,

\begin{equation}\label{eq:pf_two_point_small_regime}
\frac{\langle s_g\rangle}{m}
=
\left(\frac{2r}{p}-1\right) \theta_1 
+  \frac{2\sqrt{r(p-r)}}{p} \theta_1
\sup_{0\le \alpha\le \min\{1,\mu/\theta_1\}}
\phi_{\kappa}(\alpha) 
+
O_{\kappa}\!\left(
g\theta_2+\Delta_m
\right),
\end{equation}
with $\Delta_m=\frac{m_1^{3/4}+g\,m_2^{3/4}}{m}$, 
and $\phi_{\kappa}(\alpha):=\kappa\alpha+2\sqrt{\alpha(1-\alpha)}.$ \\

\item \textbf{Crossover Regime ($g \asymp \theta_1/\theta_2$, i.e., $g\theta_2=\Theta(\theta_1)$):}
both blocks contribute to the bulk properties, 
\begin{align}\label{eq:pf_two_point_crossover_regime}
\frac{\langle s_g\rangle}{m}
=
\left(\frac{2r}{p}-1\right)(\theta_1+g\theta_2)
+
\frac{2\sqrt{r(p-r)}}{p}\,\Gamma_g^{(2)}(\mu)
+
O_{\kappa}\!\left(
\Delta_m
\right),
\end{align}
where $   \Gamma_g^{(2)}(\mu):=
\sup_{\substack{0\le \alpha_1,\alpha_2\le 1\\ \theta_1\alpha_1+\theta_2\alpha_2\le \mu}}
\Bigl(\theta_1\phi_{\kappa}(\alpha_1)+g\theta_2\phi_{\kappa}(\alpha_2)\Bigr).$ 

\item \textbf{Strong-Weight Regime ($g \gg \theta_1/\theta_2$, i.e., $\theta_1=o(g\theta_2)$):}
the expectation value is dominated by the second block, 
\begin{equation}
\frac{\langle s_g\rangle}{m}
=\left(\frac{2r}{p}-1\right)
g\theta_2 
+ \frac{2\sqrt{r(p-r)}}{p}g \theta_2
\sup_{0\le \alpha\le \min\{1,\mu/\theta_2\}}
\phi_{\kappa}(\alpha)
+
O_{\kappa}\!\left(
\theta_1+\Delta_m
\right).
\label{eq:pf_two_point_large_regime}
\end{equation}
% with
% $\theta_1+\Delta_m\ll g\theta_2 .$ \red{(Gu: here we can delete "$\theta_1+\Delta_m\ll g\theta_2 .$")}
\end{enumerate}

\end{prop}
\begin{proof}

Applying Theorem~\ref{thm:pf_semicircle_general}, we obtain the master formula:
\begin{equation}
\frac{\langle s_g \rangle}{m} = \left(\frac{2r}{p}-1\right)(\theta_1+g\theta_2) + \frac{2\sqrt{r(p-r)}}{p}\,\Gamma_g^{(2)}(\mu) + O_{\kappa} \left( \Delta_m \right),
\label{eq:pf_two_point_master_formula}
\end{equation}
where $\Delta_m = \frac{m_1^{3/4}+g\,m_2^{3/4}}{m}$. Let $M_{\kappa} := \sup_{0\le \alpha\le 1} |\phi_{\kappa}(\alpha)|$, which is finite due to the continuity of $\phi_{\kappa}$ on the compact interval $[0,1]$.

Case I: The weak-weight regime ($g \ll \theta_1/\theta_2$).
For every feasible pair
$(\alpha_1,\alpha_2)$ in the definition of $\Gamma_g^{(2)}(\mu)$, the pair
$(\alpha_1,0)$ is also feasible. It follows that 
\[
\theta_1\phi_{\kappa}(\alpha_1)+g\theta_2\phi_{\kappa}(\alpha_2)
\le
\Gamma_0^{(2)}(\mu)+M_{\kappa}\,g\theta_2.
\]
Taking the supremum over all feasible pairs, we find 
\[
0\le \Gamma_g^{(2)}(\mu)-\Gamma_0^{(2)}(\mu)\le M_{\kappa}\,g\theta_2.
\]
Noting that
\[
\Gamma_0^{(2)}(\mu)
=
\theta_1\sup_{0\le \alpha\le \min\{1,\mu/\theta_1\}}\phi_{\kappa}(\alpha),
\]
substitution into  \eqref{eq:pf_two_point_master_formula} yields
\eqref{eq:pf_two_point_small_regime}.

Case II: The crossover regime ($g \asymp \theta_1/\theta_2$). In this regime, \eqref{eq:pf_two_point_crossover_regime} is equivalent to formula \eqref{eq:pf_two_point_master_formula}. Because $g \theta_2$ and $\theta_1$ are of the same order, both blocks contribute at the leading order.

   Case III: The strong-weight regime ($g \gg \theta_1/\theta_2$). 
Let
$\Phi_2(\mu):=
\sup_{0\le \alpha\le \min\{1,\mu/\theta_2\}}\phi_{\kappa}(\alpha).$
Since every feasible pair $(\alpha_1,\alpha_2)$ satisfies
\[
0\le \alpha_2\le \min\{1,\mu/\theta_2\},
\]
we have
\[
\theta_1\phi_{\kappa}(\alpha_1)+g\theta_2\phi_{\kappa}(\alpha_2)
\le
M_{\kappa}\,\theta_1+g\theta_2\Phi_2(\mu).
\]
Taking the supremum gives
\[
\Gamma_g^{(2)}(\mu)\le g\theta_2\Phi_2(\mu)+M_{\kappa}\,\theta_1.
\]
Conversely, if $\alpha_0$ attains the supremum in the definition of
$\Phi_2(\mu)$, then $(0,\alpha_0)$ is feasible, and therefore
\[
\Gamma_g^{(2)}(\mu)\ge g\theta_2\Phi_2(\mu).
\]
Hence
\[
\Gamma_g^{(2)}(\mu)
=
g\theta_2\Phi_2(\mu)+O_{\kappa}(\theta_1).
\]
Substituting this into \eqref{eq:pf_two_point_master_formula} yields
\eqref{eq:pf_two_point_large_regime}.

\end{proof}

\begin{Examp}

\begin{figure}[h!]
    \centering
    \includegraphics[width=1\textwidth]{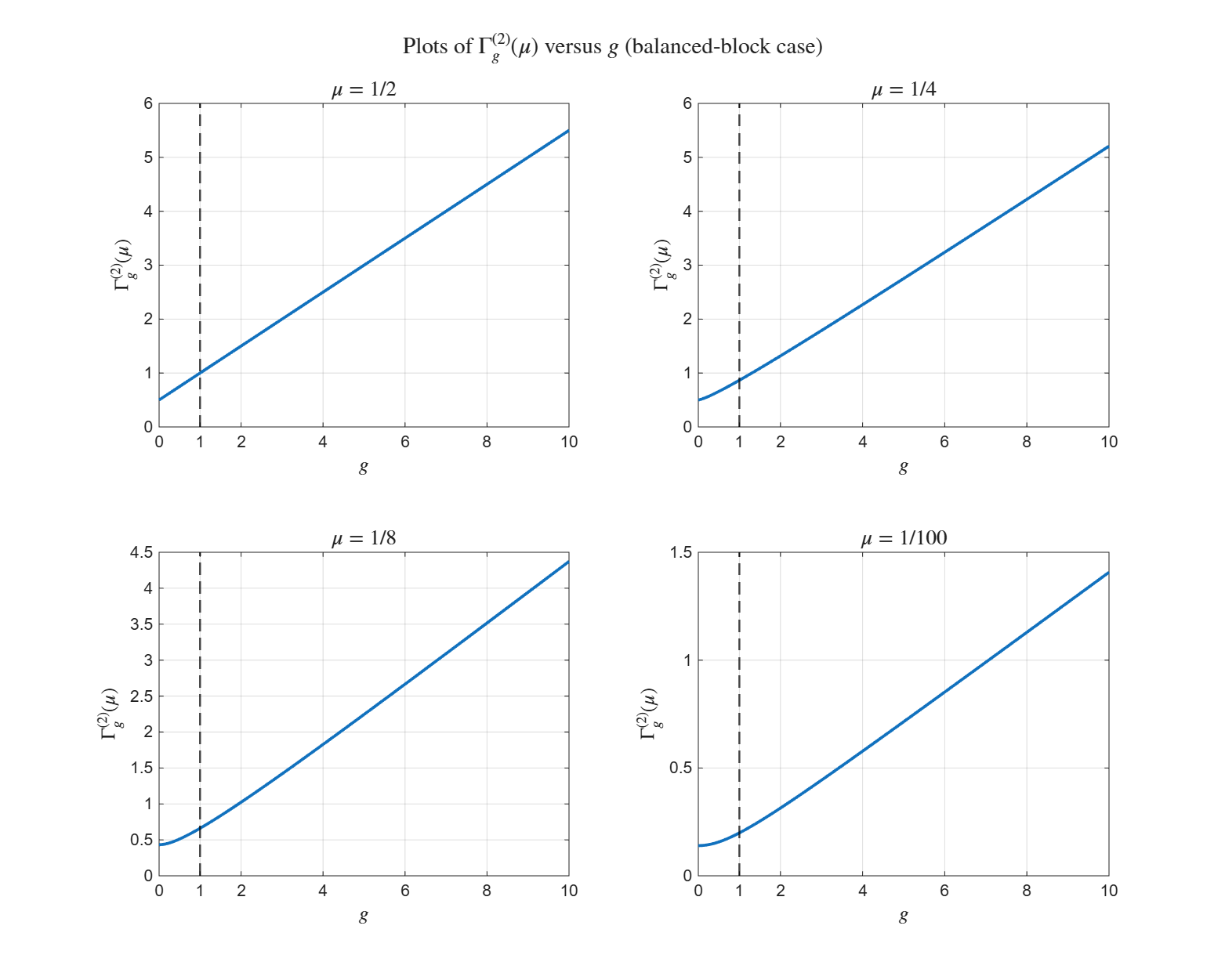}
    \caption{
    Given $\mu \in \set{ \frac{1}{2}, \frac{1}{4}, \frac{1}{8}, \frac{1}{100} }$,
    the  functional $\Gamma_g^{(2)}(\mu)$ is shown as a function of the weight $g$ in the balanced subset case ($m_1 = m_2 = m/2$) over the finite field $\mathbb{F}_2$. Here, the black dashed line indicates \(g=1\). }
    \label{fig:pf_two_point_gamma_vs_g}
\end{figure}

\begin{figure}[h!]
    \centering
    \includegraphics[width=1\textwidth]{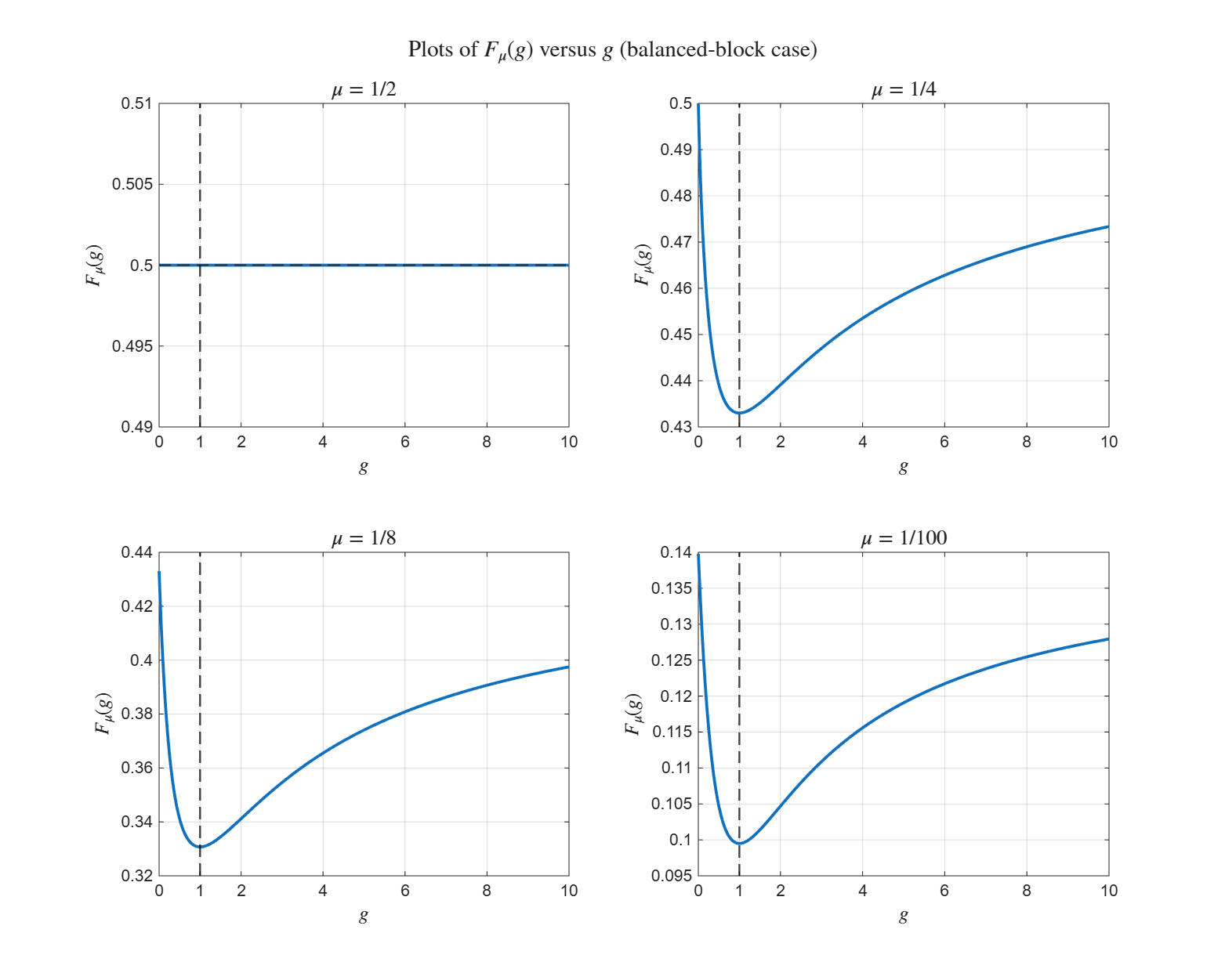}
    \caption{  Given $\mu \in \set{ \frac{1}{2}, \frac{1}{4}, \frac{1}{8}, \frac{1}{100} }$, the normalized quantity $F_\mu(g)$ is shown as a function of the weight $g$ for the balanced subset case ($m_1 = m_2 = m/2$) over the finite field $\mathbb{F}_2$.  Here, the black dashed line indicates \(g=1\).
    }
    \label{fig:pf_two_point_gamma_norm_vs_g}
\end{figure}

We now consider the special case of the finite field $\mathbb{F}_2$ ($p=2, r=1$) with balanced subset sizes $m_1 = m_2 = m/2$. Under these conditions, the coefficient $\kappa$ vanishes, and the asymptotic expectation value simplifies to:
\begin{align}\label{eq:asym_f2}
\lim_{m \to \infty} \frac{\langle s_{g} \rangle}{m} = \Gamma_g^{(2)}(\mu) = \sup_{\substack{0\le \alpha_1,\alpha_2\le 1\\ (\alpha_1+\alpha_2)/2\le \mu}}
\left(
\sqrt{\alpha_1(1-\alpha_1)}
+
g\sqrt{\alpha_2(1-\alpha_2)}
\right).
\end{align}

Figure~\ref{fig:pf_two_point_gamma_vs_g} illustrates $\Gamma_g^{(2)}(\mu)$ as a function of the weight $g$ for various degree ratios $\mu \in \set{ \frac{1}{2}, \frac{1}{4}, \frac{1}{8}, \frac{1}{100} }$. To highlight the transition behavior predicted in Proposition~\ref{prop:pf_two_point_crossover}, we plot the normalized gain functional, $F_\mu(g) := \Gamma_g^{(2)}(\mu) / (1+g)$, in Figure~\ref{fig:pf_two_point_gamma_norm_vs_g}. The numerical results reveal a clear transition point at $g=1$.

Specifically, for $0 < \mu < 1/2$, the normalized function $F_\mu(g)$ attains its global minimum at $g=1$:
\begin{equation}
F_\mu(g) \ge F_\mu(1) = \sqrt{\mu(1-\mu)} \quad \text{for all } g > 0.
\end{equation}
Indeed,
by the concavity of the function $\sqrt{x(1-x)}$,  we have 
\[ F_\mu(1) = \frac12 \sup_{\substack{0\le \alpha_1,\alpha_2\le 1\\ (\alpha_1+\alpha_2)/2\le \mu}}
\left(
\sqrt{\alpha_1(1-\alpha_1)}
+
 \sqrt{\alpha_2(1-\alpha_2)}
\right) = \sqrt{\mu(1-\mu)},\]
and
\begin{align*}
F_\mu(g) = & \frac1{1+g} \sup_{\substack{0\le \alpha_1,\alpha_2\le 1\\ (\alpha_1+\alpha_2)/2\le \mu}}
\left(
\sqrt{\alpha_1(1-\alpha_1)}
+
g\sqrt{\alpha_2(1-\alpha_2)}
\right) \\
\ge & \frac1{1+g} 
\left(
\sqrt{\mu(1-\mu)}
+
g\sqrt{\mu(1-\mu)}
\right) \\
=& F_\mu (1).
\end{align*}

\end{Examp}

\subsection{Concentration phenomenon}

We have studied the optimal expectation value achieved by multivariate DQI states in the asymptotic regime.
We now turn to the concentration behavior of the output probability distribution over $\mathbf{x}$ resulting from measurements of these states. For simplicity, we focus on the case of $\mathbb{F}^n_2$ in this section.

Let us first consider a family of states with equal weights in the decomposition of the block-symmetric polynomial state \(\ket{\mathbf{P}(\mathbf{j})}\). More precisely, let \(J_t\) and \(r_t\) be integers such that
\[
J_t=(\alpha_t+o(1))m_t
\qquad\text{for some }0<\alpha_t<1,
\]
with
\[
r_t\to\infty,\qquad r_t=o(m_t),\qquad J_t+r_t<m_t,
\]
and
\[
|\mathbf{J}|+\sum_{t=1}^N r_t+1<\frac{d^\perp}{2}.
\]
Here we write
\[
\mathbf{J}:=(J_1,\ldots,J_N),
\qquad
\mathcal{R}:=\prod_{t=1}^N\{J_t+1,J_t+2,\ldots,J_t+r_t-1\}.
\]
We then define the state \(\ket{\rho_{\mathbf{J},\mathbf{r}}}\) by
\begin{equation}\label{eq:rectangular_weighted_DQI_state}
\ket{\rho_{\mathbf J,\mathbf r}}
:=
\frac{1}{\sqrt{|\mathcal R|}}
\sum_{\mathbf j\in\mathcal R}
\ket{\mathbf P(\mathbf j)}.
\end{equation}
It is easy to see that the multivariate DQI states constructed in Theorem~\ref{thm:pf_semicircle_general}, which achieve the asymptotically optimal expectation value, are also of the form \eqref{eq:rectangular_weighted_DQI_state}.

In the following, a statement of the form
\[
|(B\mathbf x-\mathbf v)_{S_t}|=(\beta_t+o(1))m_t,
\]
means that there exists a sequence \(\varepsilon_m\to0\) such that
\[
\left|
\frac{|(B\mathbf x-\mathbf v)_{S_t}|}{m_t}-\beta_t
\right|\le \varepsilon_m,
\]
for all \(t\in[N]\). Here $|(\Bvv)_{S_t}|$ denotes the Hamming weight of the restriction of the vector \(\mathbf v\) to \(S_t\).

\begin{thm}\label{thm:concentration}
Let $N>0$ be a fixed constant. Assume that \(0<\alpha_t<\tfrac12\), and define
$\beta_t:=\frac12-\sqrt{\alpha_t(1-\alpha_t)},
\forall  t\in [N].$
Then the measurement distribution of the state \(\ket{\rho_{\mathbf J,\mathbf r}}\) defined in
\eqref{eq:rectangular_weighted_DQI_state} is asymptotically concentrated on the set
\begin{align}\label{def:supp_set}
S
:=
\left\{
\Bxx\in\mathbb F_2^n:
\left|(B\Bxx-\Bvv)_{S_t}\right|
=
(\beta_t+o(1))m_t,\ \forall\, t\in[N]
\right\}.
\end{align}
Equivalently, defining $\Pr_{\rho_{\mathbf{J},\mathbf{r}}}(\mathbf{x}) := |\langle \mathbf{x} | \rho_{\mathbf{J},\mathbf{r}} \rangle|^2$
as the probability of observing $\mathbf{x}$
 when measuring $\ket{\rho_{\mathbf{J},\mathbf{r}}}$ in the computational basis, we
have
\begin{align}
    \text{Pr}_{\rho_{\mathbf J,\mathbf r}}(\Bxx\in S)
    =1-o(1).
\end{align}

\end{thm}
 % In particular, it shows that the support of the state \(\ket{\rho_{\mathbf J,\mathbf r}}\) in the computational basis is concentrated on the set \(S\)
 % defined in \eqref{def:supp_set}.
The proof of the above theorem is presented in Appendix~\ref{Appendix:concen}.
 In addition, for any vector \(\Bxx \in S\), the value of the objective function is
\begin{align*}
F_{\mathbf{g}}(\Bxx)
&=
\sum_{t=1}^N g_t\bigl(m_t-2|(B\Bxx-\Bvv)_{S_t}|\bigr) \\
&=
\sum_{t=1}^N 2g_t m_t \sqrt{\alpha_t(1-\alpha_t)}
+o\!\left(\sum_{t=1}^N g_t m_t\right).
\end{align*}
Therefore,
\begin{align}
\frac{F_{\mathbf{g}}(\Bxx)}{m}
=
\sum_{t=1}^N 2g_t\theta_t\sqrt{\alpha_t(1-\alpha_t)}
+o\!\left(\sum_{t=1}^N g_t\theta_t\right),
\end{align}
where \(\theta_t=m_t/m\).
Hence, for constant coefficients \(\{g_t\}_t\) and suitable choices of \(\alpha_t\), the expectation value of \(\frac{F_{\mathbf{g}}(\Bxx)}{m}\) is a good approximation to the optimal expectation value in \eqref{eq:simp_F_2}.
In other words, with probability \(1-o(1)\), measuring \(\ket{\rho_{\mathbf J,\mathbf r}}\) (or, more specifically, the corresponding multivariate DQI state) yields a bitstring \(\Bxx\) such that \(\frac{F_{\mathbf g}(\Bxx)}{m}\) is close to the optimal expectation value up to lower-order terms.

Now, let us consider the size of the set $S$ on which the multivariate DQI state is concentrated.
Specifically, we have the following result.
\begin{thm}\label{thm:weighted_many_solutions}
Assume that $B:\mathbb F_2^n\to\mathbb F_2^m$ has rank $n$,  $\Bvv\in \BFF_2^m$ is a fixed vector, the number of blocks $N$ is a fixed constant
and $m_t=\theta_t m+o(m), \forall t\in [N],
\sum_{t=1}^N\theta_t=1$.
Let $C=\{B\Bxx: \Bxx\in\BFF_2^n\}\subseteq\BFF_2^m 
$ be a code in $\BFF_2^m$ and let $d^\perp=d(C^\perp)$ be the dual distance with $d^\perp=\delta^\perp m+o(m)$ for some $\delta^\perp\in(0,1]$.  
Assume also that $0<\alpha_t<\frac12,
\forall t\in [N]$,
and that there is a constant $\gamma>0$ such that
\begin{equation}\label{eq:weighted_many_solutions_gap}
\sum_{t=1}^N\theta_t\alpha_t
\le
\frac{\delta^\perp}{2}-\gamma .
\end{equation}
Then, the size of the set $S$ defined in Theorem \ref{thm:concentration}, on which the state \(|\rho_{\mathbf J,\mathbf r}\rangle\) is concentrated,
has the following lower and upper bounds: there exist constants $0<a<b<1$ such that 
\begin{align}
   2^{an}\leq |S|\leq 2^{bn}.
\end{align}
\end{thm}

The proof of the above theorem is presented in Appendix~\ref{Appendix:concen}.

\subsection{Comparison with the original (univariate) DQI framework}\label{subsec:compare-univariate}

The multivariate DQI framework developed above is strictly more general than the original (univariate) DQI framework of~\cite{jordan2024optimization}. In this subsection, we make this comparison precise: we first determine the optimal expectation value achievable by the original DQI state when applied to a weighted objective, then show that multivariate DQI strictly dominates univariate DQI for non-uniform weight profiles, and finally discuss a natural weight-aware univariate alternative whose analysis we leave open.

Let us first consider the univariate DQI on weighted objectives.
For the weighted Max-LINSAT problem $F_{\mathbf{g}}(x) = \sum_{t} g_t F_t(x)$, where $F_t = \sum_{i \in S_t} f_i(\mathbf{b}_i \cdot x)$, one may consider the expectation value of $F_{\mathbf{g}}(x)$ with respect to the original DQI state constructed from the \emph{unweighted} objective function $\sum_{i \in [m]} f_i(\mathbf{b}_i \cdot x)$, namely
\[
    \sum_{x} f\Bigl(\sum_{i} f_i(\mathbf{b}_i \cdot x)\Bigr) \ket{x},
\]
as introduced in~\cite{jordan2024optimization}. Here $f$ is a univariate polynomial of degree at most $l$. The following proposition determines the optimal performance of this state for weighted objectives; the proof is given at the end of Appendix~\ref{append:A}.

\begin{prop}\label{prop:univariate-on-weighted-main}
Given a matrix $B \in \mathbb{F}_2^{m \times n}$ and coefficients $c_1,...,c_m \in \R$. 
Let $$F_{\mathbf c}(\mathbf{x}) = \sum_{i=1}^m c_if_i (\sum_{j=1}^n B_{ij} x_j)$$ be a weighted max-XORSAT objective function. 
Consider the original (univariate) DQI state $\ket{f(F)}=\sum_{\Bxx}f(F)\ket{\Bxx}$ with $F(\mathbf{x}) )= \sum_{i=1}^m f_i (\sum_{j=1}^n B_{ij} x_j)$ being 
the unweighted objective function, let $\langle s \rangle$ be the expected value of the objective function $F_{\mathbf c}(\mathbf{x})$ for the symbol string $\Bxx$ obtained upon measuring the state $\ket{f(F)}$ in the computational basis. 
Suppose $2l +1< d^\perp$, where $d^\perp$ is the minimum distance of the code $C^\perp = \{ \mathbf{y} \in \mathbb{F}_2^m : B^\top \mathbf{y} = \mathbf{0} \}$, \textit{i.e.} the minimum Hamming weight of any nonzero codeword in $C^\perp$. 
In the limit $m \to \infty$, with $\mu = l/m\in (0,1/2)$ fixed, the optimal choice of the univariate polynomial $f$ to maximize $\langle s \rangle$ yields
\begin{equation}\label{eq:DQI_form}
\lim_{\substack{m,l\to\infty\\l/m=\mu}} \frac{\langle s \rangle}{m} = 2 (\Expect_{i} c_i ) \sqrt{\mu\left(1-\mu \right)}.
\end{equation}  
\end{prop}

Now, let us study the strict gain of multivariate over univariate DQI.
Since any univariate polynomial $f$ of degree at most $l$ can be regarded as a special case of a multivariate polynomial of total degree at most $l$ via $P(x_1, \ldots, x_N) = f\bigl(\sum_{t=1}^{N} x_t\bigr)$, the optimal performance achievable with multivariate DQI states is at least as good as, and in general strictly better than, that achievable with the original DQI states.

As an illustration, consider the special case over the finite field $\mathbb{F}_2$ (that is, $p = 2$ and $r = 1$) with balanced block sizes $m_1 = m_2 = m/2$. We compare the optimal value $\langle s_{\mathbf{g}} \rangle / m$ over all multivariate polynomials, given in~\eqref{eq:asym_f2} , with the optimal value $\langle s \rangle / m$ over all univariate polynomials, given in~\eqref{eq:DQI_form}, in the asymptotic regime $m \to \infty$. See Figure~\ref{Figure:comp_DQI} for several choices of weights $g$. The figure shows that the multivariate DQI performance strictly exceeds the univariate baseline for all $g \in \{2, 3, 5, 8\}$ and all $\mu \in (0, 1/2)$, with the gap widening as $g$ moves away from the uniform value $g = 1$.

\begin{figure}[htbp]
\centering
    \begin{minipage}{0.5\textwidth}
        \centering
        \includegraphics[width=\textwidth]{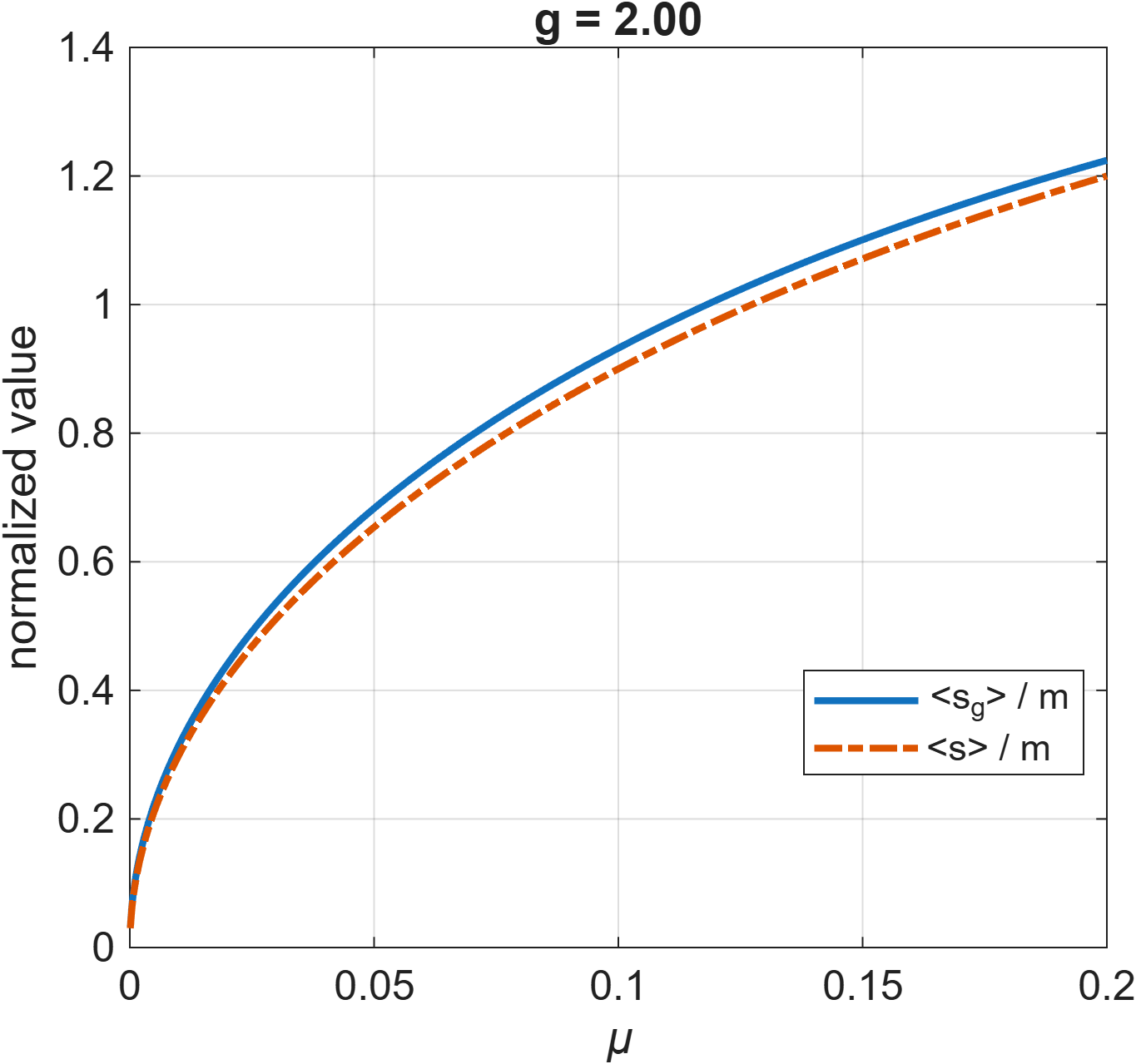}
    \end{minipage}\hfill
    \begin{minipage}{0.5\textwidth}
        \centering
        \includegraphics[width=\textwidth]{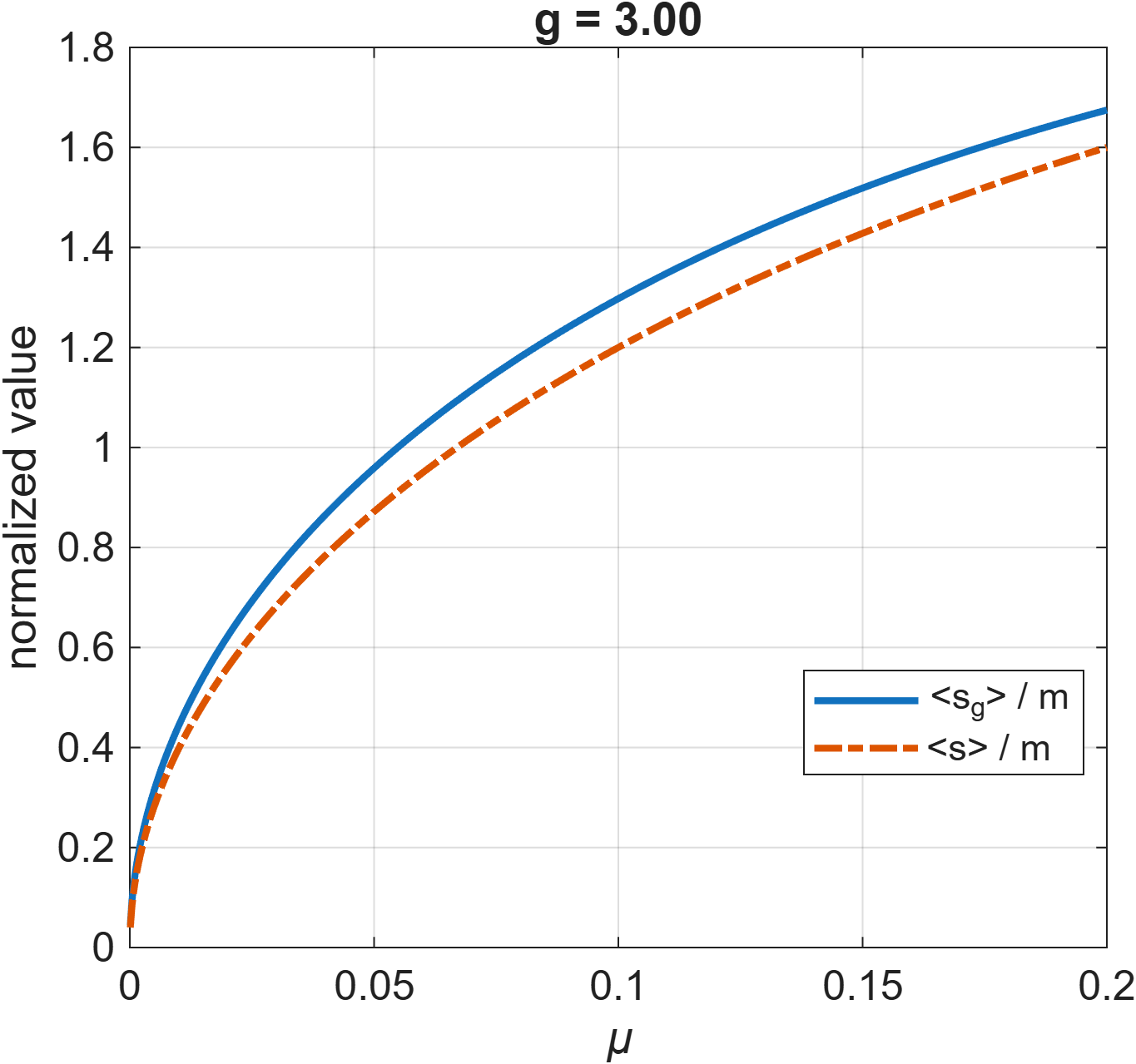}
    \end{minipage}\hfill  
        \begin{minipage}{0.5\textwidth}
        \centering
        \includegraphics[width=\textwidth]{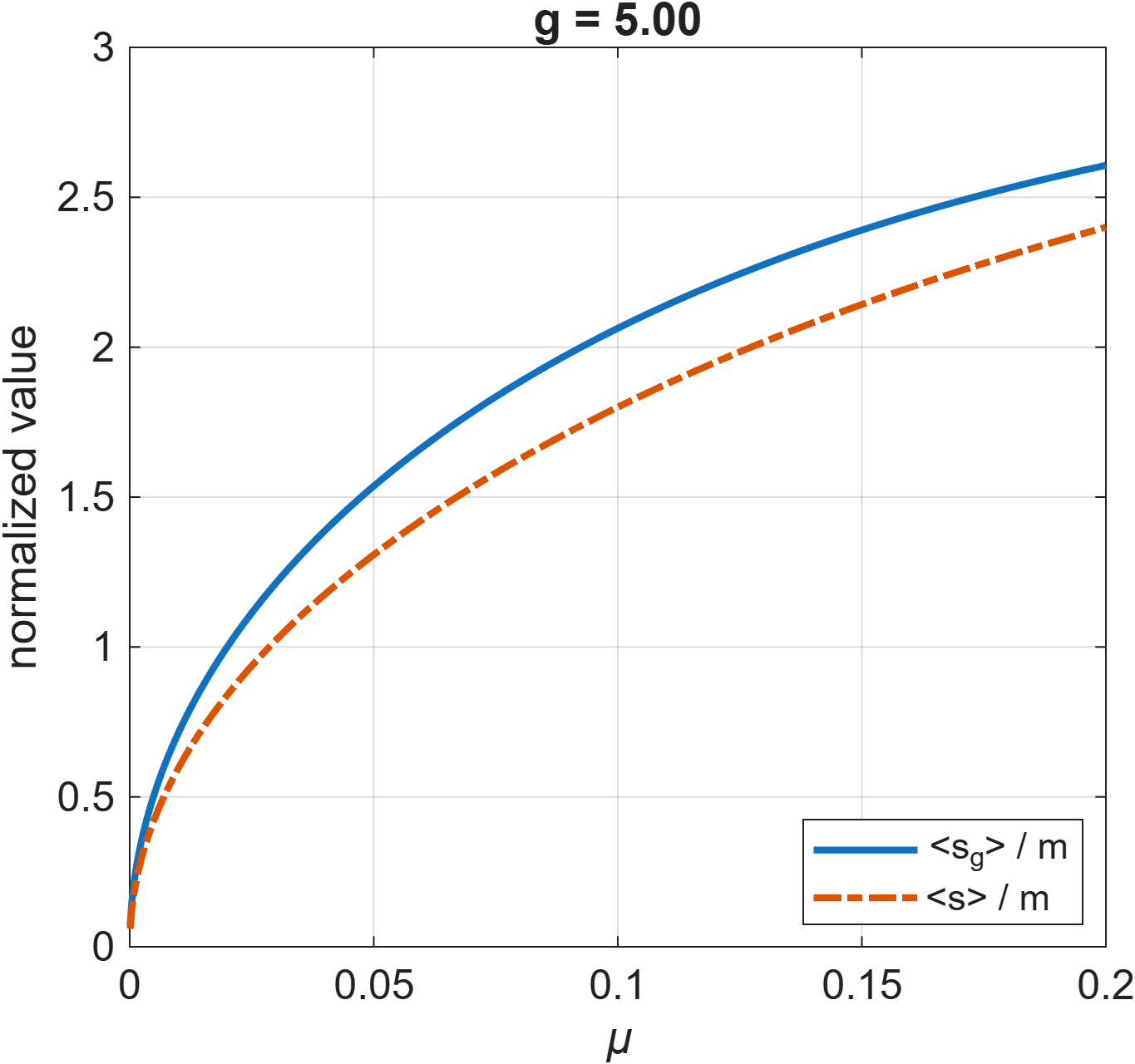}
    \end{minipage}\hfill
    \begin{minipage}{0.5\textwidth}
        \centering
        \includegraphics[width=\textwidth]{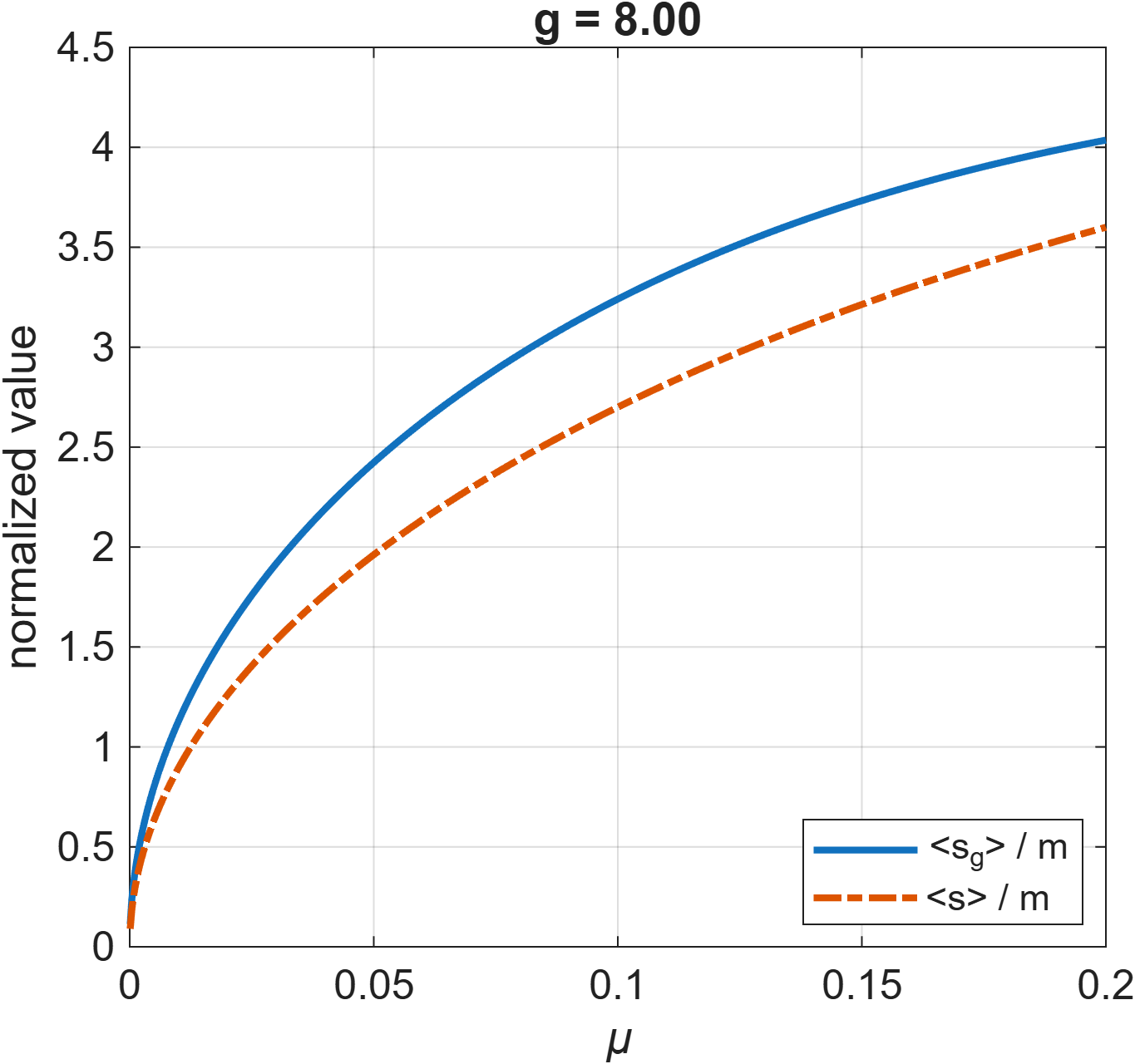}
    \end{minipage}\hfill 
\caption{  Multivariate DQI (blue) strictly outperforms univariate DQI (red dashed) for all weights  \(g \in \{2,3,5,8\}\) in the balanced two-block setting.}
    
\label{Figure:comp_DQI}
\end{figure}

Finally, let us discuss an alternative by a weight-aware univariate ansatz.
The comparison above concerns the weighted objective $F_{\mathbf{g}}$ evaluated after measurement of two different states: the multivariate DQI state, and the original DQI state constructed from the \emph{unweighted} objective. Another natural possibility is to construct a univariate DQI state directly from the \emph{weighted} objective, namely
\begin{equation}\label{eq:weight-aware-univariate}
    \sum_{x} f\Biggl(\sum_{i} c_i f_i(\mathbf{b}_i \cdot x)\Biggr) \ket{x},
\end{equation}
where $f$ is a univariate polynomial. Optimizing over states of the form~\eqref{eq:weight-aware-univariate} appears to be substantially more difficult than either of the two ansätze compared above, since the weighted objective $\sum_i c_i f_i(\mathbf{b}_i \cdot x)$ generally lacks the block-symmetric structure exploited in our multivariate framework. In particular, the reduction to a tridiagonal matrix eigenvalue problem used in the univariate case~\cite{jordan2024optimization} and the block-symmetric reduction used in our multivariate analysis (Theorem \ref{thm:pf_semicircle_general})  both rely on symmetries that are broken by the weighted sum. We leave a systematic study of this weight-aware univariate ansatz, and in particular a comparison with our multivariate framework, for future work.

\section{Preparation of Multivariate DQI States}\label{sec:gen_MDQI}
In this section, we will give the quantum process 
to generate the multivariate DQI state, which can be viewed as a weighted block generalization of the original DQI protocol of \cite{jordan2024optimization}.

\subsection{Multivariate DQI States for $\mathbb{F}_2$}
Let us start with the case where $p=2$ and $r=1$, which corresponds to the weighted XORSAT problem. In this setting, the constituent functions $f_i$ take the form $f_i(\mathbf{b}_i \cdot \mathbf{x}) = (-1)^{v_i + \mathbf{b}_i \cdot \mathbf{x}}$. 
For each block $S_t$, we let $B_t \in \mathbb{F}_2^{m_t \times n}$ denote the submatrix of $B$ restricted to the rows in $S_t$, and $v^{(t)} \in \mathbb{F}_2^{m_t}$ denote the restriction of the target vector $v$ to the same block. For any vector $y \in \mathbb{F}_2^m$, let $y^{(t)}$ be its restriction to block $S_t$. 
Then, we consider the Fourier transform of the block-symmetric polynomial state
\begin{align}
    \ket{\widetilde{\mathbf{P}(\mathbf j)}}
    :=H^{\ot n}\ket{\mathbf{P}(\mathbf j)}
    =\frac{1}{\sqrt{\prod_{t=1}^N \binom{m_t}{j_t}}}
\sum_{\substack{y^{(t)}\in \mathbb F_2^{m_t}\\ |y^{(t)}|=j_t,\ \forall t}}
(-1)^{\sum_{t=1}^N v^{(t)}\cdot y^{(t)}}
\ket{\sum_{t=1}^N B_t^\top y^{(t)}}.
\label{eq:tensor_basis_fourier}
\end{align}
Here $|y^{(t)}|$ denotes the Hamming weight of the vector $y^{(t)}\in \mathbb{F}^{m_t}_2$, for any $t\in [N]$.

The multivariate DQI state
$\ket{P(F_1, \ldots, F_N)}=\sum_{\mathbf j\in T_l}w_{\mathbf j} \ket{\mathbf{P}(\mathbf j)}$,
with $w_{\mathbf j}=\prod_{t=1}^N a^{(t)}_{j_t}$, and $(a^{(1)}_{j_1},\cdots, a^{(N)}_{j_N})$ is contained in a rectangle $I_1\times\cdots\times I_N\subset T_l$. Extending each $a_j^{(t)}$ by zero outside $I_t$,
we have
\begin{align}
   \nonumber \ket{\widetilde{P(F_1, \ldots, F_N)}}
    :=&H^{\ot n}\ket{P(F_1, \ldots, F_N)}\\
    =&
\sum_{y^{(1)},...,y^{(N)}}
\left(\prod_{t=1}^N \frac{a^{(t)}_{|y^{(t)}|}}{\sqrt{\binom{m_t}{|y^{(t)}|}}}\right)
(-1)^{\sum_{t=1}^N v^{(t)}\cdot y^{(t)}}
\ket{\sum_{t=1}^N B_t^\top y^{(t)}}.
\label{eq:tensor_rho_fourier}
\end{align}

Now, let us consider the quantum algorithm to prepare the multivariate DQI state.
The product structure of the coefficients simplifies the preparation of the state.
The preparation algorithm uses $N$ weight registers $W_1,\ldots,W_N$, $N$ error registers $E_1,\ldots,E_N$, and one syndrome register $S$ of $n$ qubits.

\textbf{Step 1.}
For each $t\in [N]$, prepare the weight register $W_t$ in the state
\[
\sum_{j\ge 0} a^{(t)}_j\ket{j}.
\]
The coefficient $a^{(t)}_j$ is defined as follows
$$a^{(t)}_j:=
\begin{cases}
r_t^{-1/2}, & J_t-r_t+1\le j\le J_t,\\
0, & \text{otherwise.}
\end{cases}$$
Here $r_t:=\lfloor \sqrt{J_t}\rfloor$, and  $J_t$ is defined  in Appendix~\ref{append:A},
each weight register has at most $\lceil \log_2(m_t+1)\rceil$ qubits.
Since the coefficients $a_j^{(t)}$ can be computed classically, the total time required for this step is $O(\sum_t  m_t)=O(m)$.

\textbf{Step 2.}
Conditioned on the value $j$ stored in $W_t$, prepare the error register $E_t$ in the Dicke state
\[
\ket{D_{m_t,j}}
=
\frac{1}{\sqrt{\binom{m_t}{j}}}
\sum_{\substack{y^{(t)}\in\mathbb F_2^{m_t}\\ |y^{(t)}|=j}}
\ket{y^{(t)}}.
\]
Then uncompute $W_t$ from the Hamming weight of $E_t$. After doing this for all blocks, the joint state of the error registers is
\begin{align}
\bigotimes_{t=1}^N
\left(
\sum_{j\ge 0} a^{(t)}_j
\frac{1}{\sqrt{\binom{m_t}{j}}}
\sum_{\substack{y^{(t)}\in\mathbb F_2^{m_t}\\ |y^{(t)}|=j}}
\ket{y^{(t)}}
\right).
\label{eq:block_dicke_tensor}
\end{align}

The superposition of Dicke states can be prepared using $O(\sum_t m^2_t)$ gates by applying the techniques from \cite{bartschi2022short, wang2024quantum}.
Note that this step could potentially be further optimized using the sparse Dicke state preparation methods introduced in \cite{khattar2025verifiable}.

\textbf{Step 3.}
For each block $t$, apply the phase operator $Z^{v^{(t)}}$ to $E_t$. This multiplies the basis vector $\ket{y^{(t)}}$ by $(-1)^{v^{(t)}\cdot y^{(t)}}$, and turns \eqref{eq:block_dicke_tensor} into
\[
\bigotimes_{t=1}^N
\left(
\sum_{j\ge 0} a^{(t)}_j
\frac{1}{\sqrt{\binom{m_t}{j}}}
\sum_{\substack{y^{(t)}\in\mathbb F_2^{m_t}\\ |y^{(t)}|=j}}
(-1)^{v^{(t)}\cdot y^{(t)}}\ket{y^{(t)}}
\right).
\]
This step requires $O(m)$ Pauli Z gates. 

\textbf{Step 4.}
Reversibly compute the syndrome
\[
\sum_{t=1}^N B_t^\top y^{(t)} = B^\top y
\]
into the register $S$. The global state becomes
\[
\sum_{y^{(1)},...,y^{(N)}}
\left(\prod_{t=1}^N \frac{a^{(t)}_{|y^{(t)}|}}{\sqrt{\binom{m_t}{|y^{(t)}|}}}\right)
(-1)^{\sum_{t=1}^N v^{(t)}\cdot y^{(t)}}
\ket{y^{(1)}}\cdots\ket{y^{(N)}}\ket{B^\top y}.
\]

\textbf{Step 5.}
Apply a decoder in \eqref{eq:decoder} for the code $C^\perp=\{z\in \mathbb F_2^m:B^\top z=0\}$ to recover $y$ from the syndrome $B^\top y$, and then uncompute all error registers. Because the support of $w$ is contained in $I_1\times\cdots\times I_N\subset T_l$, every basis state appearing above satisfies
\[
|y|=\sum_{t=1}^N |y^{(t)}|\le l.
\]
Thus, under the assumption $2l+1<d^\perp$, the error is uniquely determined by its syndrome, exactly as in the original DQI construction. After discarding the zeroed-out error registers, the state of the syndrome register is $\ket{\widetilde{P(F_1, \ldots, F_N)}}$ in \eqref{eq:tensor_rho_fourier}.

\textbf{Step 6.}
Finally, apply Hadamard gates $H^{\otimes n}$ to the syndrome register. This yields the desired state
\[
H^{\otimes n}\ket{\widetilde{P(F_1, \ldots, F_N)}}=\ket{P(F_1, \ldots, F_N)}.
\]
This step requires $O(n)$ Hadamard gates.

In summary, Steps 1–4 and Step 6 contribute a total cost of $O(m^2 + mn)$
quantum gates, all of which are efficient. The overall efficiency of the preparation procedure is therefore reduced to that of the syndrome decoder in Step 5. When the dual code $\mathcal{C}^\perp$
admits a polynomial-time decoder up to weight $\leq l$, the multivariate DQI state can be prepared in polynomial time.

\subsection{Multivariate DQI States for general $\mathbb{F}_p$}

We now extend the preparation to the general case over a prime field $\mathbb{F}_p$.
Recall that
$h_i(y)=\frac{\sqrt p}{\varphi}\bigl(f_i(y)-\bar f\bigr), $  and $
\chi_i(y)=\frac{1}{\sqrt p}h_i(y)$. Let $\widetilde\chi_i(a)
=\frac{1}{\sqrt p}
\sum_{y\in\mathbb F_p}
\omega_p^{ay}\chi_i(y)$
be its Fourier coefficient. Then  $\tilde \chi_i(0)=0$ and $\sum_{a\in\mathbb F_p^\ast}|\widetilde\chi_i(a)|^2=1$. We also consider the Fourier transform of the block-symmetric polynomial state
over $\mathbb{F}_p$, 
\begin{align}
     \ket{\widetilde{\mathbf{P}(\mathbf j)}}
    :=F^{\ot n}\ket{\mathbf{P}(\mathbf j)}
    =\frac{1}{\sqrt{\prod_{t=1}^N \binom{m_t}{j_t}}}
\sum_{\substack{y^{(t)}\in \mathbb F_p^{m_t}\\ |y^{(t)}|=j_t,\ \forall t}}
\left(
\prod_{t=1}^N \prod_{\substack{i\in S_t\\ y_i\neq 0}}\tilde \chi_i(y_i)
\right)
\ket{\sum_{t=1}^N B_t^\top y^{(t)}}.
\label{eq:pf_tensor_basis_fourier}
\end{align}
Here, $|y^{(t)}|$ denotes the Hamming weight of the vector $y^{(t)} \in \mathbb{F}_p^{m_t}$, defined as the number of nonzero entries in $y^{(t)}$.

Similarly, for the multivariate DQI state $\ket{P(F_1, \ldots, F_N)}=\sum_{\mathbf j\in T_l}w_{\mathbf j} \ket{\mathbf{P}(\mathbf j)}$,
its Fourier transform will become 
\begin{align}
\ket{\widetilde{P(F_1, \ldots, F_N)}}=
    \sum_{y^{(1)},\ldots,y^{(N)}}
\left(\prod_{t=1}^N \frac{a^{(t)}_{|y^{(t)}|}}{\sqrt{\binom{m_t}{|y^{(t)}|}}}\right)
\left(
\prod_{t=1}^N \prod_{\substack{i\in S_t\\ y_i\neq 0}}\tilde \chi_i(y_i)
\right)
\ket{\sum_{t=1}^N B_t^\top y^{(t)}}.
\label{eq:pf_tensor_rho_fourier}
\end{align}

Hence, the preparation process is similar to the $\mathbb{F}_2$ case. 
We  use $N$ weight registers $W_1,\ldots,W_N$, $N$ error registers $E_1,\ldots,E_N$, and one syndrome register $S$.
For each $t$, the register $E_t$ consists of $m_t$ subregisters of $\lceil\log_2 p\rceil$ qubits, one for each coordinate in the block $S_t$.
We encode the field element $1\in\mathbb F_p$ by the computational basis state $\ket{1}$ in each subregister.

\textbf{Step 1.}
For each $t\in [N]$, prepare the weight register $W_t$ in the state
\[
\sum_{j\ge 0} a^{(t)}_j\ket{j}.
\]
As $a_j^{(t)} = 0$ when $j>J_t$ with  $J_t$ defined  in Appendix~\ref{append:A},
each weight register has at most $\lceil \log_2(m_t+1)\rceil$ qubits, so this step is efficient once the coefficients $a_j^{(t)}$ have been classically computed.

\textbf{Step 2.}
Conditioned on $\ket{j}$ stored in $W_t$, prepare the register of the block $E_t$ in the Dicke state
\[
\ket{D_{m_t,j}}
=
\frac{1}{\sqrt{\binom{m_t}{j}}}
\sum_{\substack{\mu^{(t)}\in\{0,1\}^{m_t}\\ |\mu^{(t)}|=j}}
\ket{\mu^{(t)}}.
\]
Equivalently, the subregisters of $E_t$ are placed in a superposition of strings over $\mathbb F_p$ with entries only in $\{0,1\}$ and with Hamming weight $j$.
Then uncompute $W_t$ from the Hamming weight of the mask register.
After doing this for all blocks, the joint state of the error registers is
\begin{align}
\bigotimes_{t=1}^N
\left(
\sum_{j\ge 0} a^{(t)}_j
\frac{1}{\sqrt{\binom{m_t}{j}}}
\sum_{\substack{\mu^{(t)}\in\{0,1\}^{m_t}\\ |\mu^{(t)}|=j}}
\ket{\mu^{(t)}}
\right).
\label{eq:pf_block_dicke_tensor}
\end{align}

\textbf{Step 3.}
For each coordinate $i=1,\ldots,m$, apply to the corresponding subregister of the error register a unitary $G_i$ such that
\[
G_i\ket{0}=\ket{0},
\qquad
G_i\ket{1}=\sum_{a\in\mathbb F_p^\ast}\tilde \chi_i(a)\ket{a}.
\]
Because $\tilde \chi_i(0)=0$ and $\sum_{a\in \mathbb F_p^\ast}|\tilde \chi_i(a)|^2=1$, such a unitary exists.
Moreover, $G_i$ preserves Hamming weight.
Thus \eqref{eq:pf_block_dicke_tensor} is transformed into
\[
\sum_{y^{(1)},\ldots,y^{(N)}}
\left(\prod_{t=1}^N \frac{a^{(t)}_{|y^{(t)}|}}{\sqrt{\binom{m_t}{|y^{(t)}|}}}\right)
\left(
\prod_{t=1}^N \prod_{\substack{i\in S_t\\ y_i\neq 0}}\tilde \chi_i(y_i)
\right)
\ket{y^{(1)}}\cdots\ket{y^{(N)}}.
\]

\textbf{Step 4.}
Reversibly compute the syndrome
$\sum_{t=1}^N B_t^\top y^{(t)} = B^\top y$
into the register $S$.
The state becomes
\[
\sum_{y^{(1)},\ldots,y^{(N)}}
\left(\prod_{t=1}^N \frac{a^{(t)}_{|y^{(t)}|}}{\sqrt{\binom{m_t}{|y^{(t)}|}}}\right)
\left(
\prod_{t=1}^N \prod_{\substack{i\in S_t\\ y_i\neq 0}}\tilde \chi_i(y_i)
\right)
\ket{y^{(1)}}\cdots\ket{y^{(N)}}\ket{B^\top y}.
\]

\textbf{Step 5.}
Apply the decoder $\Dec_B^{(l)}$ in \eqref{eq:decoder}  to the registers $E_t$ and the register $S$, which maps
each $\ket{y^{(1)}}\cdots\ket{y^{(N)}}\ket{B^\top y}$ to $\ket{0}\cdots\ket{0}\ket{B^\top y}$,
then we discard the registers $E_t$,
obtaining the state $\ket{\widetilde{P(F_1, \ldots, F_N)}}$ in \eqref{eq:pf_tensor_rho_fourier}.

\textbf{Step 6.}
Finally, apply the inverse quantum Fourier transform $(F^{-1})^{\otimes n}$ on the syndrome register.
This yields the desired state
\[
(F^{-1})^{\otimes n}\ket{\widetilde{P(F_1, \ldots, F_N)}}=\ket{P(F_1, \ldots, F_N)}.
\]

\section{Performance Without Minimum Distance Assumption}\label{sec:no_min_dist}

The above results assume that decoding is performed perfectly under the condition $2l + 1 < d^\perp$. We now examine the regime where this condition is not satisfied.
We focus on $p=2$ for clarity; the extension to arbitrary $p$ is straightforward.

Without the minimum distance assumption, the decoding process is inherently imperfect. Consequently, we must account for potential decoding errors in our analysis.

For each $\mathbf j=(j_1,\ldots,j_N)\in T_l$, define
\[
\CEE_{\mathbf j}
:=
\left\{\mathbf y=(y^{(1)},\ldots,y^{(N)})\in \mathbb F_2^m:
|y^{(t)}|_H=j_t \text{ for every }t\right\}.
\]
Then
\[
|\CEE_{\mathbf j}|=\prod_{t=1}^N \binom{m_t}{j_t}.
\]

When the reversible decoder is no longer exact, we model its performance by partitioning the error space $\mathbb{F}_2^m$ into two disjoint subsets:
\[
\BFF_2^m = \CDD \cup \CFF.
\]
Here, $\mathcal{D}$ contains the error vectors $\mathbf{y}$ that are correctly decoded by their syndrome $B^\top \mathbf{y}$, while $\mathcal{F}$ contains the errors that lead to decoding failures. By defining the intersections $\mathcal{D}_{\mathbf{j}} := \mathcal{D} \cap \mathcal{E}_{\mathbf{j}}$ and $\mathcal{F}_{\mathbf{j}} := \mathcal{F} \cap \mathcal{E}_{\mathbf{j}}$, we obtain the decomposition $\mathcal{E}_{\mathbf{j}} = \mathcal{D}_{\mathbf{j}} \cup \mathcal{F}_{\mathbf{j}}$ for each $\mathbf{j} \in T_l$. Finally, we define the decoder failure rate on the block-weight layer $\mathcal{E}_{\mathbf{j}}$ as
\[
\gamma_{\mathbf j}:=\frac{|\CFF_{\mathbf j}|}{|\CEE_{\mathbf j}|},
\qquad
\gamma_{\max}:=\max_{\mathbf j\in T_l}\gamma_{\mathbf j}.
\]

Given a coefficient vector $\mathbf w=(w_{\mathbf j})_{\mathbf j\in T_l}$, the state obtained after the decoding step and postselection on the error register  is
\begin{equation}\label{eq:f2_nd_imperfect_tilde}
\ket{\widetilde\rho^{(\CDD)}}
:=
\sum_{\mathbf j\in T_l}
\frac{w_{\mathbf j}}{\sqrt{|\CEE_{\mathbf j}|}}
\sum_{\mathbf y\in \CDD_{\mathbf j}}
(-1)^{\mathbf v\cdot \mathbf y}
\ket{B^\top \mathbf y}.
\end{equation}
Because two distinct correctly decoded errors cannot have the same syndrome, we have
\begin{equation}\label{eq:f2_nd_norm}
\braket{\widetilde\rho^{(\CDD)}|\widetilde\rho^{(\CDD)}}
=
\sum_{\mathbf j\in T_l}|w_{\mathbf j}|^2(1-\gamma_{\mathbf j}).
\end{equation}
Hence, the output state after the action of Hadamard gates is 
\begin{equation}\label{eq:f2_nd_imperfect_state}
\ket{\rho^{(\CDD)}}:=H^{\otimes n}\ket{\widetilde\rho^{(\CDD)}}.
\end{equation}
Denote by $\langle s_{\mathbf g}^{(\CDD)}\rangle$ the expected value  obtained by measuring
the normalized state $\frac{\ket{\rho^{(\CDD)}}}{\|\ket{\rho^{(\CDD)}}\|}$ in the computational basis.

For $i\in S_t$ and $\mathbf j+\mathbf e_t\in T_l$, define
\[
\CEE_{\mathbf j}^{(0,i)}:=\{\mathbf y\in \CEE_{\mathbf j}: y_i=0\},
\qquad
\CEE_{\mathbf j+\mathbf e_t}^{(1,i)}:=\{\mathbf y\in \CEE_{\mathbf j+\mathbf e_t}: y_i=1\},
\]
\[
\CFF_{\mathbf j}^{(0,i)}:=\CFF_{\mathbf j}\cap \CEE_{\mathbf j}^{(0,i)},
\qquad
\CFF_{\mathbf j+\mathbf e_t}^{(1,i)}:=\CFF_{\mathbf j+\mathbf e_t}\cap \CEE_{\mathbf j+\mathbf e_t}^{(1,i)}.
\]
We then define the blockwise effective failure rates
\begin{equation}\label{eq:f2_nd_gamma0}
\widetilde\gamma_{\mathbf j,t}^{(0)}
:=
\frac{\sum_{i\in S_t}|\CFF_{\mathbf j}^{(0,i)}|}{(m_t-j_t)|\CEE_{\mathbf j}|},
\qquad
\widetilde\gamma_{\mathbf j+\mathbf e_t,t}^{(1)}
:=
\frac{\sum_{i\in S_t}|\CFF_{\mathbf j+\mathbf e_t}^{(1,i)}|}{(j_t+1)|\CEE_{\mathbf j+\mathbf e_t}|},
\end{equation}
for every admissible pair $(\mathbf j,t)$ with $\mathbf j+\mathbf e_t\in T_l$, and set
\begin{equation}\label{eq:f2_nd_gamma_max_1}
\widetilde\gamma_{\max}
:=
\max_{\substack{\mathbf j\in T_l,\ 1\le t\le N\\ \mathbf j+\mathbf e_t\in T_l}}
\frac{\widetilde\gamma_{\mathbf j,t}^{(0)}+\widetilde\gamma_{\mathbf j+\mathbf e_t,t}^{(1)}}{2}.
\end{equation}

\begin{thm}\label{thm:f2_nd_average}
Let $\mathbf w=(w_{\mathbf j})_{\mathbf j\in T_l}$ be a coefficient vector with nonnegative entries, and let $\ket{\rho^{(\CDD)}}$ be the corresponding 
unnormalized  multivariate DQI state using an imperfect decoder.
Let $\langle s_{\mathbf g}^{(\CDD)} \rangle$ be the expected value obtained upon measuring the normalized state $\frac{\ket{\rho^{(\CDD)}}}{\|\ket{\rho^{(\CDD)}}\|}$  in the computational basis.
Assume $\gamma_{\max}<1$.
If $v_1, \ldots, v_m$ are chosen independently and uniformly at random from $\mathbb{F}_2$, then
\[
\Expect_{v_1,\ldots,v_m}\langle s_{\mathbf g}^{(\CDD)}\rangle
\ge
\frac{\mathbf w^\dagger A^{(\mathbf g,l)}\mathbf w}{\|\mathbf w\|_2^2}
-
2\,\frac{\widetilde\gamma_{\max}}{1-\gamma_{\max}}
\sum_{t=1}^N g_t(m_t+1).
\]
Here 
$ \widetilde\gamma_{\max}$ is defined in \eqref{eq:f2_nd_gamma_max_1}, and 
$A^{(\mathbf{g},l)}$ is a $|T_l|\times |T_l|$ matrix whose columns and rows are labeled by $T_l$,
and the $(\mathbf{j},\mathbf{k})$-entry of $A^{(\mathbf{g},l)}$ is
\begin{align}
A^{(\mathbf{g},l)}(\mathbf{j},\mathbf{k}) = \left\{
\begin{aligned}
&g_t\sqrt{j_t(m_t-j_t+1)}, && \text{ if } \mathbf{j}=\mathbf{k}+\mathbf{e}_t \text{ for some }t,\\
&g_t\sqrt{k_t(m_t-k_t+1)}, && \text{ if } \mathbf{j}=\mathbf{k}-\mathbf{e}_t \text{ for some }t,\\
&0, && \text{ otherwise. }
\end{aligned}
\right.
\end{align}
\end{thm}

The details of the proof are presented in Appendix \ref{append:B}. Moreover, we obtain the following corollary by choosing a specific coefficient vector
$w=(w_{\mathbf j})_{\mathbf j\in T_l}$.

\begin{cor}\label{cor:f2_nd_semicircle}
Assume that $m_t=(\theta_t+o(1))m$, $\forall t\in [N]$, and 
$l=(\mu+o(1))m$ with $0<\mu<\frac12$. 
Then there exists some vector $w=(w_{\mathbf{j}})_{\mathbf j\in T_l
}$ such that 
\[
\frac1m\Expect_{v_1,\ldots,v_m}\langle s_{\mathbf g}^{(\CDD)}\rangle
\ge
\Gamma_{\mathbf g,\boldsymbol\theta}(\mu)
-
2\,\frac{\widetilde\gamma_{\max}}{1-\gamma_{\max}}
\left(\sum_{t=1}^N g_t\theta_t\right)
-
o(1).
\] 
Here  \(\Gamma_{\mathbf g,\boldsymbol\theta}(\mu)\) denotes the specialization of
\eqref{eq:im_form} to the case \(p=2\), where \(\kappa=0\).

\end{cor}

\section{Application to Weighted OPI Problems}\label{sec:opi}
We now turn to the weighted optimal polynomial intersection (weighted OPI) problem, which serves as a natural test case for multivariate DQI for three reasons. First, assigning weights to OPI constraints captures heterogeneous reliability or priority across evaluation points. This setting may arise whenever constraints carry unequal informational value, such as in polynomial reconstruction with non-uniform measurement noise. Second, the Reed–Solomon structure of the dual code admits efficient decoding via the Berlekamp–Massey algorithm, satisfying the key prerequisite for an efficient DQI implementation. Third, the unweighted version of this problem is the canonical setting in which the original  DQI's asymptotic advantage over the strongest known classical benchmark (Prange's algorithm) has been established~\cite{jordan2024optimization}, making it the natural starting point for asking whether this advantage survives in the weighted regime.

Given a prime $p$, an integer $n<p-1$,
and a subset $F_y\subseteq \mathbb F_p$ for each nonzero element
$y\in \mathbb F_p^*$.
The weighted OPI problem seeks a polynomial 
$Q(t) = \sum_{j=0}^{n-1} q_j t^j \in \mathbb{F}_p[t]$
of degree at most $n-1$ that maximizes the objective function:
\[
S_{\mathrm{wOPI}}(Q)
:=
\sum_{y\in \mathbb F_p^*} c_y\,\mathbf 1_{\{Q(y)\in F_y\}},
\]
where $c_y > 0$ are weight coefficients.
That is, $m=p-1$ in this OPI problem.
For simplicity, 
we consider the special two-weight case in which half of the $c_y$'s are equal to $1$ and the other half are equal to a fixed constant $g$.
Specifically, let $\gamma$ be a primitive element of $\mathbb{F}_p^*$, and index the nonzero field elements as $\gamma^0, \gamma^1, \dots, \gamma^{p-2}$. We define a balanced OPI problem with two weight blocks of equal size: for a given polynomial $Q(y)$, the weighted sum of satisfied constraints is given by
$$S_g(Q) := \sum_{i \in S_1} \mathbf{1}_{\{Q(\gamma^i) \in F_{\gamma^i}\}} + g \sum_{i \in S_2} \mathbf{1}_{\{Q(\gamma^i) \in F_{\gamma^i}\}},$$
where $g>0$ and $S_1, S_2\subseteq \{0,...,p-2\}$ with $|S_1|=|S_2| = \frac{p-1}{2}$.
To evaluate performance, we analyze the normalized weighted satisfaction ratio:
$$R_g(Q) := \frac{S_g(Q)}{\frac{m}{2}(1 + g)}.$$

As in the unweighted OPI case, the dual code is a Reed–Solomon code with distance $d^\perp = n+1$. Consequently, by employing multivariate DQI with a Berlekamp–Massey decoder, we may set the decoding radius to
$l = \left\lfloor \frac{n-1}{2} \right\rfloor.$

Let $x := n/p \in (0,1)$ denote the code rate. 
Since the problem instance is balanced (i.e., $|F_y| = \lfloor p/2 \rfloor$ for all $y$), the affine term in the expectation vanishes. 
Theorem \ref{thm:pf_semicircle_general} then yields the normalized satisfaction ratio
\begin{equation}
R_g^{\mathrm{DQI}}(x) = \frac{1}{2} + \frac{1}{1+g} \Gamma_g(x),
\label{eq:weighted_opi_dqi_ratio}
\end{equation}
where
\begin{equation}
\Gamma_g(x):=
\sup_{\substack{0\le \alpha_1,\alpha_2\le 1\\ \alpha_1+\alpha_2\le x}}
\left(
\sqrt{\alpha_1(1-\alpha_1)}
+
g\sqrt{\alpha_2(1-\alpha_2)}
\right).\label{eq:weighted_opi_gamma}
\end{equation}

In the original DQI paper \cite{jordan2024optimization}, the authors compared  the original DQI with several classical algorithms and heuristics, including local-search methods (such as simulated annealing and greedy optimization), Prange's algorithm, and—specifically for OPI—algebraic approaches like Reed–Solomon list-recovery and lattice-based heuristics. For the balanced OPI regime (where $|F_y| \approx p/2$ for all $y$), they concluded that the strongest known polynomial-time classical benchmark is Prange's algorithm. In the unweighted case, this algorithm achieves the asymptotic satisfaction ratio
$\frac{\langle s_{\mathrm{sat}}\rangle_{\mathrm{Prange}}}{p} = \frac{1}{2} + \frac{n}{2p}.$

Following the approach of~\cite{jordan2024optimization}, we compare multivariate DQI with a natural weighted analogue of Prange's algorithm in this section. This comparison is consistent with the unweighted setting, since multivariate DQI reduces to the original DQI framework when \(N=1\).

Regarding Prange’s algorithm for the weighted OPI problem for $g\geq 1$, the optimal strategy is to prioritize the $n$ exactly satisfied constraints on the heavier block $S_2$. This yields the following satisfaction ratio:

\begin{equation}
R_g^{\mathrm{Pr}}(x)
=
\begin{cases}
\displaystyle
\frac12+\frac{g}{1+g}\,x,
& 0\le x\le \frac12,\\[1.2ex]
\displaystyle
\frac{g+x}{1+g},
& \frac12\le x\le 1.
\end{cases}
\label{eq:weighted_opi_prange_ratio}
\end{equation}

\begin{prop}[Comparison with the weighted Prange's benchmark]
\label{prop:weighted_opi_dqi_beats_prange}
For any fixed $g \ge 1$ and every $x \in (0,1)$,
$$R_g^{\mathrm{DQI}}(x) > R_g^{\mathrm{Pr}}(x).$$
Hence, in this weighted OPI model, multivariate DQI maintains a strict asymptotic advantage over Prange's algorithm across the entire non-trivial regime $0 < n/p < 1$.
\end{prop}
\begin{proof}
This result follows directly from comparing the performance of the Multivariate DQI in \eqref{eq:weighted_opi_dqi_ratio} with that of Prange’s algorithm for the weighted OPI problem in \eqref{eq:weighted_opi_prange_ratio}.

 Specifically, if $0<x<1/2$, we choose $(\alpha_1,\alpha_2)=(0,x)$.
Then we have $\Gamma_g(x)\ge g\sqrt{x(1-x)}>gx $ for $0<x<1/2$.
Hence
\[
R_g^{\mathrm{DQI}}(x)
>
\frac12+\frac{gx}{1+g}
=
R_g^{\mathrm{Pr}}(x).
\]

If $1/2<x<1$, we choose $(\alpha_1,\alpha_2)= \left(x-\frac12,\frac12\right)$.
Then $\Gamma_g(x)\geq \sqrt{(x-\frac{1}{2})(\frac{3}{2}-x)}+\frac{g}{2}>(x-\frac{1}{2})+\frac{g}{2}$ for $1/2<x<1$.
Hence, 
\[
R_g^{\mathrm{DQI}}(x)
>\frac12+\frac{ x-\frac12 +\frac{g}{2}}{1+g} = \frac{g+x}{1+g} = R_g^{\mathrm{Pr}}(x).
\]

If $x=1/2$, we choose $(\alpha_1,\alpha_2)=(\varepsilon,1/2-\varepsilon)$ with $\varepsilon>0$ small.
Then $\Gamma_g\left(\frac12 \right) \ge \sqrt{\varepsilon(1-\varepsilon)}+g\sqrt{(\frac{1}{2}-\varepsilon)(\frac{1}{2}+\varepsilon)} > \frac g2$,
hence
\[ R_g^{\mathrm{DQI}}\left(\frac12 \right) 
\ge \frac12 + \frac{\Gamma_g\left(\frac12 \right)}{1+g}
> \frac12+\frac{g/2}{1+g} = R_g^{\mathrm{Pr}}\left(\frac12 \right).\]
Hence for every $x \in (0,1)$ we have
$R_g^{\mathrm{DQI}}(x) > R_g^{\mathrm{Pr}}(x)$.

\end{proof}

\begin{figure}[t]
    \centering
    \includegraphics[width=0.95\textwidth]{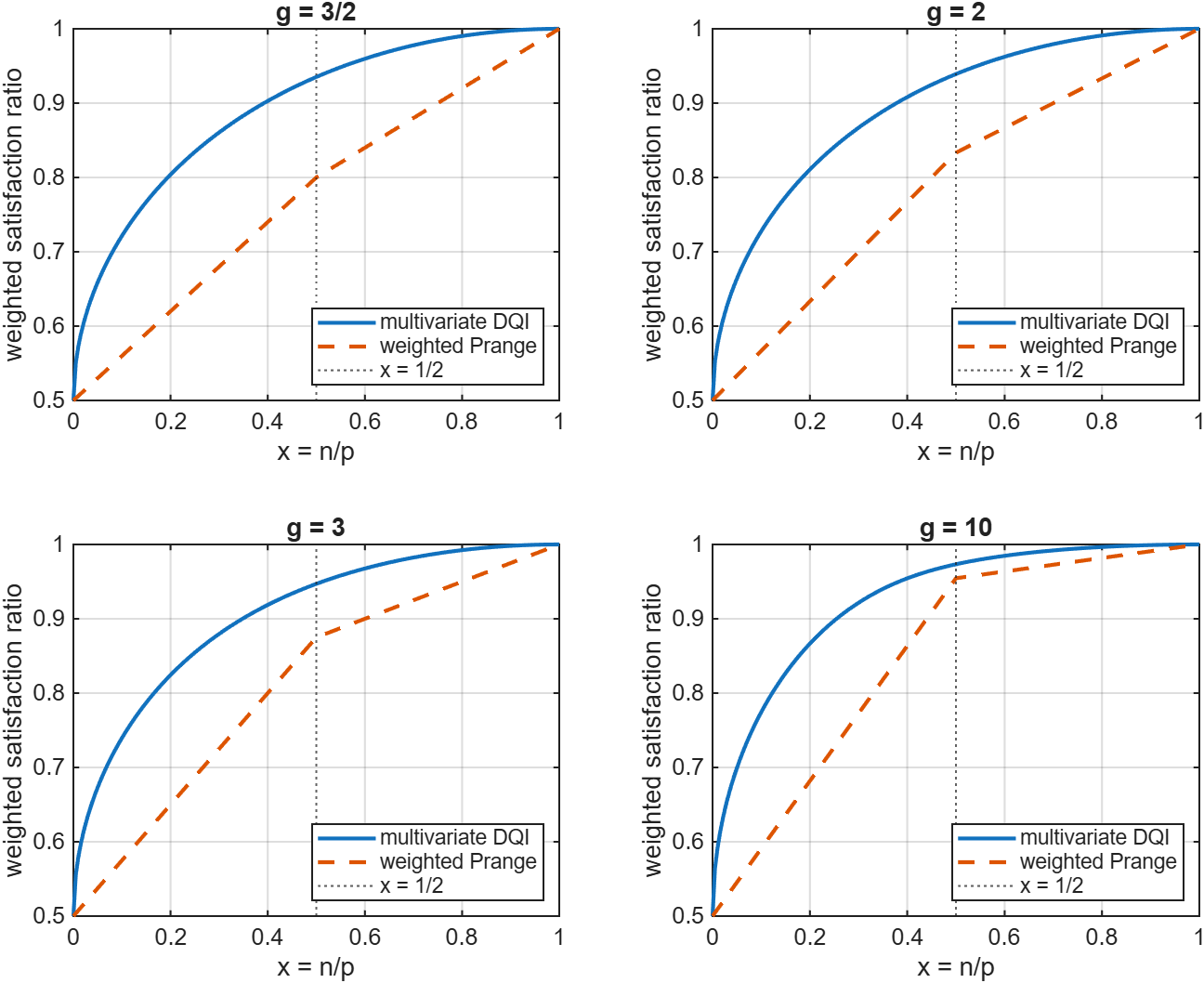}
    \caption{Weighted OPI with two balanced-size blocks of weights $1$ and $g$ for $g\in \{3/2, 2, 3, 10\}$. 
    Each panel plots the asymptotic weighted satisfaction ratio versus $x=n/p$. 
    The blue solid curve is the multivariate DQI prediction, and the dashed  orange  curve is the weighted 
    Prange's benchmark. The vertical dotted line indicates $x=1/2$.}
    \label{fig:weighted_opi_dqi_vs_prange}
\end{figure}

 Figure \ref{fig:weighted_opi_dqi_vs_prange} illustrates $R_g^{\mathrm{DQI}}(x)$ and $R_g^{\mathrm{Pr}}(x)$ for several fixed values of $g \in \{1.5, 2, 3, 10\}$, plotting the asymptotic weighted satisfaction ratio as a function of $x$. 
The solid curves represent the DQI prediction,
$$R_g^{\mathrm{DQI}}(x) = \frac{1}{2} + \frac{1}{1+g} \Gamma_g(x),$$
where $\Gamma_g(x)$ is computed by numerically optimizing  \eqref{eq:weighted_opi_gamma}. 
The dashed curves correspond to the weighted Prange's benchmark from \eqref{eq:weighted_opi_prange_ratio}. 
The kink in the Prange curves at $x = 1/2$ reflects the algorithm's greedy strategy: it first allocates its budget of $n$ exactly satisfied constraints to the heavier block (weighted by $g$); only after this block is exhausted does it begin to allocate the remaining budget to the lighter block.
Hence, the Prange's baseline strengthens as $g$ increases. Nevertheless, the figure demonstrates that in all four cases examined, the multivariate DQI performance remains strictly above the corresponding Prange's benchmark throughout the entire plotted range.

\begin{Rem}[The regime $0<g\le 1$.]
Although the discussion above was written in the case $g\ge 1$,
the case $0<g\le 1$ is completely symmetric.
Indeed, the two blocks $S_1$ and $S_2$ have the same size $m/2$, and
interchanging their roles replaces the pair of weights $(1,g)$ by $(1/g,1)$,
or equivalently replaces $g$ by $1/g$.
After normalization by the total weight $\frac{m}{2}(1+g)$, this symmetry leaves
the weighted satisfaction ratio unchanged.
Thus the regime $0<g\le 1$ does not introduce any genuinely new phenomenon;
it is the mirror image of the regime $g\ge 1$.

More concretely, for the multivariate DQI prediction we still have
\[
R_g^{\mathrm{DQI}}(x)
=
\frac12+\frac{1}{1+g}\,\Gamma_g(x),
\]
with
$\Gamma_g(x):=
\sup_{\substack{0\le \alpha_1,\alpha_2\le 1\\ \alpha_1+\alpha_2\le x}}
\left(
\sqrt{\alpha_1(1-\alpha_1)}
+
g\sqrt{\alpha_2(1-\alpha_2)}
\right).$
The symmetry is encoded in the identity $\Gamma_g(x)=g\,\Gamma_{1/g}(x), g>0$.
Because the constraint set
$
\{(\alpha_1,\alpha_2):0\le \alpha_1,\alpha_2\le 1,\ \alpha_1+\alpha_2\le x\}
$
is symmetric under exchanging $\alpha_1$ and $\alpha_2$.
Consequently,
\begin{equation}
R_g^{\mathrm{DQI}}(x)
=
R_{1/g}^{\mathrm{DQI}}(x),
\qquad g>0.
\label{eq:weighted_opi_dqi_symmetry}
\end{equation}

For the weighted analogue of Prange's algorithm, when $0<g\le 1$ the heavier
block is now $S_1$, not $S_2$.
Hence the optimal strategy is to use the $n$ exactly-satisfied constraints first
on $S_1$.
This gives
\begin{equation}
R_g^{\mathrm{Pr}}(x)
=
\begin{cases}
\displaystyle
\frac12+\frac{1}{1+g}\,x,
& 0\le x\le \frac12,\\[1.2ex]
\displaystyle
\frac{1+gx}{1+g},
& \frac12\le x\le 1.
\end{cases}
\label{eq:weighted_opi_prange_ratio_smallg}
\end{equation}
Equivalently, this is exactly the $g\mapsto 1/g$ image of
\eqref{eq:weighted_opi_prange_ratio}, since
\begin{equation}
R_g^{\mathrm{Pr}}(x)
=
R_{1/g}^{\mathrm{Pr}}(x),
\qquad g>0.
\label{eq:weighted_opi_prange_symmetry}
\end{equation}

Therefore, the comparison between multivariate DQI and Prange's algorithm for $0<g\le 1$ follows
immediately from the case $g\ge 1$.
Specifically, if $0<g\le 1$, then $1/g\ge 1$.
Applying Proposition \ref{prop:weighted_opi_dqi_beats_prange} to $1/g$ gives $R_{1/g}^{\mathrm{DQI}}(x) > R_{1/g}^{\mathrm{Pr}}(x)$ for $x \in (0, 1)$. Invoking the symmetries in Eqs. \eqref{eq:weighted_opi_dqi_symmetry} and \eqref{eq:weighted_opi_prange_symmetry}, we conclude that
\[
R_g^{\mathrm{DQI}}(x)>R_g^{\mathrm{Pr}}(x),
\qquad 0<x<1,
\]
for every $0<g\le 1$ as well.

Figure~\ref{fig:weighted_opi_small_g} illustrates this comparison for values $g \in \{1/2, 1/4, 1/8, 1/10\}$. As before, the horizontal axis denotes $x = n/p$ and the vertical axis denotes the normalized weighted satisfaction ratio.
\begin{figure}[h]
    \centering
    \includegraphics[width=0.95\textwidth]{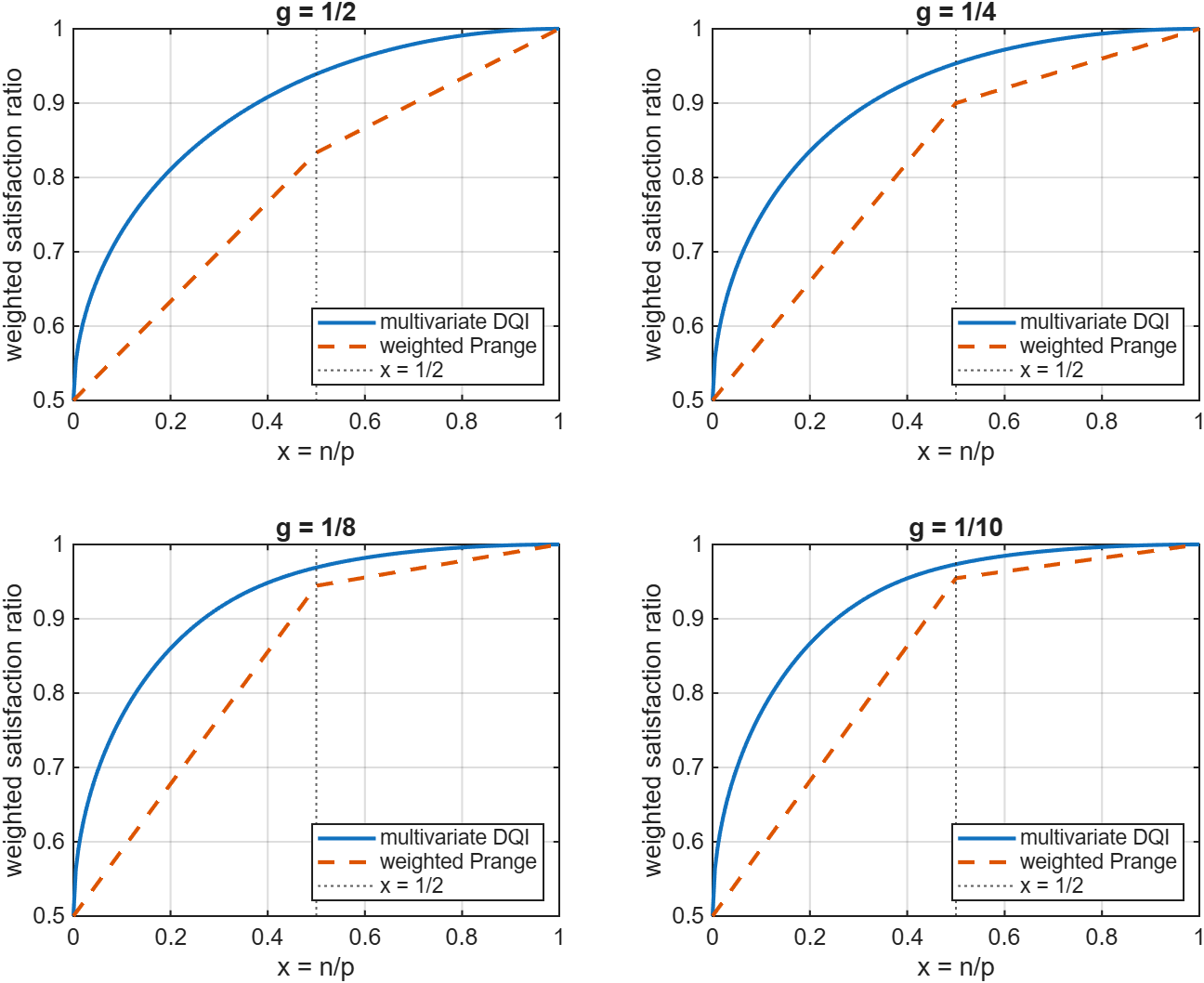}
    \caption{
    Weighted OPI with two balanced-size blocks of weights $1$ and $g$ for $g\in \{1/2,1/4,1/8,1/10\}$. 
    Each panel plots the asymptotic weighted satisfaction ratio versus $x=n/p$. 
    The blue solid curve is the multivariate DQI prediction, and the dashed  orange  curve is the weighted 
    Prange's benchmark. The vertical dotted line indicates $x=1/2$.
    }
    \label{fig:weighted_opi_small_g}
\end{figure}

\end{Rem}

\section{Block-Structured Hamiltonian DQI}\label{sec:HDQI}

The multivariate DQI framework developed above can also be viewed as a block-polynomial method for state preparation. 
We now illustrate this perspective in the setting of Hamiltonian DQI. 
While we concern weighted classical optimization problems in Sections~\ref{sec:perform_MDQI}--\ref{sec:opi}, the same block structure appears for commuting Pauli Hamiltonians.
Recall that the goal of Hamiltonian DQI is 
to generate a density matrix $\rho_{P}(H)=P(H)^2/\Tr{P(H)^2}$ given a commuting $n$-qubit Hamiltonian $H$ and 
a polynomial $P$ with real coefficients~\cite{schmidhuber2025hamiltonian,bu2026hamiltonian}.

First, let us recall some basic knowledge of Pauli operators.
For a single qubit, we consider the computational basis
$\set{\ket{0}, \ket{1}}$. The Pauli $X$ and $Z$ operators are defined by
\begin{equation}
X=\left[
\begin{array}{cc}
0&1\\
1&0
\end{array}
\right],\quad
Z=\left[
\begin{array}{cc}
1&0\\
0&-1
\end{array}
\right].
\end{equation}
Any single-qubit Pauli operator is given by $i^{-\alpha\beta}Z^{\alpha}X^{\beta}$
with $\alpha,\beta\in \mathbb{F}_2$. Hence, any $n$-qubit Pauli operator $Q$
can be written as $Q= \ot^n_{i=1} i^{-\alpha_i\beta_i}Z^{\alpha_i}_iX^{\beta_i}_i.$
Here,
the vector $(\bm \alpha, \bm\beta)\in \BFF_2^{2n}$  is called its symplectic representation, denoted by $\mathrm{symp}(Q)$.

Let $S_1,\dots,S_N$ be a
 partition of $\{1,\dots,m\}$ with  $m_t:=|S_t|$.
Let us consider the following Pauli Hamiltonian
\begin{align}
    H_{\mathbf g}=\sum^N_{t=1}g_t\sum_{i\in S_t}P_i,
\end{align}
where $g_1,\dots,g_N\in \R$, and the Pauli operators $P_i$ are distinct and mutually commuting.
For convenience, we relabel the Pauli operators $P_i$ for $i \in S_t$ as $P_{t,a}$ for $a \in [m_t]$. Defining the block operators $H_t := \sum_{a \in [m_t]} P_{t,a}$, the Hamiltonian $H_{\mathbf{g}}$ can be rewritten as
\begin{align}
H_{\mathbf{g}} = \sum_{t=1}^N g_t H_t.
\end{align}

\begin{lem}
   Given a commuting Hamiltonian $H_{\mathbf{g}} = \sum_{t=1}^N g_t H_t$ with $H_t := \sum_{a \in [m_t]} P_{t,a}$, and  a degree-$l$ polynomial $P(x)=\sum_{j=0}^l a_j x^j$, then $P(H_{\mathbf{g}})$ can be rewritten as 
   \begin{align}
       P(H_{\mathbf{g}})
       =\sum_{\Byy^{(1)}\in \BFF_2^{m_1},\cdots,\Byy^{(N)}\in \BFF_2^{m_N}}
r_{\Byy^{(1)},\dots,\Byy^{(N)}}
P_{\Byy^{(1)},\dots,\Byy^{(N)}},
   \end{align}
   where $P_{\mathbf{y}^{(1)}, \ldots, \mathbf{y}^{(N)}} = \prod_{t=1}^N \prod_{a=1}^{m_t} P_{t,a}^{(\mathbf{y}^{(t)})_a}$ is a Pauli operator, $(\mathbf{y}^{(t)})_a$ denotes the $a$-th entry of the vector $\mathbf{y}^{(t)}$, and the coefficients $r_{\mathbf{y}^{(1)}, \dots, \mathbf{y}^{(N)}}$ are given by
\begin{align}
\label{eq:def_r}
    r_{\Byy^{(1)},\ldots,\Byy^{(N)}}
:=
\sum_{j=0}^l a_j\, j!
\sum_{\substack{
\bmu^{(t)}\in \Z_{\ge 0}^{m_t}\ \forall t,\\
\sum_{t=1}^N |\bmu^{(t)}|=j,\\
\bmu^{(t)}\equiv \Byy^{(t)}\ (\mathrm{mod}\ 2)\ \forall t
}}
\frac{
\prod_{t=1}^N g_t^{|\bmu^{(t)}|}
}{
\prod_{t=1}^N \bmu^{(t)}!
},
\end{align}
with $\bmu^{(t)}!:=\prod_{a=1}^{m_t} (\bmu^{(t)})_a!$ and $|\bmu^{(t)}|:=\sum_{a=1}^{m_t} (\bmu^{(t)})_a$.

\end{lem}
\begin{proof}
Since the Pauli operators $P_{t,a}$ all commute, expanding 
$(\sum_t g_t H_t)^j$ by the multinomial theorem and then 
expanding each $H^{k_t}_t=(\sum_a P_{t,a})^{k_t}$ gives
    \begin{align*}
P(H_{\mathbf g})
&=
\sum_{j=0}^l a_j \left(\sum_{t=1}^N g_t H_t\right)^j \\
&=
\sum_{j=0}^l a_j
\sum_{\substack{
\bmu^{(1)},\dots,\bmu^{(N)}\\
\bmu^{(t)}\in \Z_{\ge 0}^{m_t},\ 
\sum_{t=1}^N |\bmu^{(t)}|=j
}}
\frac{
j!\,\prod_{t=1}^N g_t^{|\bmu^{(t)}|}
}{
\prod_{t=1}^N \bmu^{(t)}!
}
\prod_{t=1}^N \prod_{a=1}^{m_t} P_{t,a}^{(\bmu^{(t)})_a} \\
&=
\sum_{\Byy^{(1)}\in \BFF_2^{m_1}}
\cdots
\sum_{\Byy^{(N)}\in \BFF_2^{m_N}}
r_{\Byy^{(1)},\dots,\Byy^{(N)}}
P_{\Byy^{(1)},\dots,\Byy^{(N)}}.
\end{align*}
\end{proof}

Note that the coefficient $r_{\mathbf{y}^{(1)}, \dots, \mathbf{y}^{(N)}}$ depends only on the block Hamming weights $j_t := |\mathbf{y}^{(t)}|$ for $t \in [N]$. Specifically, if $(\mathbf{x}^{(1)}, \dots, \mathbf{x}^{(N)})$ and $(\mathbf{y}^{(1)}, \dots, \mathbf{y}^{(N)})$ satisfy $|\mathbf{x}^{(t)}| = |\mathbf{y}^{(t)}|$ for all $t$, then $r_{\mathbf{x}^{(1)}, \dots, \mathbf{x}^{(N)}} = r_{\mathbf{y}^{(1)}, \dots, \mathbf{y}^{(N)}}$. Hence, we relabel these coefficients using the vector of block Hamming weights $\mathbf{j} = (j_1, \dots, j_N)$, denoted simply as $r_{\mathbf{j}}$.
Moreover, we denote $r_{\mathbf{j}}=0$ for $\sum_{t=1}^N j_t>l$.

We now define the following reference state:
\begin{align}
\ket{R^l(H_{\mathbf{g}})} 
= \sum_{\mathbf{j} \in T_l} \gamma_{\mathbf{j}} \ket{D_{m_1, j_1}} \otimes \dots \otimes \ket{D_{m_N, j_N}},
\end{align}
where $\ket{D_{m_t, j_t}}$ denotes the $m_t$-qubit Dicke state with weight $j_t$. 
The coefficients are defined as
$$\gamma_{\mathbf{j}} := \frac{\left( \prod_{t=1}^N \binom{m_t}{j_t} \right)^{1/2} r_{\mathbf{j}}}{\mathcal{N}},$$
where  \(\mathcal N\) is the normalization constant
$$\mathcal{N}^2 = \sum_{\mathbf{j} \in T_l} \left( \prod_{t=1}^N \binom{m_t}{j_t} \right) |r_{\mathbf{j}}|^2.$$
By the definition, the reference state can be rewritten as
\begin{align}
 \ket{R^l(H_{\mathbf g})}:=
\frac{1}{\CNN}
\sum_{\Byy^{(1)}\in \BFF_2^{m_1}}
\cdots
\sum_{\Byy^{(N)}\in \BFF_2^{m_N}}
r_{\Byy^{(1)},\dots,\Byy^{(N)}}
\ket{\Byy^{(1)}}\ot \cdots\ot \ket{\Byy^{(N)}}.
\label{eq:block_pauli_reference_state}
\end{align}

Furthermore, we consider the symplectic representation of the Hamiltonian $H_{\mathbf{g}}$, defined by the binary matrix $B^{\top}_{H_{\mathbf{g}}} \in \mathbb{F}_2^{2n \times m}$ as follows:
\begin{align}\label{eq:symH}
        B^{\top}_{H_{\mathbf g}}  = \begin{bmatrix}
        | & | &    | &      & | \\
        \mathrm{symp}(P_{1,1}) & \cdots& \mathrm{symp}(P_{1,m_1}) & \cdots & \mathrm{symp}(P_{N,m_N}) \\
        | & | &     | &     & | 
        \end{bmatrix}, 
    \end{align}
    where $\mathrm{symp}(P_{t,a})$ is the symplectic representation  of the Pauli operator $P_{t,a}$.

Here, we consider the classical linear code with $B^{\top}_{H_{\mathbf g}} $ as the parity-check matrix of a classical linear code. 
For a given integer $l \geq 0$, a weight-$l$ decoder $\mathcal{D}_{H_{\mathbf{g}}}^{(l)}$ for $H_{\mathbf{g}}$ is defined as a unitary operator that satisfies
\begin{align}\label{eq:oracle}
        \CDD_{H_{\mathbf g}} ^{ (l)} \ket{\mathbf{y}} \ket{B^\top_{H_{\mathbf g}}  \mathbf{y}} =  \ket{\bm 0} \ket{B^{\top}_{H_{\mathbf g}}  \mathbf{y}},
    \end{align} 
  for any $\mathbf{y} \in \mathbb{F}_2^m$ with Hamming weight $ |\mathbf{y}| \leq l$.
    
\begin{thm}\label{thm:block_pauli_poly_state}
Fix $N>0$, let $H_{\mathbf g}
=
\sum_{t=1}^N g_t \sum_{a=1}^{m_t} P_{t,a}$ be a commuting Hamiltonian as defined above. Given access to a weight-$l$ decoder for $H_{\mathbf{g}}$ as defined in \eqref{eq:oracle}, and a polynomial $P(x)$ of degree $l$, there exists a quantum algorithm that prepares the state
\begin{align}
    \rho_{P}(H_{\mathbf g})
=
\frac{P^2(H_{\mathbf g})}{\Tr{P^2(H_{\mathbf g})}},
\end{align}
using a single call to the decoder. In addition,
the classical preprocessing step takes time $\mathrm{poly}(m,N,\prod_{t=1}^N m_t)$,
and the subsequent quantum circuit has size $\mathrm{poly}(m,n,|T_l|)=\mathrm{poly}(m,n,\min\set{m,l}^N)$.
In particular, when $N$ is a fixed constant,
the classical preprocessing step takes time $\mathrm{poly}(m)$,
and the subsequent quantum circuit has size $
\mathrm{poly}(m,n)$.
\end{thm}

We now describe the quantum algorithm used to prepare the polynomial state $\rho_{P}(H_{\mathbf{g}})$. The detailed complexity analysis
of Theorem \ref{thm:block_pauli_poly_state} is provided in Appendix~\ref{appen:C}. The procedure is structured as follows:
    
    \textbf{Step 1: Reference State Preparation.}
    First, we prepare $N$ weight registers jointly in the state:
    \begin{align*}
      \sum_{\mathbf{j} \in T_l} \gamma_{\mathbf{j}} \ket{j_1} \ot \cdots \ot\ket{j_N}.  
    \end{align*}
    Conditioned on the values $(j_1, \dots, j_N)$, we prepare the error registers in the corresponding Dicke states:
    \begin{align*}
        \ket{D_{m_t, j_t}} = \frac{1}{\sqrt{\binom{m_t}{j_t}}} \sum_{\substack{\mathbf{y}^{(t)} \in \mathbb{F}_2^{m_t} \\ |\mathbf{y}^{(t)}| = j_t}} \ket{\mathbf{y}^{(t)}}, \quad \forall t \in [N].
    \end{align*}
    By uncomputing the weight registers via a Hamming weight calculation on the error registers, we obtain the reference state in register $A$:
    \begin{align*}
        \ket{R^l(H_{\mathbf{g}})}_A = \sum_{\mathbf{j} \in T_l} \gamma_{\mathbf{j}} \ket{D_{m_1, j_1}} \ot \cdots \ot  \ket{D_{m_N, j_N}}.
    \end{align*}

  \textbf{Step 2: Control-Pauli Operation.}
  Prepare a maximally entangled state $\ket{\Phi_n}_{BC} = \frac{1}{\sqrt{2^n}} \sum_{\mathbf{x} \in \mathbb{F}_2^n} \ket{\mathbf{x}}\ot \ket{\mathbf{x}}$ on $n$ pairs of qubits in registers $B$ and $C$. The initial joint state is:
  \begin{align*}
      \ket{\psi_1} = \ket{R^l(H_{\mathbf{g}})}_A \otimes \ket{\Phi_n}_{BC}.
  \end{align*}
  We then apply a controlled-Pauli unitary that maps:
  \begin{align*}
      \ket{\mathbf{y}^{(1)}, \cdots, \mathbf{y}^{(N)}} \otimes \ket{\Phi_n} \longmapsto \ket{\mathbf{y}^{(1)}, \cdots , \mathbf{y}^{(N)}} \otimes \bigl(P_{\mathbf{y}^{(1)}, \dots, \mathbf{y}^{(N)}} \otimes I\bigr) \ket{\Phi_n}.
  \end{align*}
  The state evolves to:
  \begin{align*}
      \ket{\psi_2} = \frac{1}{\mathcal{N}} \sum_{\mathbf{y}^{(1)}, \dots, \mathbf{y}^{(N)}} r_{\mathbf{y}^{(1)}, \dots, \mathbf{y}^{(N)}} \ket{\mathbf{y}^{(1)},\cdots, \mathbf{y}^{(N)}} \otimes \bigl(P_{\mathbf{y}^{(1)}, \dots, \mathbf{y}^{(N)}} \otimes I\bigr) \ket{\Phi_n}.
  \end{align*}

\textbf{Step 3: Coherent Bell Measurement and Decoding. }
Apply a coherent Bell measurement on $BC$, which maps $(P \otimes I) \ket{\Phi_n} \mapsto \ket{\mathrm{symp}(P)}$ for any $n$-qubit Pauli operator $P$. The state becomes:
\begin{align*}
    \ket{\psi_3} = \frac{1}{\mathcal{N}} \sum_{\mathbf{y}^{(1)}, \dots, \mathbf{y}^{(N)}} r_{\mathbf{y}^{(1)}, \dots, \mathbf{y}^{(N)}} \ket{\mathbf{y}^{(1)}, \dots, \mathbf{y}^{(N)}} \otimes \ket{\mathrm{symp}(P_{\mathbf{y}^{(1)}, \dots, \mathbf{y}^{(N)}})}.
\end{align*}
Since $\mathrm{symp}(P_{\mathbf{y}^{(1)}, \dots, \mathbf{y}^{(N)}}) = B^\top_{H_{\mathbf g}} (\mathbf{y}^{(1)}, \dots, \mathbf{y}^{(N)})$ and the support of $\ket{R^l(H_{\mathbf{g}})}$ is restricted to total Hamming weights at most $l$, the weight-$l$ decoder uncomputes register $A$ in a single call. Discarding the zeroed-out register $A$ yields:$$\ket{\psi_4} = \frac{1}{\mathcal{N}} \sum_{\mathbf{y}^{(1)}, \dots, \mathbf{y}^{(N)}} r_{\mathbf{y}^{(1)}, \dots, \mathbf{y}^{(N)}} \ket{\mathrm{symp}(P_{\mathbf{y}^{(1)}, \dots, \mathbf{y}^{(N)}})}.$$

  \textbf{Step 4: Partial Tracing.} Undoing the coherent Bell measurement results in:
  \begin{align*}
      \ket{\psi_5} = \frac{1}{\mathcal{N}} \sum_{\mathbf{y}^{(1)}, \dots, \mathbf{y}^{(N)}} r_{\mathbf{y}^{(1)}, \dots, \mathbf{y}^{(N)}} \bigl(P_{\mathbf{y}^{(1)}, \dots, \mathbf{y}^{(N)}} \otimes I\bigr) \ket{\Phi_n} = \frac{1}{\mathcal{N}} \bigl(P(H_{\mathbf{g}}) \otimes I\bigr) \ket{\Phi_n}.
  \end{align*}
  Tracing out register $C$ yields the target  state:
\[
\rho_{P}(H_{\mathbf g})
=
\frac{P^2(H_{\mathbf g})}{\Tr{P^2(H_{\mathbf g})}}.
\]

\begin{cor}
Under the same conditions as Theorem \ref{thm:block_pauli_poly_state} with inverse temperature $\beta$, 
there exists a polynomial $P$ with degree  $l\leq 1.12\beta \norm{H_{\mathbf g}}+0.648\ln \frac{2}{\delta}$ such that $\rho_{P}(H_{\mathbf g})$ is close to the
Gibbs state as follows,
\begin{align}
   \norm{ \rho_{P}(H_{\mathbf g})-\frac{\exp(-\beta H_{\mathbf g})}{\Tr{\exp(-\beta H_{\mathbf g})}}}_1\leq \delta.
\end{align}
Moreover, if $N$ is a constant, then the total running time 
is $\text{poly}(m,n, \beta \norm{H_{\mathbf g}}, \ln \frac{2}{\delta})$.
\end{cor}

\begin{proof}
The result follows from~\cite{schmidhuber2025hamiltonian}, which shows that there exists a polynomial $P(x)$ of degree
$l \le 1.12\,\beta \|H_{\mathbf g}\| + 0.648 \ln \frac{2}{\delta},$
for which the state $\rho_{P}(H_{\mathbf g})$ is $\delta$-close in trace distance to the Gibbs state
$\exp(-\beta H_{\mathbf g})/\Tr{\exp(-\beta H_{\mathbf g})}$.
Hence, when $N$ is constant, Theorem~\ref{thm:block_pauli_poly_state} yields a total running time of
\[
\text{poly}\!\left(m,n,\beta\|H_{\mathbf g}\|,\ln \frac{2}{\delta}\right).
\]
\end{proof}

\begin{Rem}
Theorem~\ref{thm:block_pauli_poly_state} is the operator-valued analogue of the multivariate block decomposition developed in this paper. In the optimization setting, the block structure organizes the objective as a weighted sum of partial objectives. Here, the same idea organizes a commuting Pauli Hamiltonian according to blocks of terms with distinct coupling strengths.

This should be contrasted with the general Hamiltonian DQI construction of~\cite{bu2026hamiltonian} by the same authors, which applies to arbitrary Pauli Hamiltonians. For the structured class considered here—commuting Hamiltonians with a constant number of distinct weights—Theorem~\ref{thm:block_pauli_poly_state} yields a simpler protocol: the reference state is specified by a linear combination of product Dicke states, rather than by a general matrix-product-state (MPS) representation.

This simplification is specific to the constant-\(N\) regime. If the number of distinct weights grows with the system size, the block coefficient table may become superpolynomially large, and the general Hamiltonian DQI method of~\cite{bu2026hamiltonian} may be preferable.
\end{Rem}

\section{Conclusion}

In this work, we develop multivariate DQI for weighted optimization problems, with a particular focus on the weighted Max-LINSAT problem over a prime field. By grouping the constraints into \(N\) blocks according to their distinct weights, we introduce multivariate DQI states constructed from \(N\)-variable polynomials of bounded total degree, and derive a closed-form asymptotic expression for their optimal expectation value.
Moreover, we show that measurements of the multivariate DQI state are concentrated on a set of solutions whose objective values achieve the asymptotically optimal expectation value.
We also provide an explicit preparation circuit requiring only a single decoder call, and extend the analysis to the regime of imperfect decoding. In addition, we discuss the advantage of multivariate DQI over the weighted analogue of Prange’s algorithm for certain weighted OPI problems.
Finally, 
from the block-polynomial viewpoint, we considered Hamiltonian DQI for commuting Pauli Hamiltonians with block structure. 
In this structured setting, the same block-Hamming-weight organization used for weighted Max-LINSAT yields a simplified reference-state preparation based on product Dicke states or block coefficient tables, and leads to approximate Gibbs-state preparation through known polynomial approximations.

Several directions remain open for future investigation. First, while our analysis focuses on a fixed number of weight blocks \(N\), it would be interesting to understand the regime in which \(N\) grows with the problem size, and to determine whether multivariate DQI still admits an efficient description and preparation in that setting. Second, our asymptotic characterization raises the question of how the optimal multivariate polynomial should be chosen in finite-size instances, and whether one can derive explicit near-optimal constructions beyond the asymptotic limit. Third, although we extended the  multivariate DQI to imperfect decoding, a more refined understanding of the tradeoff between decoder failure rates, polynomial degree, and achievable approximation ratio would be valuable, especially for realistic noisy implementations. 
Fourth,  as discussed at the end of Section~\ref{subsec:compare-univariate},
a natural alternative to our multivariate framework is the weight-aware
univariate ansatz of Eq.~\eqref{eq:weight-aware-univariate}, which
constructs a univariate DQI state directly from the weighted objective.
A systematic analysis of this ansatz, including the identification of an optimal polynomial and a precise comparison with multivariate DQI, is left for future work.

Another important direction is to sharpen the comparison with classical algorithms. In particular, it would be desirable to identify broader families of weighted optimization problems for which multivariate DQI provably outperforms the best known classical approaches, and to clarify the complexity-theoretic evidence for such an advantage. It is also natural to ask whether the multivariate DQI framework can be extended beyond weighted Max-LINSAT and weighted OPI to more general constraint satisfaction and coding-theoretic optimization problems. Finally, on the Hamiltonian side, an interesting open problem is whether the present block-structured construction can be generalized beyond commuting Pauli Hamiltonians, potentially leading to new efficient schemes for preparing approximate Gibbs states of more general many-body systems.

\section{Acknowledgement}
K. Bu thanks Arthur Jaffe, Seth Lloyd, Peter Love, Feng Qian, Alexander Schmidhuber  for helpful discussions.
K. Bu is partly supported by the JobsOhio GR138220, and ARO Grant W911NF19-1-0302 and the ARO MURI Grant W911NF-20-1-0082.

\newpage
\begin{appendix}

\appendix

\section{Detailed Proof for the Performance }
\label{append:A}

Recall that $p$ is a prime, $B \in \mathbb{F}_p^{m \times n}$,
and $1\le r\le p-1$.
For each $i=1,\ldots,m$,  $L_i\subset \mathbb{F}_p$ is a subset of cardinality
\[
|L_i|=r,
\]
and  the corresponding $\{\pm1\}$-valued function $f_i$ is defined as follows
\[
f_i(y)=
\left\{
\begin{aligned}
&1, && \text{ if } y\in L_i,\\
&-1, && \text{ if } y\notin L_i.
\end{aligned}
\right.
\]
Since all $|L_i|$ are equal to $r$, all functions $f_i$ have the same mean value
\[
\bar f:=\frac1p\sum_{y\in \mathbb F_p} f_i(y)=\frac{2r}{p}-1.
\]

Let us define the normalization constant
\[
\varphi:=\left(\sum_{y\in \mathbb F_p}|f_i(y)-\bar f|^2\right)^{1/2}
=
\sqrt{4r\left(1-\frac{r}{p}\right)},
\]
and define the rescaled centered functions
\[
h_i(y):=\frac{\sqrt p}{\varphi}\bigl(f_i(y)-\bar f\bigr),\qquad i=1,\ldots,m.
\]
Then
\[
\frac1p\sum_{y\in \mathbb F_p} h_i(y)=0,
\qquad
\frac1p\sum_{y\in \mathbb F_p}|h_i(y)|^2=1,
\]
and, because $f_i(y)^2=1$, the functions $h_i$ satisfy the quadratic identity
\begin{align}
h_i(y)^2 = 1+\kappa\,h_i(y),
\qquad
\kappa:=\frac{p-2r}{\sqrt{r(p-r)}}.
\label{eq:block_hp_quadratic}
\end{align}
Equivalently,
\begin{align}
f_i(y)=\bar f+\frac{\varphi}{\sqrt p}\,h_i(y)
=
\left(\frac{2r}{p}-1\right)+\frac{2\sqrt{r(p-r)}}{p}\,h_i(y).
\label{eq:block_f_from_h}
\end{align}

Let us define $\chi_i(y)$ as follows
\[
\chi_i(y):=\frac{1}{\sqrt p}h_i(y)=\frac{1}{\varphi}\bigl(f_i(y)-\bar f\bigr),
\]
and let
\[
\tilde \chi_i(a):=\frac{1}{\sqrt p}\sum_{y\in \mathbb F_p}\omega_p^{ay}\chi_i(y),
\qquad
\omega_p=e^{2\pi i/p}.
\]
Then $\tilde \chi_i(0)=0$ and
\[
\sum_{a\in \mathbb F_p}|\tilde \chi_i(a)|^2=1.
\]
The inverse Fourier formula gives
\begin{align}
h_i(y)=\sum_{a\in \mathbb F_p^\ast}\tilde \chi_i(a)\,\omega_p^{-ay},
\label{eq:block_h_inverse_fourier}
\end{align}
where $\mathbb F_p^\ast=\mathbb F_p\setminus\{0\}$.

For the objective function 
$F_{\mathbf g}$, we consider the corresponding 
Hamiltonian 
\begin{align}
    H_{\mathbf g}
=\sum_{\Bxx\in \BFF_p^n}\sum^N_{t=1} g_t F_t(\Bxx)\proj{\Bxx}.
\end{align}
Hence, for any state $\psi$,
the expectation 
$\langle s_{\mathbf g} \rangle=\sum_{\Bxx\in \BFF_p^n}|\iinner{\Bxx}{\psi}|^2F_{\mathbf g}(\Bxx)$ is equal to 
the expected value of the
Hamiltonian $\bra{\psi}H_{\mathbf g}\ket{\psi}$.

Since 
$f_i(y)=\bar f+\frac{\varphi}{\sqrt p}\,h_i(y)
=
\left(\frac{2r}{p}-1\right)+\frac{2\sqrt{r(p-r)}}{p}\,h_i(y)$, 
the Hamiltonian can be rewritten as 
follows
\begin{align*}
     H_{\mathbf g}
     =&\bar f\left(\sum_{t=1}^N g_t m_t\right)I
     +\frac{\varphi}{\sqrt p}
    \sum_{\Bxx\in \BFF_p^n}\sum^N_{t=1}g_t\sum_{i\in S_t}
     h_i\!\left(\mathbf b_i\cdot \mathbf x\right)
     \proj{\Bxx}\\
     =&\bar f\left(\sum_{t=1}^N g_t m_t\right)I
     +\frac{\varphi}{\sqrt p}
     \sum^N_{t=1}g_tH_t,
\end{align*}
where 
\begin{align}
    H_t= \sum_{\Bxx\in \BFF_p^n}\sum_{i\in S_t}
     h_i\!\left(\mathbf b_i\cdot \mathbf x\right)
     \proj{\Bxx}.
\end{align}

For each $t\in [N]$ and each $k\geq 0$, denote 
\begin{align}
P^{(k)}_{H_t}
:=
\sum_{\mathbf{x} \in \mathbb{F}_p^n}
\left(
\sum_{\substack{s_1,\ldots,s_k\in S_t\\ s_1<\cdots<s_k}}
h_{s_1}\!\left(\mathbf{b}_{s_1} \cdot \mathbf{x} \right)\cdots
h_{s_k}\!\left(\mathbf{b}_{s_k} \cdot \mathbf{x} \right)
\right)
\ket{\mathbf{x}}\bra{\mathbf{x}}.
\end{align}
By \eqref{eq:block_hp_quadratic}, every higher-order factor $h_i^a$ can be expressed as a linear combination of $1$ and $h_i$.
Hence, the power $H_t^j$ can be written as a linear combination of
\[
P_{H_t}^{(0)},P_{H_t}^{(1)},\ldots,P_{H_t}^{(j)}.
\]

Hence, let us consider the block-symmetric polynomial state 
\begin{align}
    \ket{\mathbf{P}(\mathbf j)}:
    =\frac{1}{\sqrt{\prod^N_{t=1}\binom{m_t}{j_t}}}\frac{1}{\sqrt{p^n}}\sum_{\Bxx\in \mathbb{F}^n_p}\prod^N_{t=1}P^{(j_t)}_{H_t}\ket{\Bxx},
\end{align}
where each integer $j_t\geq 0$.

\begin{lem}
    \label{lem:pf_orthogonality}
Suppose $2l< d^\perp$,
and let $\mathbf j=(j_1,\cdots, j_N), \mathbf k=(k_1,\cdots, k_N)$ be index vectors such that $j_1+\cdots+j_N\le l$, $k_1+\cdots+k_N\le l$.
Then, we have
\begin{align}
    \iinner{\mathbf P(\mathbf j)}{\mathbf P(\mathbf k)}
    =\prod^N_{t=1}\delta_{j_t,k_t}.
\end{align}
\end{lem}
\begin{proof}
First,
for any subset $T\subseteq [m]$, let us define the monomial state
\[
\ket{h_T}
:=
\frac{1}{\sqrt{p^n}}
\sum_{\mathbf x\in \mathbb F_p^n}
\left(\prod_{i\in T} h_i(\mathbf b_i\cdot \mathbf x)\right)\ket{\mathbf x}.
\]
By \eqref{eq:block_h_inverse_fourier},  the quantum Fourier transform of the monomial state
is 
\[
F^{\otimes n}\ket{h_T}
=\sum_{\mathbf a\in (\mathbb F_p^\ast)^T}
\left(\prod_{i\in T}\tilde \chi_i(a_i)\right)
\ket{\sum_{i\in T} a_i \mathbf b_i}.
\]

If the set size satisfies $|T| < d^\perp$, then any two distinct vectors of coefficients $\mathbf{a}, \mathbf{a}' \in (\mathbb{F}_p^\ast)^T$ must result in distinct syndromes. Indeed, if they produced the same syndrome, then:
\[\sum_{i \in T} (a_i - a_i') \mathbf{b}_i = \mathbf{0},
\]
which implies that the difference vector would be a non-zero codeword in $C^\perp$ with a Hamming weight of at most $|T|$. Since $|T| < d^\perp$, this contradicts the definition of the minimum distance of $C^\perp$. Hence, the computational basis vectors appearing in the state $F^{\otimes n}\ket{h_T}$ are all distinct. It follows that:
\[
\iinner{h_T}{h_T} = \sum_{\mathbf{a} \in (\mathbb{F}_p^\ast)^T} \prod_{i \in T} |\tilde{\chi}_i(a_i)|^2 = \prod_{i \in T} \sum_{a \in \mathbb{F}_p^\ast} |\tilde{\chi}_i(a)|^2 = 1.
\]

Second, 
consider $T, U \subseteq [m]$ such that $|T| + |U| < d^\perp$ and $T \neq U$. Suppose a computational basis vector were to appear in the support of both $F^{\otimes n}\ket{h_T}$ and $F^{\otimes n}\ket{h_U}$. This would imply the existence of non-zero coefficients $a_i, b_j' \in \mathbb{F}_p^\ast$ such that:$$\sum_{i \in T} a_i \mathbf{b}_i = \sum_{j \in U} b_j' \mathbf{b}_j.$$Rearranging this equality, we obtain:$$\sum_{i \in T} a_i \mathbf{b}_i - \sum_{j \in U} b_j' \mathbf{b}_j = \mathbf{0},$$which defines a non-zero codeword in $C^\perp$ with support contained in $T \cup U$. The Hamming weight of this codeword is at most $|T \cup U| \leq |T| + |U| < d^\perp$. This contradicts the assumption that $d^\perp$ is the minimum distance of the code. Consequently, the sets of computational basis vectors appearing in the two states must be disjoint, leading to:
\[
\braket{h_U|h_T}=0\qquad\text{whenever }T\neq U,\ |T|+|U|<d^\perp.
\]

Finally, both states
$\ket{\mathbf P(\mathbf j)}$ and$\quad\ket{\mathbf P(\mathbf k)} $
can be expressed as linear combinations of the monomial states $\ket{h_T}$. By invoking the orthonormality established above, the result of the lemma follows immediately.

\end{proof}

Define
\[
T_l = \left\{(j_1,\ldots,j_N)\in \mathbb Z_{\ge 0}^N: j_t\le m_t\text{ for all }t,\text{ and } j_1+\cdots+j_N\le l\right\}.
\]
Then the collection of states
\[
\left\{\ket{\mathbf P(\mathbf j)} : \mathbf j=(j_1,\ldots,j_N)\in T_l\right\}
\]
forms an orthonormal basis for a subspace, which we denote by $V_l$. It follows that the dimension of this space is given by $\dim(V_l) = |T_l|$.

     Now, consider the operator $H_{\mathbf{g}} = \bar{f} \left( \sum_{t=1}^N g_t m_t \right) I + \frac{\varphi}{\sqrt{p}} \sum_{t=1}^N g_t H_t$ restricted to the subspace $V_l$. Since the first term is proportional to the identity, studying $H_{\mathbf{g}}$ on $V_l$ is equivalent to considering the restricted operator $\sum_{t=1}^N g_t H_t$ acting on the same subspace.
     
\begin{lem}\label{lem:pf_matrix}
Assume $2l+1<d^\perp$.
For any index vectors $\mathbf{j}=(j_1,\ldots,j_N)\in T_l$ and $\mathbf{k}=(k_1,\ldots,k_N)\in T_l$, we have
\begin{align}
&\bra{\mathbf P(\mathbf j)} \sum^N_{t=1}g_t H_t\ket{\mathbf P(\mathbf k)} \nonumber\\
= & \left\{
\begin{aligned}
&g_t\sqrt{j_t(m_t-j_t+1)}, && \text{ if } \mathbf{j}=\mathbf{k}+\mathbf{e}_t \text{ for some }t,\\
&g_t\sqrt{k_t(m_t-k_t+1)}, && \text{ if } \mathbf{j}=\mathbf{k}-\mathbf{e}_t \text{ for some }t,\\
&\kappa\sum_{t=1}^N g_t j_t, && \text{ if } \mathbf{j}=\mathbf{k},\\
&0, && \text{ otherwise, }
\end{aligned}
\right.
\label{eq:block_matrix_entries_general_p}
\end{align}
where $\mathbf{e}_t$ denotes the $t$-th standard basis vector in $\mathbb Z^N$.
\end{lem}

\begin{proof}
    Let us first consider the Hamiltonian $H_t$ for any $t\in [N]$. By definition,
\[
H_tP^{(k)}_{H_t}
=
\sum_{i\in S_t} h_i
\sum_{\substack{T\subseteq S_t\\ |T|=k}}
\prod_{j\in T} h_j.
\]

The terms are partitioned into two cases based on the index $i$.
First, if $i\notin T$, then we obtain
\[
h_i\prod_{j\in T}h_j=\prod_{j\in T\cup\{i\}}h_j.
\]
Each $(k+1)$-subset of $S_t$ occurs exactly $k+1$ times in this way, so these terms contribute
\[
(k+1)P_{H_t}^{(k+1)}.
\]
Second, if $i\in T$, then applying \eqref{eq:block_hp_quadratic} yields
\[
h_i\prod_{j\in T}h_j
=
h_i^2\prod_{j\in T\setminus\{i\}}h_j
=
\prod_{j\in T\setminus\{i\}}h_j
+
\kappa\prod_{j\in T}h_j.
\]
Summing over all pairs $(i,T)$ with $i\in T$ gives
\[
(m_t-k+1)P_{H_t}^{(k-1)} + \kappa\,k\,P_{H_t}^{(k)}.
\]

Therefore,
\begin{align}
H_tP_{H_t}^{(k)}
=
(k+1)P_{H_t}^{(k+1)} + \kappa\,k\,P_{H_t}^{(k)} + (m_t-k+1)P_{H_t}^{(k-1)}.
\label{eq:block_X_action}
\end{align}

Applying Lemma \ref{lem:pf_orthogonality}, we evaluate the coefficients for each case. For the case $\mathbf{j} = \mathbf{k} + \mathbf{e}_t$, the coefficient is given by

\[g_t \cdot \frac{(k_t+1)\sqrt{\binom{m_t}{k_t+1}}}{\sqrt{\binom{m_t}{k_t}}} = g_t\sqrt{(k_t+1)(m_t-k_t)}.
\]
By substituting $j_t = k_t + 1$, it is equivalent to
\[
g_t\sqrt{j_t(m_t-j_t+1)} \quad \text{for } \mathbf{j} = \mathbf{k} + \mathbf{e}_t.
\]
Similarly, the coefficient for $\mathbf{j} = \mathbf{k} - \mathbf{e}_t$ is
\[
g_t\sqrt{k_t(m_t-k_t+1)},
\]
while the diagonal coefficient is $\kappa\,g_t k_t$. Summing these contributions over $t = 1, \dots, N$ completes the proof of \eqref{eq:block_matrix_entries_general_p}.

\end{proof}

The multivariate DQI state $ \ket{P(F_1,F_2, \ldots, F_N)}$, where the polynomial $P$ has total degree at most $l$,  can be written as
a linear combination 
of  block-symmetric polynomial state 
$\ket{\mathbf{P}(\mathbf j)} $ with $ \mathbf j=(j_1,\ldots,j_N)\in T_l$. Specifically,
\begin{align}
    \ket{P(F_1, \ldots, F_N)}
    =\sum_{\mathbf j\in T_l}w_{\mathbf j}\ket{\mathbf P(\mathbf j)}.
\end{align}

\begin{lem}\label{lem:char_perf}
    Let $p$ be a prime, $B \in \mathbb{F}_p^{m \times n}$ be a matrix, and let   $f_i: \mathbb{F}_p \to \{\pm 1\}$ be functions  with $|f_i^{-1}(+1)|=r$ for all $i$. 
Consider a partition of the index set $[m] = S_1 \sqcup \dots \sqcup S_N$ with subset sizes $m_t := |S_t|$, and let $F_{\mathbf{g}}(\mathbf{x})$ be the weighted Max-LINSAT objective function:
$$F_{\mathbf{g}}(\mathbf{x}) = \sum_{t=1}^N g_t F_t(\mathbf{x}), \quad F_t(\mathbf{x}) = \sum_{i \in S_t} f_i(\mathbf{b}_i \cdot \mathbf{x}),\quad g_t>0, \forall t\in [N].$$
Let $P$ be an $N$-variable multivariate polynomial  with a total degree at most $l$ such that the normalized 
multivariate  DQI state$\ket{P(F_1,\ldots, F_N)}$ is decomposed as $$\ket{P(F_1,F_2, \ldots, F_N)}
    =\sum_{\mathbf j\in T_l}w_{\mathbf j}\ket{\mathbf P(\mathbf j)}.$$ 
Let $\langle s_{\mathbf{g}} \rangle$ denote the expectation value obtained upon measuring the corresponding multivariate  DQI state$\ket{P(F_1,\ldots, F_N)}$.
Suppose $2l + 1 < d^\perp$, where $d^\perp$ is the minimum distance of the dual code $C^\perp = \{ \mathbf{d} \in \mathbb{F}_p^m : B^\top \mathbf{d} = \mathbf{0} \}$, then 
\begin{align}
    \langle s_{\mathbf g} \rangle
=
\bar f\left(\sum_{t=1}^N g_t m_t\right)
+
\frac{\varphi}{\sqrt p}\, w^\dag A^{(\mathbf g,l,\kappa)} w,
\end{align}
where $\bar f=\frac{2r}{p}-1$, $\frac{\varphi}{\sqrt p}=\frac{2\sqrt{r(p-r)}}{p}$,
and $A^{(\mathbf{g},l,\kappa)}$ is a $|T_l|\times |T_l|$ matrix whose columns and rows are labeled by $T_l$,
and the $(\mathbf{j},\mathbf{k})$-entry of $A^{(\mathbf{g},l,\kappa)}$ is
\begin{align}\label{eq:form_A}
A^{(\mathbf{g},l,\kappa)}(\mathbf{j},\mathbf{k}) = \left\{
\begin{aligned}
&g_t\sqrt{j_t(m_t-j_t+1)}, && \text{ if } \mathbf{j}=\mathbf{k}+\mathbf{e}_t \text{ for some }t,\\
&g_t\sqrt{k_t(m_t-k_t+1)}, && \text{ if } \mathbf{j}=\mathbf{k}-\mathbf{e}_t \text{ for some }t,\\
&\kappa\sum_{t=1}^N g_t j_t, && \text{ if } \mathbf{j}=\mathbf{k},\\
&0, && \text{ otherwise. }
\end{aligned}
\right.
\end{align}

\end{lem}

\begin{proof}
Since  $ H_{\mathbf g}
     =\bar f\left(\sum_{t=1}^N g_t m_t\right)I
     +\frac{\varphi}{\sqrt p}
     \sum^N_{t=1}g_tH_t$,  the expectation value  $\langle s_{\mathbf{g}} \rangle$ with respect to the state $\ket{P(F_1,\ldots, F_N)}$
     is given by
     \begin{align*}
     \langle s_{\mathbf{g}} \rangle
     =&\bra{P(F_1,\ldots, F_N)}H_{\mathbf g} \ket{P(F_1,\ldots, F_N)}\\
     =&\bar f\left(\sum_{t=1}^N g_t m_t\right)+
     \frac{\varphi}{\sqrt p}
     \bra{P(F_1,\ldots, F_N)}\sum^N_{t=1} g_tH_t\ket{P(F_1,\ldots, F_N)}\\
     =&\bar f\left(\sum_{t=1}^N g_t m_t\right)+
     \frac{\varphi}{\sqrt p}
     \sum_{\mathbf{j},\mathbf{k}\in T_l}\bar{w}_{\mathbf j}w_{\mathbf k}
     \bra{\mathbf{P}(\mathbf j)}\sum^N_{t=1} g_tH_t  \ket{\mathbf{P}(\mathbf k)}\\
     =&\bar f\left(\sum_{t=1}^N g_t m_t\right)+
     \frac{\varphi}{\sqrt p}
    w^\dag A^{(\mathbf g, l, \kappa)} w,
     \end{align*}
     where the matrix $A$ is defined as \eqref{eq:form_A}.

\end{proof}

Based on the above lemma, it follows that maximizing the expectation value
 $\langle s_{\mathbf g} \rangle$ over all
multivariate  DQI states $\ket{P(F_1, F_2, \ldots, F_N)}$  is equivalent to finding the largest eigenvalue 
of the matrix $ A^{(\mathbf g,l,\kappa)}$.

\begin{lem}[Upper bound on the largest eigenvalue]
\label{lem:pf_upper}
Assume $2l+1<d^\perp$ and write $\mu=l/m$.
Then
\[
\frac{\lambda_{\max}(A^{(\mathbf{g},l, \kappa)})}{m}
\le
\Gamma_{\mathbf g,\mathbf m, \kappa}(\mu),
\]
where 
\[
\Gamma_{\mathbf g,\mathbf m, \kappa}(\mu):=
\sup_{\substack{0\le \alpha_t\le 1\ \forall t\in [N]\\ \sum_{t=1}^N \frac{m_t}{m}\alpha_t\le \mu}}
\sum_{t=1}^N \frac{m_t}{m}g_t\left(\kappa\alpha_t+2\sqrt{\alpha_t(1-\alpha_t)}\right).
\]
\end{lem}
\begin{proof}
Let $D$ be a diagonal matrix indexed by $T_l$, with diagonal entries given by
\[
D_{\mathbf{j},\mathbf{j}}=\sqrt{\prod_{t=1}^N \binom{m_t}{j_t}}, \quad \forall \mathbf{j}\in T_l.
\]

We define the matrix
\[
G := D^{-1}A^{(\mathbf{g},l, \kappa)}D.
\]
For $\mathbf j\in T_l$, the non-zero entries of $G$ are
\[
G_{\mathbf j,\mathbf j-\mathbf e_t}=g_tj_t, \qquad
G_{\mathbf j,\mathbf j+\mathbf e_t}=g_t(m_t-j_t), \qquad
G_{\mathbf j,\mathbf j}=\kappa\sum_{t=1}^N g_tj_t,
\]
provided the respective indices belong to $T_l$.

Since $A^{(\mathbf{g},l, \kappa)}$ and $G$ are similar, we have
\[
\lambda_{\max}(A^{(\mathbf{g},l, \kappa)})=\lambda_{\max}(G).
\]
Moreover, the off-diagonal entries of $G$ are nonnegative.
We can choose $c>0$ large enough such that  $G+cI$ is entrywise nonnegative.
Then by Lemma \ref{260320lem2},
for every strictly positive vector $u$,
\[
\lambda_{\max}(A^{(\mathbf{g},l, \kappa)})+c
=
\lambda_{\max}(G+cI)
\le
\max_{\mathbf j\in T_l}\frac{((G+cI)u)_{\mathbf j}}{u_{\mathbf j}},
\]
and therefore
\[
\lambda_{\max}(A^{(\mathbf{g},l,\kappa)})
\le
\max_{\mathbf j\in T_l}\frac{(Gu)_{\mathbf j}}{u_{\mathbf j}}.
\]

For $x_1,\ldots,x_N>0$, we define a strictly positive test vector
\[
u^{(x_1,\ldots,x_N)}_{\mathbf j}:=\prod_{t=1}^N x_t^{j_t},
\qquad
\mathbf j=(j_1,\ldots,j_N)\in T_l.
\]
For every $\mathbf j=(j_1,\ldots,j_N)\in T_l$, a direct calculation yields
\[
\frac{(Gu^{(x_1,\ldots,x_N)})_{\mathbf j}}{u^{(x_1,\ldots,x_N)}_{\mathbf j}}
\le
\sum_{t=1}^N g_t\left((m_t-j_t)x_t+\frac{j_t}{x_t}+\kappa\,j_t\right).
\]
Hence, the maximum eigenvalue satisfies
\[
\frac{\lambda_{\max}(A^{(\mathbf{g},l,\kappa)})}{m}
\le
\inf_{x_1,\ldots,x_N>0}
\sup_{\substack{\mathbf j\in T_l}}
\sum_{t=1}^N \frac{m_t}{m}g_t\left(\left(1-\frac{j_t}{m_t}\right)x_t+\frac{j_t/m_t}{x_t}+\kappa\,\frac{j_t}{m_t}\right).
\]
Since $\mathbf j\in T_l$ implies $0\le j_t/m_t\le 1$ and
\[
\sum_{t=1}^N \frac{m_t}{m}\cdot \frac{j_t}{m_t}
=
\frac{j_1+\cdots+j_N}{m}
\le
\frac{l}{m}
=
\mu,
\]
we obtain the upper bound
\[
\frac{\lambda_{\max}(A^{(\mathbf{g},l,\kappa)})}{m}
\le
\inf_{x_1,\ldots,x_N>0}
\sup_{\substack{0\le \alpha_t\le 1\ \forall t\in [N]\\ \sum_{t=1}^N \frac{m_t}{m}\alpha_t\le \mu}}
\sum_{t=1}^N \frac{m_t}{m}g_t\left((1-\alpha_t)x_t+\frac{\alpha_t}{x_t}+\kappa\alpha_t\right).
\]
The function inside the brackets is affine in $(\alpha_1,\ldots,\alpha_N)$ and convex in $(x_1,\ldots,x_N)$.
Moreover, the supremum diverges whenever some $x_t\to 0$ or some $x_t\to \infty$, so we may restrict $(x_1,\ldots,x_N)$ to a large compact rectangle and apply Sion's minimax theorem:
\[
\frac{\lambda_{\max}(A^{(\mathbf{g},l, \kappa)})}{m}
\le
\sup_{\substack{0\le \alpha_t\le 1\ \forall t\in [N]\\ \sum_{t=1}^N \frac{m_t}{m}\alpha_t\le \mu}}
\inf_{x_1,\ldots,x_N>0}
\sum_{t=1}^N \frac{m_t}{m}g_t\left((1-\alpha_t)x_t+\frac{\alpha_t}{x_t}+\kappa\alpha_t\right).
\]
The infimum separates blockwise across $t$, and
\[
\inf_{x_t>0}\left((1-\alpha_t)x_t+\frac{\alpha_t}{x_t}\right)
=
2\sqrt{\alpha_t(1-\alpha_t)},
\qquad\forall t\in [N].
\]
Therefore
\[
\frac{\lambda_{\max}(A^{(\mathbf{g},l,\kappa)})}{m}
\le
\Gamma_{\mathbf g,\mathbf m, \kappa}(\mu).
\]

\end{proof}

\begin{lem}[Collatz--Wielandt upper bound\cite{HornJohnson2012}]\label{260320lem2}
Let $M$ be an entry-wise nonnegative matrix, and let $u>0$ be a strictly positive vector. Then the spectral radius $\rho(M)$ of $M$ satisfies
\[
\rho(M)\le \max_i \frac{(Mu)_i}{u_i}.
\]

\end{lem}

\begin{lem}[Lower bound on the largest eigenvalue]\label{lem:pf_eigenvalue}
Assume $2l+1<d^\perp$ and let $\mu=l/m\in(0,1/2)$.
Then there exists a unit vector $w=(w_{\mathbf j})_{\mathbf j\in T_l}$ of product form
\[
w_{\mathbf j}=\prod_{t=1}^N a^{(t)}_{j_t},
\qquad \mathbf j=(j_1,\ldots,j_N)\in T_l,
\]
such that
\[
\frac{\langle w|A^{(\mathbf{g},l, \kappa)}|w\rangle}{m}
\ge
\Gamma_{\mathbf g,\mathbf m,\kappa}(\mu)
-
O_{\kappa}\!\left(
\frac{1}{m}\sum_{t=1}^N g_tm_t^{3/4}
\right).
\]
\end{lem}
\begin{proof}
 Let $\alpha_1, \dots, \alpha_N \in [0,1]$ be parameters that attain the supremum in the definition of $\Gamma_{\mathbf{g}, \mathbf{m}, \kappa}(\mu)$. For each $t \in \{1, \dots, N\}$, we define
\begin{align}\label{eq:pf_J_choice}
 J_t := \lfloor \alpha_t m_t \rfloor.
\end{align}
It follows that $J_t \le \alpha_t m_t$ for all $t$, and consequently,
$$\sum_{t=1}^N J_t \le \sum_{t=1}^N \alpha_t m_t \le \mu m = l.$$

Next, we define $a^{(t)}=(a_j^{(t)})_{j\ge 0}$ blockwise. If $J_t=0$, we set$$a^{(t)}_0=1 \quad \text{and} \quad a^{(t)}_j=0 \text{ for } j\ge 1.$$If $J_t\ge 1$, we let $r_t:=\lfloor \sqrt{J_t}\rfloor$ and define
$$a^{(t)}_j:=
\begin{cases}
r_t^{-1/2}, & J_t-r_t+1\le j\le J_t,\\
0, & \text{otherwise.}
\end{cases}$$
The vector entry $w_{\mathbf j}$ is then defined as the product
$$w_{\mathbf j}:=\prod_{t=1}^N a^{(t)}_{j_t} \quad \text{for }\quad \mathbf j=(j_1,\dots,j_N)\in T_l.$$
By construction, the support of $w$ is contained in a rectangle $I_1\times\cdots\times I_N$ with $I_t\subseteq \{0,1,\ldots,J_t\}$. Moreover, the squared norm satisfies
$$\|w\|_2^2 = \prod_{t=1}^N \left(\sum_{j\ge 0}(a_j^{(t)})^2\right) = 1,$$
confirming that $w$ is a unit vector of product form.

Now, let us compute its Rayleigh quotient.
By the definition of $A^{(\mathbf{g},l,\kappa)}$,
\begin{align*}
\langle w|A^{(\mathbf{g},l,\kappa)}|w\rangle
=
&\sum_{t=1}^N 2g_t
\sum_{\mathbf j\in T_l: \mathbf j+\mathbf e_t\in T_l}
\sqrt{(j_t+1)(m_t-j_t)}\,w_{\mathbf j}w_{\mathbf j+\mathbf e_t}  +\kappa\sum_{\mathbf j\in T_l}\left(\sum_{t=1}^N g_tj_t\right)w_{\mathbf j}^2.
\end{align*}
Because the support of $w$ is a rectangle contained in $T_l$, both sums factor, and we obtain
\begin{align*}
\langle w|A^{(\mathbf{g},l,\kappa)}|w\rangle
=&
\sum_{t=1}^N 2g_t
\left( \sum_{j\ge 0}\sqrt{(j+1)(m_t-j)}\,a^{(t)}_ja^{(t)}_{j+1}\right) \prod_{s\neq t} \left( \sum_{r \ge 0 } (a_r^{(s)})^2\right)\\
&+  
\kappa\sum_{t=1}^N g_t \left( \sum_{j\ge 0}j\,(a^{(t)}_j)^2 \right) \prod_{s\neq t} \left( \sum_{r \ge 0 } (a_r^{(s)})^2\right)\\
=&
\sum_{t=1}^N 2g_t
\sum_{j\ge 0}\sqrt{(j+1)(m_t-j)}\,a^{(t)}_ja^{(t)}_{j+1}
+
\kappa\sum_{t=1}^N g_t\sum_{j\ge 0}j\,(a^{(t)}_j)^2.
\end{align*}
For the blocks with $J_t=0$, both contributions vanish.
Since then $\alpha_t<1/m_t$, the target block contribution satisfies
\[
g_tm_t\left(\kappa\alpha_t+2\sqrt{\alpha_t(1-\alpha_t)}\right)
=O_{\kappa}\bigl(g_t m_t^{1/2}\bigr).
\]

Now assume $J_t\ge 1$.
Then
\[
\sum_{j\ge 0}\sqrt{(j+1)(m_t-j)}\,a^{(t)}_ja^{(t)}_{j+1}
=
\frac1{r_t}\sum_{j=J_t-r_t+1}^{J_t-1}\sqrt{(j+1)(m_t-j)},
\]
and
\[
\sum_{j\ge 0}j\,(a^{(t)}_j)^2
=
\frac1{r_t}\sum_{j=J_t-r_t+1}^{J_t}j.
\]
For $J_t-r_t+1\le j\le J_t-1$, we have
\[
\left|\frac{j+1}{m_t}-\alpha_t\right|\le \frac{r_t}{m_t},
\qquad
\left|\frac{j}{m_t}-\alpha_t\right|\le \frac{r_t+1}{m_t}.
\]
Hence
\begin{align*}
&\left|\frac{j+1}{m_t}\left(1-\frac{j}{m_t}\right)-\alpha_t(1-\alpha_t)\right|
\le  \left|\frac{j+1}{m_t}-\alpha_t\right|+\left|\frac{j}{m_t}-\alpha_t\right|
\le \frac{2r_t+1}{m_t}.
\end{align*}
Using $|\sqrt{u}-\sqrt{v}|\le \sqrt{|u-v|}$ for $u,v\ge 0$, we obtain
\[
\left|\sqrt{(j+1)(m_t-j)}-m_t\sqrt{\alpha_t(1-\alpha_t)}\right|
\le
\sqrt{(2r_t+1)m_t}.
\]
Therefore
\begin{align*}
&\left| \left(\frac1{r_t}\sum_{j=J_t-r_t+1}^{J_t-1}\sqrt{(j+1)(m_t-j)} \right)-m_t\sqrt{\alpha_t(1-\alpha_t)}\right| \\
\le & \frac{r_t-1}{r_t}\sqrt{(2r_t+1)m_t}
+\frac{m_t}{r_t}\sqrt{\alpha_t(1-\alpha_t)}.
\end{align*}
Since $r_t\le \sqrt{J_t}\le \sqrt{m_t}$ and
\[
\sqrt{\alpha_t(1-\alpha_t)}\le \sqrt{\alpha_t}\le \sqrt{\frac{J_t+1}{m_t}},
\]
we get
\[
\frac{m_t}{r_t}\sqrt{\alpha_t(1-\alpha_t)}=O\bigl(m_t^{1/2}\bigr),
\qquad
\sqrt{(2r_t+1)m_t}=O\bigl(m_t^{3/4}\bigr).
\]
Hence
\[
\frac1{r_t}\sum_{j=J_t-r_t+1}^{J_t-1}\sqrt{(j+1)(m_t-j)}
=
m_t\sqrt{\alpha_t(1-\alpha_t)}+O\bigl(m_t^{3/4}\bigr).
\]
Also,
\[
\frac1{r_t}\sum_{j=J_t-r_t+1}^{J_t}j
=
J_t-\frac{r_t-1}{2}
=
\alpha_t m_t+O\bigl(m_t^{1/2}\bigr).
\]
It follows that
\[
\frac{2g_t}{r_t}\sum_{j=J_t-r_t+1}^{J_t-1}\sqrt{(j+1)(m_t-j)}
+
\kappa g_t\frac1{r_t}\sum_{j=J_t-r_t+1}^{J_t}j
=
g_tm_t\left(2\sqrt{\alpha_t(1-\alpha_t)}+\kappa\alpha_t\right)
+O_{\kappa}\bigl(g_tm_t^{3/4}\bigr).
\]
Summing over $t\in [N]$, including the $J_t=0$ blocks, we obtain
\[
\langle w|A^{(\mathbf{g},l,\kappa)}|w\rangle
=
m\,\Gamma_{\mathbf g,\mathbf m,\kappa}(\mu)
+
O_{\kappa}\!\left(\sum_{t=1}^N g_t m_t^{3/4}\right).
\]
That is, 
\begin{align}
\left|\frac{\langle w|A^{(\mathbf{g},l,\kappa)}|w\rangle}{m}
-\Gamma_{\mathbf g,\mathbf m,\kappa}(\mu)|
\right|=O_{\kappa}\!\left(\frac{1}{m}\sum_{t=1}^N g_t m_t^{3/4}\right).
\end{align}
Therefore, the lower bound is proved.

\end{proof}

\begin{mproof}[Proof of Theorem~\ref{thm:pf_semicircle_general}]
    By Lemma \ref{lem:char_perf}, the  optimal  expected value  is 
    \[
\langle s_{\mathbf g} \rangle
=
\bar f\left(\sum_{t=1}^N g_t m_t\right)
+
\frac{\varphi}{\sqrt p}\,\lambda_{\max}(A^{(\mathbf g,l,\kappa)}).
\]
Then the result comes from the Lemmas \ref{lem:pf_upper} and \ref{lem:pf_eigenvalue}.
\end{mproof}

To compare with the result from the measurement on  the original DQI, we derive the 
result in Proposition~\ref{prop:univariate-on-weighted-main} based on the techniques from \cite{jordan2024optimization,bu2025decoded}. 
Here, we give a proof for completeness. 

% \begin{prop}\label{260202thm1}
% Given a matrix $B \in \mathbb{F}_2^{m \times n}$ and coefficients $c_1,...,c_m \in \R$. 
% Let $$F_{\mathbf c}(\mathbf{x}) = \sum_{i=1}^m c_if_i (\sum_{j=1}^n B_{ij} x_j)$$ be a weighted max-XORSAT objective function. 
% Consider the original (univariate) DQI state $\ket{f(F)}=\sum_{\Bxx}f(F)\ket{\Bxx}$ with $F(\mathbf{x}) )= \sum_{i=1}^m f_i (\sum_{j=1}^n B_{ij} x_j)$ being 
% the unweighted objective function, let $\langle s \rangle$ be the expected value of the objective function $F_{\mathbf c}(\mathbf{x})$ for the symbol string $\Bxx$ obtained upon measuring the state $\ket{f(F)}$ in the computational basis. 
% Suppose $2l +1< d^\perp$, where $d^\perp$ is the minimum distance of the code $C^\perp = \{ \mathbf{y} \in \mathbb{F}_2^m : B^\top \mathbf{y} = \mathbf{0} \}$, \textit{i.e.} the minimum Hamming weight of any nonzero codeword in $C^\perp$. 
% In the limit $m \to \infty$, with $\mu = l/m\in (0,1/2)$ fixed, the optimal choice of the univariate polynomial $f$ to maximize $\langle s \rangle$ yields
% \begin{equation}
% \lim_{\substack{m,l\to\infty\\l/m=\mu}} \frac{\langle s \rangle}{m} = 2 (\Expect_{i} c_i ) \sqrt{\mu\left(1-\mu \right)}.
% \end{equation}  
% \end{prop}
\begin{mproof}[Proof of Proposition~\ref{prop:univariate-on-weighted-main}]
    Following the same techniques from \cite{jordan2024optimization,bu2025decoded}, we can show that the 
    expected value of the objective function $F_{\mathbf c}(\mathbf{x})$ for the symbol string $\Bxx$ obtained upon measuring the DQI 
    state for the unweighted objective function  $F(\mathbf{x}) = \sum_{i=1}^m f_i (\sum_{j=1}^n B_{ij} x_j)$
    is 
    \begin{align}
        \Expect_{\Bxx\sim \ket{f(F)}} F_c( \Bxx)
        = \mathbf{w}^\dagger A_1 \mathbf{w}.
    \end{align}
    where matrix \[A_1 = \left(\Expect_{i=1}^m c_i\right) A_0,\]
where $A_0$ is an $(l+1)\times (l+1)$ matrix whose $(k,j)$-entry is
\begin{align}\label{260202eq2}
A_0(k,j)
= 
\left\{
\begin{aligned} 
&  \sqrt{j(m-j+1) }    && \text{ if } j=k+1,\\
&  \sqrt{k(m-k+1) }    && \text{ if } j=k-1,\\
& 0   && \text{ otherwise. }
\end{aligned}
\right.
\end{align} 

As shown  in \cite{jordan2024optimization},
if $l\leq m/2$,
the maximum eigenvalue of $A_0(k,j)$ satisfies the 
semicircle law:
\begin{align}
    \lim_{\substack{m,l\to\infty\\l/m=\mu}}\frac{\lambda_{\max}(A_0)}{m}= 2\sqrt{\mu(1-\mu)}
\end{align}
where the limit is taken as both $m$ and $l$ tend to infinity with the ratio $\mu=l/m$ fixed.
Hence, we get the result.

\end{mproof}

\section{Detailed Proof for Concentration Phenomenon}
\label{Appendix:concen}

\begin{lem}\label{lem:hamiltonian_action}
Let \(\mathbf j=(j_1,\ldots,j_N)\in T_l\) be such that both \(\mathbf j-\mathbf e_t\in T_l\) and \(\mathbf j+\mathbf e_t\in T_l\). Then
\begin{equation}\label{eq:Xt_action_on_P_basis}
H_t\ket{\mathbf P(\mathbf j)}
=
\sqrt{j_t(m_t-j_t+1)}\,\ket{\mathbf P(\mathbf j-\mathbf e_t)}
+
\sqrt{(j_t+1)(m_t-j_t)}\,\ket{\mathbf P(\mathbf j+\mathbf e_t)}.
\end{equation}
\end{lem}

\begin{proof}
By \eqref{eq:block_X_action} with \(p=2\), we have
\[
H_t P_t^{(j)}
=
(j+1)P_t^{(j+1)}+(m_t-j+1)P_t^{(j-1)}.
\]
Combining this with the definition of \(\ket{\mathbf P(\mathbf j)}\) yields \eqref{eq:Xt_action_on_P_basis}.
\end{proof}

\begin{lem}\label{claim:weighted_claim32}
%\red{(Gu:) Under the same assumptions as Lemma \ref{lem:hamiltonian_action},}
For every fixed block $t$, the family of states $\ket{\rho_{\mathbf J,\mathbf r}}$ satisfies the following property
\begin{equation}\label{eq:block_approx_eigenstate_Xt}
\left\|
\left(
H_t-2\sqrt{\alpha_t(1-\alpha_t)}\,m_t
\right)
\ket{\rho_{\mathbf J,\mathbf r}}
\right\|
=o(m_t).
\end{equation}
Hence,
\begin{equation}\label{eq:weighted_approx_eigenstate_Hg}
\left\|
\left(
H_{\mathbf g}
-
\lambda_{\mathbf J}
\right)
\ket{\rho_{\mathbf J,\mathbf r}}
\right\|
=o(\sum_tg_tm_t),
\end{equation}
where
\[
\lambda_{\mathbf J}
:=
2\sum_{t=1}^N g_t\sqrt{J_t(m_t-J_t)}
=
(1+o(1))\sum_{t=1}^N2g_t\sqrt{\alpha_t(1-\alpha_t)}\,m_t.
\]
\end{lem}

\begin{proof}
For any $t\in [N]$, let
\[
C_t:=\sqrt{J_t(m_t-J_t)},
\]
and 
\begin{equation}\label{eq:eps_t_concentration}
\varepsilon_t
:=
\left|\frac{J_t}{m_t}-\alpha_t\right|
+\frac{r_t}{m_t}
+\frac{1}{\sqrt{r_t}}
+\frac{1}{m_t}.
\end{equation}
Under the assumptions $\alpha_t\in (0,1)$ are constants, $J_t=(\alpha_t+o(1))m_t$, $r_t\to\infty$, and $r_t=o(m_t)$, we have $\varepsilon_t=o(1)$.

For every $\mathbf j\in\mathcal R$, the $t$-th coordinate satisfies $j_t=J_t+O(r_t)$. Since $\alpha_t\in(0,1)$ is fixed, the function
$x\mapsto \sqrt{x(m_t-x+1)}$
has a derivative that is uniformly bounded for $x=J_t+O(r_t)$ and sufficiently large $m_t$, and
the bound only depends on $\alpha_t$. Hence, for all
$\mathbf j\in\mathcal R$,
\[
\sqrt{j_t(m_t-j_t+1)}
=
C_t+O(r_t),
\]
and similarly
\[
\sqrt{(j_t+1)(m_t-j_t)}
=
C_t+O(r_t).
\]
By \eqref{eq:Xt_action_on_P_basis}, it follows that
\[
H_t\ket{\rho_{\mathbf J,\mathbf r}}
=
C_t\left(U_t^-+U_t^+\right)
\ket{\rho_{\mathbf J,\mathbf r}}
+
O(r_t)\ket{\xi},
\]
where $U_t^-$ and $U_t^+$ shift the $t$-th coordinate down and up by one,
respectively, on the basis vectors appearing in the rectangle, and 
$\ket{\xi}$ is a unit vector.

The vector $U_t^+\ket{\rho_{\mathbf J,\mathbf r}}$ is a normalized state with uniform amplitudes on the shifted rectangle $\mathcal R+\mathbf e_t$. 
Therefore 
\[
\left\|
U_t^+\ket{\rho_{\mathbf J,\mathbf r}}
-
\ket{\rho_{\mathbf J,\mathbf r}}
\right\|^2
=
O\!\left(\frac1{r_t}\right).
\]
Thus
\[
\left\|
U_t^+\ket{\rho_{\mathbf J,\mathbf r}}
-
\ket{\rho_{\mathbf J,\mathbf r}}
\right\|
=
O(r_t^{-1/2}).
\]
The same estimate holds for $U_t^-$. Since $C_t=O(m_t)$, we obtain
\[
H_t\ket{\rho_{\mathbf J,\mathbf r}}
=
2C_t\ket{\rho_{\mathbf J,\mathbf r}}
+
O\!\left(\frac{m_t}{\sqrt{r_t}}+r_t+1\right).
\]
Equivalently,
\begin{equation}\label{eq:Ht_minus_Ct_explicit}
\left\|
\left(H_t-2C_t\right)
\ket{\rho_{\mathbf J,\mathbf r}}
\right\|
=
O\!\left(\frac{m_t}{\sqrt{r_t}}+r_t+1\right).
\end{equation}

It remains to compare $2C_t$ with
$2\sqrt{\alpha_t(1-\alpha_t)}m_t$. Since the map
$x\mapsto \sqrt{x(1-x)}$ is Lipschitz on a fixed neighborhood of $\alpha_t$,
we have
\[
\left|
\frac{C_t}{m_t}
-
\sqrt{\alpha_t(1-\alpha_t)}
\right|
=
O\!\left(
\left|\frac{J_t}{m_t}-\alpha_t\right|
\right).
\]
Therefore
\[
\left|
2C_t
-
2\sqrt{\alpha_t(1-\alpha_t)}m_t
\right|
=
O\!\left(
m_t\left|\frac{J_t}{m_t}-\alpha_t\right|
\right).
\]
Combining this estimate with \eqref{eq:Ht_minus_Ct_explicit} gives
\[
\left\|
\left(
H_t-2\sqrt{\alpha_t(1-\alpha_t)}\,m_t
\right)
\ket{\rho_{\mathbf J,\mathbf r}}
\right\|
=
O(\varepsilon_t m_t),
\]
where $\varepsilon_t$ is defined in \eqref{eq:eps_t_concentration}. Since $\varepsilon_t=o(1)$, we have
\[
\left\|
\left(
H_t-2\sqrt{\alpha_t(1-\alpha_t)}\,m_t
\right)
\ket{\rho_{\mathbf J,\mathbf r}}
\right\|
=o(m_t).
\]

Moreover, summing \eqref{eq:Ht_minus_Ct_explicit} over $t$ with weights $g_t$
gives
\[
\left\|
\left(
H_{\mathbf g}
-
\lambda_{\mathbf J}
\right)
\ket{\rho_{\mathbf J,\mathbf r}}
\right\|
=
O\!\left(
\sum_{t=1}^N
g_t m_t
\left(
\frac{r_t}{m_t}
+
\frac{1}{\sqrt{r_t}}
+
\frac{1}{m_t}
\right)
\right),
\]
where
\[
\lambda_{\mathbf J}
:=
2\sum_{t=1}^N g_t\sqrt{J_t(m_t-J_t)}.
\]
Since
\[
2\sqrt{J_t(m_t-J_t)}
=
\left(2\sqrt{\alpha_t(1-\alpha_t)}+O\!\left(
\left|\frac{J_t}{m_t}-\alpha_t\right|
\right)\right)m_t,
\]
we also have
\[
\lambda_{\mathbf J}
=
\sum_{t=1}^N2g_t\sqrt{\alpha_t(1-\alpha_t)}\,m_t
+
O\!\left(
\sum_{t=1}^N
g_t m_t
\left|\frac{J_t}{m_t}-\alpha_t\right|
\right).
\]
By the definition of $\varepsilon_t$ in \eqref{eq:eps_t_concentration} and $\varepsilon_t=o(1)$, 
we have 
$$\lambda_{\mathbf J}=(1+o(1))\sum_{t=1}^N2g_t\sqrt{\alpha_t(1-\alpha_t)}\,m_t.$$

\end{proof}

Now, we are ready to prove the concentration result in Theorem \ref{thm:concentration}.

\begin{mproof}[Proof of Theorem \ref{thm:concentration}]
For each $t\in[N]$, define
\[
\varepsilon_t
:=
\left|\frac{J_t}{m_t}-\alpha_t\right|
+\frac{r_t}{m_t}
+\frac{1}{\sqrt{r_t}}
+\frac{1}{m_t}.
\]
By the assumptions on $J_t$ and $r_t$, we have $\varepsilon_t=o(1)$.

By the definition of $H_t$, the operator $H_t$ is diagonal in the computational
basis, with eigenvalue $m_t-2|(B\Bxx-\Bvv)_{S_t}|$
on $\ket{\Bxx}$. By Lemma \ref{claim:weighted_claim32}, we have
\begin{align*}
   &\sum_{\Bxx\in\mathbb{F}^n_2}
   \Pr_{\rho_{\mathbf J,\mathbf r}}(\Bxx)
   \left(
   m_t-2|(B\Bxx-\Bvv)_{S_t}|
   -2\sqrt{\alpha_t(1-\alpha_t)}m_t
   \right)^2 \\
=&
\left\|
\left(
H_t-2\sqrt{\alpha_t(1-\alpha_t)}m_t
\right)
\ket{\rho_{\mathbf J,\mathbf r}}
\right\|^2  \\
=&\ O(\varepsilon_t^2 m_t^2).
\end{align*}
Therefore, by Markov's inequality,
\begin{align*}
\Pr_{\rho_{\mathbf J,\mathbf r}}
\left(
\left|
m_t-2|(B\Bxx-\Bvv)_{S_t}|
-2\sqrt{\alpha_t(1-\alpha_t)}m_t
\right|
>
\varepsilon_t^{1/2}m_t
\right)  
\le
\frac{
O(\varepsilon_t^2m_t^2)
}{
\varepsilon_t m_t^2
}
=
O(\varepsilon_t).
\end{align*}
Equivalently, with probability at least $1-O(\varepsilon_t)$, we have
\[
m_t-2|(B\Bxx-\Bvv)_{S_t}|
=
2\sqrt{\alpha_t(1-\alpha_t)}m_t
+
O(\varepsilon_t^{1/2}m_t).
\]
This is the same as
\[
|(B\Bxx-\Bvv)_{S_t}|
=
\left(
\frac12-\sqrt{\alpha_t(1-\alpha_t)}
\right)m_t
+
O(\varepsilon_t^{1/2}m_t) = \beta_t m_t
+
O(\varepsilon_t^{1/2}m_t).
\] 

Define the explicit concentration set
\[
S_{\varepsilon}
:=
\left\{
\Bxx\in\mathbb F_2^n:
\left|
\frac{|(B\Bxx-\Bvv)_{S_t}|}{m_t}
-\beta_t
\right|
\le
C\varepsilon_t^{1/2}
\ \text{for every }t\in[N]
\right\},
\]
where $C>0$ is a sufficiently large constant independent of $m$. Applying the
union bound over the finitely many blocks gives
\[
\Pr_{\rho_{\mathbf J,\mathbf r}}(\Bxx\notin S_{\varepsilon})
\le
\sum_{t=1}^N O(\varepsilon_t).
\]
Hence
\[
\Pr_{\rho_{\mathbf J,\mathbf r}}(\Bxx\in S_{\varepsilon})
=
1-
O\!\left(
\sum_{t=1}^N\varepsilon_t
\right).
\]
Since $N$ is a fixed constant, this may equivalently be written as
\[
\Pr_{\rho_{\mathbf J,\mathbf r}}(\Bxx\in S_{\varepsilon})
=
1-
O\!\left(
\max_{t\in[N]}\varepsilon_t
\right).
\]
In particular, because each $\varepsilon_t=o(1)$, the set $S_{\varepsilon}$ is
of the form
\[
S_{\varepsilon}
\subseteq
\left\{
\Bxx\in\mathbb F_2^n:
|(B\Bxx-\Bvv)_{S_t}|
=
(\beta_t+o(1))m_t,\ \forall t\in[N]
\right\},
\]
and thus we have 
\[
\Pr_{\rho_{\mathbf J,\mathbf r}}(\Bxx\in S)=1-o(1).
\]

\end{mproof}

We now prove the result on the size of the set \(S\).
Let us start with the lower bound.

\begin{prop}\label{prop:weighted_many_solutions}
Assume that $B:\mathbb F_2^n\to\mathbb F_2^m$ has rank $n$,  $\Bvv\in \BFF_2^m$ is a fixed vector, the number of blocks $N$ is a fixed constant
and $m_t=\theta_t m+o(m), \forall t\in [N],
\sum_{t=1}^N\theta_t=1$.
Let $C=\{B\Bxx: \Bxx\in\BFF_2^n\}\subseteq\BFF_2^m 
$ be a code in $\BFF_2^m$ and let $d^\perp=d(C^\perp)$ be the dual distance with $d^\perp=\delta^\perp m+o(m)$ for some $\delta^\perp\in(0,1]$.  
Assume also that $0<\alpha_t<\frac12,
\forall t\in [N]$,
and that there is a constant $\gamma>0$ such that
\begin{equation}\label{eq:weighted_many_solutions_gap_app}
\sum_{t=1}^N\theta_t\alpha_t
\le
\frac{\delta^\perp}{2}-\gamma .
\end{equation}
Then, the size of the set $S$ defined in Theorem \ref{thm:concentration},
on which the state \(|\rho_{\mathbf J,\mathbf r}\rangle\) is concentrated,
has the following lower bound: there exists some constant 
$a>0$ such that 
\begin{align}
    |S|\geq 2^{an}.
\end{align}
\end{prop}

\begin{proof}
    Based on Theorem \ref{thm:concentration},
    the state $\ket{\rho_{\mathbf J,\mathbf r}}$ is asymptotically concentrated on bitstrings $\Bxx$ satisfying
\[
|(B\Bxx-\Bvv)_{S_t}|=(\beta_t+o(1))m_t,
\qquad \forall t\in [N].
\]
We will prove the upper bound on  the probability of observing any one such bitstring.

Fix $\Bxx$, let us  denote $q_t$ to be 
the Hamming weight of the vector  $B\Bxx-\Bvv$ restricted to 
$S_t$, that is 
\[
q_t:=|(B\Bxx-\Bvv)_{S_t}|,
\qquad \forall t\in [N].
\]

Define the binary Krawtchouk polynomial ~\cite{krawtchouk1929generalisation,nikiforov1991classical}  by
\begin{equation}\label{eq:Krawtchouk_def_section}
K_k(q;m):=
\sum_{a=0}^k(-1)^a\binom{q}{a}\binom{m-q}{k-a}.
\end{equation}
Then, for every $\mathbf j$, the block-symmetric polynomial state  can be rewritten as follows
\begin{equation}\label{eq:Pj_computational_basis_Krawtchouk}
\ket{\mathbf P(\mathbf j)}
=
\frac{1}{\sqrt{2^n\prod_{t=1}^N\binom{m_t}{j_t}}}
\sum_{\Bxx\in\BFF_2^n}
\left(
\prod_{t=1}^N
K_{j_t}\!\left(\wt((B\Bxx-\Bvv)_{S_t});m_t\right)
\right)
\ket{\Bxx}.
\end{equation}
That is, the coefficient 
\[
\langle\Bxx|\mathbf P({\mathbf j})\rangle
=
\frac{1}{\sqrt{2^n\prod_{t=1}^N\binom{m_t}{j_t}}}
\prod_{t=1}^NK_{j_t}(q_t;m_t).
\]
Hence, by the definition of the state $\ket{\rho_{\mathbf J,\mathbf r}}=
\frac{1}{\sqrt{|\mathcal R|}}
\sum_{\mathbf j\in\mathcal R}
\ket{\mathbf P(\mathbf j)}$, we have
\[
|\langle\Bxx|\rho_{\mathbf J,\mathbf r}\rangle|^2
\le
|\mathcal R|
\max_{\mathbf j\in\mathcal R}
\frac{1}{2^n\prod_{t=1}^N\binom{m_t}{j_t}}
\prod_{t=1}^N K_{j_t}(q_t;m_t)^2.
\]
Moreover, by the orthogonality relation of  Krawtchouk polynomials
\[
\sum_{q=0}^{m_t}\binom{m_t}{q}K_{j_t}(q;m_t)^2
=2^{m_t}\binom{m_t}{j_t},
\]
we have 
\[
K_{j_t}(q_t;m_t)^2
\le
\frac{2^{m_t}\binom{m_t}{j_t}}{\binom{m_t}{q_t}},
\]

 Therefore
\begin{equation}\label{eq:single_bitstring_probability_bound}
|\langle\Bxx|\rho_{\mathbf J,\mathbf r}\rangle|^2
\le
|\mathcal R|\,
\frac{2^{m-n}}{\prod_{t=1}^N\binom{m_t}{q_t}}.
\end{equation}

Since the 
dual code $C^\perp$ has length $m$, dimension $m-n$, and distance $d^\perp$,  the Hamming bound \cite{lint1999introduction} gives
\[
2^{m-n}
\le
\frac{2^m}{\sum_{i=0}^{\lfloor(d^\perp-1)/2\rfloor}\binom mi}
\le
2^{m(1-h(\delta^\perp/2))+o(m)},
\]
where $h(x)=-x\log_2x-(1-x)\log_2(1-x)$.
If $\Bxx$ lies in the concentration region, then
\[
\prod_{t=1}^N\binom{m_t}{q_t}
=
2^{m\sum_{t=1}^N\theta_t h(\beta_t)+o(m)}.
\]
Moreover, as $N$ is a fixed constant and $|\mathcal R|=\prod_{t=1}^N(r_t-1)=m^{O(1)}$, 
we have
\[
|\langle\Bxx|\rho_{\mathbf J,\mathbf r}\rangle|^2
\le
2^{m\left(1-h(\delta^\perp/2)-\sum_{t=1}^N\theta_t h(\beta_t)\right)+o(m)}.
\]
For $0<\alpha<1/2$ and $\beta=1/2-\sqrt{\alpha(1-\alpha)}$, we have
\[
h(\beta)>1-h(\alpha).
\]
Indeed, the elementary inequality $h(x)>4x(1-x)$ for $0<x<1$, $x\neq1/2$, gives
\[
h(\beta)>4\beta(1-\beta)=(1-2\alpha)^2=1-4\alpha(1-\alpha)>1-h(\alpha).
\]
Therefore
\begin{align*}
1-h(\delta^\perp/2)-\sum_{t=1}^N\theta_t h(\beta_t)
&<
-h(\delta^\perp/2)+\sum_{t=1}^N\theta_t h(\alpha_t)\\
&\le
-h(\delta^\perp/2)+h\left(\sum_{t=1}^N\theta_t\alpha_t\right)\\
&<0,
\end{align*}
where we used the concavity of $h$ and the gap condition \eqref{eq:weighted_many_solutions_gap_app}. Hence there is a constant $a>0$ such that every bitstring in the set $S$ has probability at most $2^{-am}$, 
that is, $\text{Pr}(\Bxx)\leq 2^{-am}$ for any $\Bxx\in S$. 
Since 
$\text{Pr}(\Bxx\in S)=1-o(1)=\sum_{\Bxx\in S}\text{Pr}(\Bxx)$,  then we have 
$|S|\geq 2^{am}$. Moreover, since $B:\mathbb F_2^n\to\mathbb F_2^m$ has rank $n$,  we have 
$m\geq n$. Hence, 
the size of $S$ has the following lower bound
$$|S|\geq 2^{an}.$$

\end{proof}

Next, we derive an upper bound on the size of the set \(S\), for which we first need some preparation.

For each \(t\in[N]\), the polynomials
\[
K_0(\cdot;m_t),\,K_1(\cdot;m_t),\,\dots,\,K_{m_t}(\cdot;m_t)
\]
form a basis for the space of all real-valued functions on \(\{0,1,\dots,m_t\}\). Therefore, the tensor-product functions
\[
\Phi_{\mathbf{k}}(\mathbf{x})
:=
\prod_{t=1}^N K_{k_t}(x_t;m_t),
\qquad
\mathbf{k}=(k_1,\dots,k_N),\quad 0\le k_t\le m_t,
\]
form a basis for the space of real-valued functions on $\prod_{t=1}^N \{0,1,\dots,m_t\}$.
Equivalently, every function
\[
\alpha:\prod_{t=1}^N \{0,1,\dots,m_t\}\to \mathbb{R}
\]
admits a unique expansion of the form
\begin{equation}\label{eq:block-krawtchouk-expansion}
\alpha(\mathbf{x})
=
\sum_{\mathbf{k}} \alpha_{\mathbf{k}}
\prod_{t=1}^N K_{k_t}(x_t;m_t).
\end{equation}

Let us prove the following result which 
is a Block-profile extension of Lemma 1 in \cite{krasikov2002estimates}.
\begin{prop}\label{prop:block-delsarte}
Let $C\subset \mathbb{F}_2^m$ be a binary linear code, and let
$[m]=S_1\sqcup \cdots \sqcup S_N,
|S_t|=m_t,
\sum_{t=1}^N m_t=m.$
For a block profile $\mathbf{i}=(i_1,\dots,i_N)$ with $0\le i_t\le m_t$, define
\[
A_{\mathbf{i}}(C)
:=
\bigl|\{c\in C:\ |c_{S_t}|=i_t\ \text{for all }t\in [N]\}\bigr|.
\]
Similarly, define $A_{\mathbf{k}}(C^\perp)$ for the dual code $C^\perp\le \mathbb{F}_2^m$.
Suppose that $\alpha$ is expanded as in \eqref{eq:block-krawtchouk-expansion}. Then
\begin{equation}\label{eq:block-delsarte}
|C|\sum_{\mathbf{k}}\alpha_{\mathbf{k}}\,A_{\mathbf{k}}(C^\perp)
=
\sum_{\mathbf{i}} A_{\mathbf{i}}(C)\,\alpha(\mathbf{i}).
\end{equation}
Equivalently,
\begin{equation}\label{eq:block-delsarte-zero}
|C|\,\alpha_{\mathbf{0}}
+
|C|\sum_{\mathbf{k}\neq \mathbf{0}}\alpha_{\mathbf{k}}\,A_{\mathbf{k}}(C^\perp)
=
\sum_{\mathbf{i}} A_{\mathbf{i}}(C)\,\alpha(\mathbf{i}),
\end{equation}
where $\mathbf{0}=(0,\dots,0)$.
\end{prop}

\begin{proof}
For each $t\in [N]$ and each $u\in \mathbb{F}_2^{m_t}$, one has
\begin{equation}\label{eq:single-block-character-identity}
\sum_{\substack{z\in \mathbb{F}_2^{m_t}\\ |z|=k}}
(-1)^{u\cdot z}
=
K_k(|u|;m_t).
\end{equation}
Now fix a block profile $\mathbf{k}=(k_1,\dots,k_N)$. For any $c\in C$, applying \eqref{eq:single-block-character-identity} on each block yields
\[
\prod_{t=1}^N K_{k_t}(|c_{S_t}|;m_t)
=
\prod_{t=1}^N
\sum_{\substack{z_t\in \mathbb{F}_2^{S_t}\\ |z_t|=k_t}}
(-1)^{c_{S_t}\cdot z_t}
=
\sum_{\substack{y\in \mathbb{F}_2^m\\ |y_{S_t}|=k_t\ \forall t}}
(-1)^{c\cdot y}.
\]
Therefore,
\begin{align*}
\sum_{\mathbf{i}} A_{\mathbf{i}}(C)\prod_{t=1}^N K_{k_t}(i_t;m_t)
=
\sum_{c\in C}\prod_{t=1}^N K_{k_t}(|c_{S_t}|;m_t)
=
\sum_{\substack{y\in \mathbb{F}_2^m\\ |y_{S_t}|=k_t\ \forall t}}
\sum_{c\in C} (-1)^{c\cdot y}.
\end{align*}
By character orthogonality,
\[
\sum_{c\in C}(-1)^{c\cdot y}
=
\begin{cases}
|C|, & y\in C^\perp,\\
0, & y\notin C^\perp.
\end{cases}
\]
Hence
\begin{equation}\label{eq:block-macwilliams}
\sum_{\mathbf{i}} A_{\mathbf{i}}(C)\prod_{t=1}^N K_{k_t}(i_t;m_t)
=
|C|\,A_{\mathbf{k}}(C^\perp).
\end{equation}
Finally, multiply \eqref{eq:block-macwilliams} by $\alpha_{\mathbf{k}}$ and sum over $\mathbf{k}$. Using
\eqref{eq:block-krawtchouk-expansion}, we obtain
\begin{align*}
|C|\sum_{\mathbf{k}}\alpha_{\mathbf{k}}A_{\mathbf{k}}(C^\perp)
&=
\sum_{\mathbf{k}}\alpha_{\mathbf{k}}
\sum_{\mathbf{i}} A_{\mathbf{i}}(C)\prod_{t=1}^N K_{k_t}(i_t;m_t)\\
&=
\sum_{\mathbf{i}} A_{\mathbf{i}}(C)
\sum_{\mathbf{k}}\alpha_{\mathbf{k}}\prod_{t=1}^N K_{k_t}(i_t;m_t)\\
&=
\sum_{\mathbf{i}} A_{\mathbf{i}}(C)\,\alpha(\mathbf{i}),
\end{align*}
which is exactly \eqref{eq:block-delsarte}.  
\end{proof}

\begin{Rem}\label{rem:block-macwilliams-reference}
The coefficient identity \eqref{eq:block-macwilliams} is the binary split Hamming weight MacWilliams transform.
For the classical two-block case, see \cite[Ch. 5, Eq. (52)]{macwilliams1977theory}.
For the product-partition formulation, see \cite[Thm. 3.4]{gluesing2015fourier}.
The proof above is included only to keep the argument self-contained.
\end{Rem}

\begin{cor}\label{cor:block-kl02}
Assume that \(C^\perp\) has minimum distance \(d^\perp\). If
\[
\alpha_{\mathbf{k}}=0
\qquad
\text{for all }\mathbf{k}\text{ with }|\mathbf{k}|:=k_1+\cdots+k_N\ge d^\perp,
\]
then
\begin{equation}\label{eq:block-kl02}
|C|\,\alpha_{\mathbf{0}}
=
\sum_{\mathbf{i}} A_{\mathbf{i}}(C)\,\alpha(\mathbf{i}).
\end{equation}
\end{cor}

\begin{proof}
If \(\mathbf{k}\neq \mathbf{0}\) and \(|\mathbf{k}|<d^\perp\), then no nonzero codeword of \(C^\perp\) can have block profile \(\mathbf{k}\), since every nonzero codeword of \(C^\perp\) has total Hamming weight at least \(d^\perp\). Therefore,
\[
A_{\mathbf{k}}(C^\perp)=0
\qquad
\text{for all }\mathbf{k}\neq \mathbf{0}\text{ with }|\mathbf{k}|<d^\perp.
\]
On the other hand, by assumption, \(\alpha_{\mathbf{k}}=0\) whenever \(|\mathbf{k}|\ge d^\perp\). Hence, in \eqref{eq:block-delsarte}, the only nonzero term on the left-hand side is the one corresponding to \(\mathbf{k}=\mathbf{0}\). This gives
\[
|C|\,\alpha_{\mathbf{0}}
=
\sum_{\mathbf{i}} A_{\mathbf{i}}(C)\,\alpha(\mathbf{i}),
\]
as claimed.
\end{proof}

 For any vector $v\in \mathbb{F}_2^m$ and a block profile
$\mathbf i=(i_1,\ldots,i_N)$,  we denote
\begin{align}\label{eq:codeword-from-v}
A_{\mathbf{i}}(v)
:=
\bigl|\{c\in C:\ |(c-v)_{S_t}|=i_t\ \forall t\in [N]\}\bigr|.
\end{align}

\begin{cor}[Shifted block-profile identity]\label{cor:block-shifted}
Let $\alpha$ be the function defined as in \eqref{eq:block-krawtchouk-expansion},
 we have
\begin{equation}\label{eq:block-shifted}
\sum_{\mathbf{i}} A_{\mathbf{i}}(v)\,\alpha(\mathbf{i})
=
|C|
\sum_{\mathbf{k}}\alpha_{\mathbf{k}}
\sum_{\substack{y\in C^\perp\\ |y_{S_t}|=k_t\ \forall t}}
(-1)^{v\cdot y}.
\end{equation}
In particular, if $\alpha_{\mathbf{k}}=0$ whenever $|\mathbf{k}|\ge d^\perp$, then
\begin{equation}\label{eq:block-shifted-truncated}
|C|\,\alpha_{\mathbf{0}}
=
\sum_{\mathbf{i}} A_{\mathbf{i}}(v)\,\alpha(\mathbf{i}).
\end{equation}
\end{cor}

\begin{proof}
We
repeat the proof of Proposition~\ref{prop:block-delsarte} with $c$ replaced by $c-v$.
Then
\[
\prod_{t=1}^N K_{k_t}(|(c-v)_{S_t}|;m_t)
=
\sum_{\substack{y\in \mathbb{F}_2^m\\ |y_{S_t}|=k_t\ \forall t}}
(-1)^{(c-v)\cdot y},
\]
and summing over $c\in C$ gives
\[
\sum_{\mathbf{i}} A_{\mathbf{i}}(v)\prod_{t=1}^N K_{k_t}(i_t;m_t)
=
|C|
\sum_{\substack{y\in C^\perp\\ |y_{S_t}|=k_t\ \forall t}}
(-1)^{v\cdot y}.
\]
Multiplying by $\alpha_{\mathbf{k}}$ and summing over $\mathbf{k}$ yields \eqref{eq:block-shifted}.
The final statement follows from that every $y\in C^\perp $ has $wt(y)\ge d^\perp$.
\end{proof}

We will need the following multiplication formula for Krawtchouk polynomials.

\begin{lem}[Krawtchouk multiplication formula~\cite{lint1999introduction,krasikov1996integral}]
\label{lem:krawtchouk_multiplication}
Let
\[
K_r(x;m)
=
\sum_{a=0}^r (-1)^a \binom{x}{a}\binom{m-x}{r-a}
\]
denote the binary Krawtchouk polynomial of degree \(r\) and length \(m\).
Then, for any \(0\le i,j\le m\) and any \(0\le x\le m\),
\begin{equation}
\label{eq:krawtchouk_multiplication}
K_i(x;m)K_j(x;m)
=
\sum_{k=\max(0,i+j-m)}^{\min(i,j)}
\binom{m-i-j+2k}{k}
\binom{i+j-2k}{j-k}
K_{i+j-2k}(x;m).
\end{equation}
\end{lem}

\begin{prop}[Finite-length block-profile upper bound]\label{prop:block_KL02_finite}
Assume that $B:\mathbb F_2^n\to \mathbb F_2^m$, and set
\[
C:=\{B\Bxx:\Bxx\in\mathbb F_2^n\}\subseteq \mathbb F_2^m.
\]
Let $d^\perp$ be the minimum distance of $C^\perp$, and fix
$\Bvv\in\mathbb F_2^m$.
Let
\[
\boldsymbol\ell=(\ell_1,\ldots,\ell_N)\in \prod_{t=1}^N\{0,1,\ldots,m_t\}
\]
satisfying $2|\boldsymbol\ell|<d^\perp$.
Then
\begin{equation}\label{eq:block_KL02_finite_identity}
|C|\prod_{t=1}^N \binom{m_t}{\ell_t}
=
\sum_{\mathbf i} A_{\mathbf i}(\Bvv)\prod_{t=1}^N K_{\ell_t}(i_t;m_t)^2.
\end{equation}
Consequently, if $K_{\ell_t}(i_t;m_t)\neq 0$ for all $t$, then
\begin{equation}\label{eq:block_KL02_finite_raw}
A_{\mathbf i}(\Bvv)
\le
\frac{|C|\prod_{t=1}^N \binom{m_t}{\ell_t}}
{\prod_{t=1}^N K_{\ell_t}(i_t;m_t)^2}.
\end{equation}
Equivalently,
\begin{equation}\label{eq:block_KL02_finite_binomial}
A_{\mathbf i}(\Bvv)
\le
\frac{|C|\prod_{t=1}^N \binom{m_t}{i_t}}{2^{m}}
\prod_{t=1}^N
\frac{2^{m_t}\binom{m_t}{\ell_t}}
{\binom{m_t}{i_t}K_{\ell_t}(i_t;m_t)^2}.
\end{equation}
\end{prop}

\begin{proof}
Consider the test function
\[
\alpha(\mathbf x)
:=
\prod_{t=1}^N K_{\ell_t}(x_t;m_t)^2,
\qquad
\mathbf x=(x_1,\ldots,x_N).
\]
By the multiplication formula for Krawtchouk polynomials  in Lemma \ref{lem:krawtchouk_multiplication}, 
for each $t$ one has
\begin{equation}\label{eq:block_KL02_square_expansion}
K_{\ell_t}(x_t;m_t)^2
=
\sum_{u_t=0}^{\ell_t}
\binom{2u_t}{u_t}\binom{m_t-2u_t}{\ell_t-u_t}\,
K_{2u_t}(x_t;m_t).
\end{equation}
Hence the expansion of $\alpha$ in the tensor-product Krawtchouk basis is supported on multi-indices $\mathbf k=(k_1,\ldots,k_N)$ satisfying
\[
k_t=2u_t,\qquad 0\le u_t\le \ell_t,
\]
and therefore $|\mathbf k|\le 2|\boldsymbol\ell|<d^\perp $.
Moreover, the constant coefficient in \eqref{eq:block_KL02_square_expansion} is $\binom{m_t}{\ell_t}$, and hence
\[
\alpha_{\mathbf 0}
=
\prod_{t=1}^N \binom{m_t}{\ell_t}.
\]
Applying Corollary \ref{cor:block-shifted} in its truncated form \eqref{eq:block-shifted-truncated}, we get
\[
|C|\,\alpha_{\mathbf 0}
=
\sum_{\mathbf i}A_{\mathbf i}(\Bvv)\alpha(\mathbf i),
\]
which is exactly \eqref{eq:block_KL02_finite_identity}.

Since every term on the right-hand side of \eqref{eq:block_KL02_finite_identity} is nonnegative, the $\mathbf i$-term gives
\[
A_{\mathbf i}(\Bvv)
\prod_{t=1}^N K_{\ell_t}(i_t;m_t)^2
\le
|C|\prod_{t=1}^N \binom{m_t}{\ell_t}.
\]
This proves \eqref{eq:block_KL02_finite_raw}. 
Finally, since
\[
2^m 
=
\prod_{t=1}^N2^{m_t},
\]
we obtain \eqref{eq:block_KL02_finite_binomial}.
\end{proof}

The following lemma comes from the proof of Theorem 1 in \cite{krasikov2002estimates}.
\begin{lem}[\cite{krasikov2002estimates}]
\label{lem:KL02_ratio_estimate}
Let $n\to\infty$, and let $j=j(n)$ be a sequence of integers satisfying
\[
\frac{j}{n}\to \beta,
\qquad
0<\beta<1.
\]
Define
\[
\tau(\beta):=\frac12-\sqrt{\beta(1-\beta)}.
\]
%Assume first that $0<\beta<1/2$. 
Then there exists a sequence of integers $\ell=\ell(n)$ such that
\[
\frac{\ell}{n}\to \tau(\beta), 
\] and
\begin{equation}\label{eq:KL02_ratio_estimate}
R_n(\ell,j)
:=
\frac{2^n\binom{n}{\ell}}
{\binom{n}{j}K_\ell(j;n)^2}
=
O(n).
\end{equation}
\end{lem}

\begin{proof}
For $0<\beta<1/2$, this is exactly the one-dimensional estimate used in \cite[Theorem 1, Eq. (10)]{krasikov2002estimates}. 
For $1/2<\beta<1$, apply the previous case to $n-j$, whose relative value tends to $1-\beta<1/2$. Since
\[
\binom{n}{n-j}=\binom{n}{j}
\]
and the binary Krawtchouk polynomials satisfy
\[
K_\ell(n-j;n)=(-1)^\ell K_\ell(j;n),
\]
the square $K_\ell(j;n)^2$ is unchanged. 
Hence the same bound follows.
Finally,
when $\beta=1/2$,
we can choose $\ell =0$,
then \eqref{eq:KL02_ratio_estimate} also holds.
\end{proof}

\begin{thm}\label{thm:block_KL02_theorem1}
Under the same assumptions as Proposition \ref{prop:weighted_many_solutions}, 
let $\mathbf q=\mathbf q(m)=(q_1,\ldots,q_N)$
be a sequence of block profiles satisfying
\[
q_t=(\beta_t+o(1))m_t,
\qquad
\beta_t = \frac12 - \sqrt{\alpha_t (1-\alpha_t)}\in(0,1/2),
\qquad
\forall t\in [N].
\]
Then we have
\begin{equation}\label{eq:block_KL02_bound}
A_{\mathbf q}(\Bvv)
\le
O(m^N)\,
\frac{\prod_{t=1}^N \binom{m_t}{q_t}}{2^{m-n}}.
\end{equation}

\end{thm}

\begin{proof}
% For each $t$, define
% \[
% \tau_t:=\frac12-\sqrt{\beta_t(1-\beta_t)}
% =\frac12\psi(\beta_t).
% \]
Apply Lemma~\ref{lem:KL02_ratio_estimate} with
$n=m_t,j=q_t,\beta=\beta_t,$
there exists a sequence of integers $\ell_t=\ell_t(m)$ such that
\[
\frac{\ell_t}{m_t}\to \frac12 - \sqrt{\beta_t (1-\beta_t)} = \alpha_t
\]
and
\begin{equation}\label{eq:block_KL02_ratio_each_block}
R_{m_t}(\ell_t,q_t)
:=
\frac{2^{m_t}\binom{m_t}{\ell_t}}
{\binom{m_t}{q_t}K_{\ell_t}(q_t;m_t)^2}
=
O(m_t).
\end{equation} 

Since
\[
\sum_{t=1}^N\theta_t\alpha_t <\frac{\delta^\perp}2,
\]
we have
\[
2|\boldsymbol\ell|
=
2\sum_{t=1}^N\ell_t
=
\left(
2\sum_{t=1}^N\theta_t\alpha_t+o(1)
\right)m
% =
% \left(
% \sum_{t=1}^N\theta_t\psi(\beta_t)+o(1)
% \right)m
<
d^\perp,
\]
for all sufficiently large $m$.
Note that $|C| = 2^n$.
Therefore, by  Proposition~\ref{prop:block_KL02_finite}, we have
\[
A_{\mathbf q}(\Bvv)
\le
\frac{\prod_{t=1}^N\binom{m_t}{q_t}}{2^{m-n}}
\prod_{t=1}^N
R_{m_t}(\ell_t,q_t).
\]
Using \eqref{eq:block_KL02_ratio_each_block}, we get
\[
\prod_{t=1}^N R_{m_t}(\ell_t,q_t)
=
O\!\left(\prod_{t=1}^N m_t\right)
=
O(m^N),
\]
because $N$ is fixed and $m_t=\theta_t m+o(m)$ with $\theta_t>0$.
Hence
\[
A_{\mathbf q}(\Bvv)
\le
O(m^N)
\frac{\prod_{t=1}^N\binom{m_t}{q_t}}{2^{m-n}}.
\]
Dividing by $2^n$ yields
\[
\frac{A_{\mathbf q}(\Bvv)}{2^n}
\le
O(m^N)
\prod_{t=1}^N\frac{\binom{m_t}{q_t}}{2^{m_t}}.
\]
This proves the theorem.
\end{proof}

\begin{Rem}\label{rem:block_KL02_N1}
When $N=1$, condition \eqref{eq:weighted_many_solutions_gap_app} becomes
\[
1-2\sqrt{\beta(1-\beta)}<\delta^\perp.
\]
This is equivalent to
\[
\beta\in
\left(
\frac12-\frac12\sqrt{\delta^\perp(2-\delta^\perp)},
\frac12+\frac12\sqrt{\delta^\perp(2-\delta^\perp)}
\right).
\]
Moreover, \eqref{eq:block_KL02_bound} becomes
\[
A_q(\Bvv)
\le
O(m)\frac{\binom mq}{2^{m-n}}.
\]
For $\Bvv=0$, this is exactly the form of \cite[Theorem~1]{krasikov2002estimates}, since
$|C^\perp|=2^{m-n}$.
\end{Rem}

\begin{prop}
\label{prop:weighted_effective_support}
Under the same assumptions as in Proposition~\ref{prop:weighted_many_solutions}, there exists a constant \(0<b<1\) such that the set \(S\) defined in Theorem~\ref{thm:concentration} satisfies the following upper bound
\[
|S|\le 2^{bn}.
\]
\end{prop}
\begin{proof}
Recall the set $S=\{\Bxx\in\mathbb F_2^n:\ |(B\Bxx-\Bvv)_{S_t}|=(\beta_t+o(1))m_t, \forall t\in [N]\}$.
Let $\mathcal W$ be the set of block profiles $\mathbf q=(q_1,\ldots,q_N)$ satisfying
\[
q_t=(\beta_t+o(1))m_t,
\qquad \forall t\in [N].
\]
Since $N$ is fixed and the window is sublinear in each block, $|\mathcal W|=O(m^N)$.

Now, let us provide an upper bound on the number of bitstrings in  $S$.
For each $\mathbf q\in\mathcal W$, Theorem~\ref{thm:block_KL02_theorem1} gives
\[
\frac{A_{\mathbf q}(\Bvv)}{2^n}
\le
m^{O(1)}
\prod_{t=1}^N\frac{\binom{m_t}{q_t}}{2^{m_t}}.
\]
By Stirling's approximation, for every $\mathbf q\in\mathcal W$, 
\[
\prod_{t=1}^N\frac{\binom{m_t}{q_t}}{2^{m_t}}
=
2^{-m\sum_{t=1}^N\theta_t(1-h(\beta_t))+o(m)}.
\]
Since $\alpha_t\in(0,1/2)$ in every nontrivial block, we have $\beta_t\neq 1/2$, hence
\[
\eta:=\sum_{t=1}^N\theta_t(1-h(\beta_t))>0.
\]
Therefore
\[
A_{\mathbf q}(\Bvv)\le 2^n\,2^{-\eta m+o(m)}
\]
for every $\mathbf q\in\mathcal W$.
Taking a union bound over $\mathcal W$ gives
\[
\sum_{\mathbf q\in\mathcal W} A_{\mathbf q}(\Bvv)
\le
2^n\,2^{-\eta m+o(m)}.
\]
Since $d^\perp=\delta^\perp m+o(m)$, the Hamming bound applied to $C^\perp$ implies
$n=\Theta(m)$. Hence there exists a constant $b<1$ such that
\[
\sum_{\mathbf q\in\mathcal W} A_{\mathbf q}(\Bvv)
\le
2^{b n}
\]
for all sufficiently large $m$.
Note that as $B$ has rank $n$,
we have
\[
A_{\mathbf q}(\Bvv)
:=
\bigl|
\{\Bxx\in\mathbb F_2^n:\ |(B\Bxx-\Bvv)_{S_t}|=q_t\ \forall t\in [N]\}
\bigr|.
\]
Hence we have
\[|S|\le 2^{b n}.\]
\end{proof}

Now, we are ready to prove the Theorem \ref{thm:weighted_many_solutions}.

\begin{mproof}[Proof of Theorem \ref{thm:weighted_many_solutions}]
This result comes directly from Proposition \ref{prop:weighted_many_solutions} and Proposition \ref{prop:weighted_effective_support}.
\end{mproof}

\section{Detailed Proof for Performance Without Distance Assumption}
\label{append:B}
\begin{lem}\label{lem:f2_nd_general_expression}
We have
\[
\langle s_{\mathbf g}^{(\CDD)}\rangle
=
\frac{\mathbf w^\dagger A^{(\CDD)}\mathbf w}{\braket{\rho^{(\CDD)}|\rho^{(\CDD)}}},
\]
where $A^{(\CDD)}$ is the $|T_l|\times |T_l|$ matrix indexed by $T_l$ whose $(\mathbf j,\mathbf k)$-entry is
\begin{equation}\label{eq:f2_nd_AD}
A^{(\CDD)}(\mathbf j,\mathbf k)
=
\frac{1}{\sqrt{|\CEE_{\mathbf j}|\,|\CEE_{\mathbf k}|}}
\sum_{i=1}^m c_i
\sum_{(\mathbf y_1,\mathbf y_2)\in S^{(i,\CDD)}_{\mathbf j,\mathbf k}}
(-1)^{\mathbf v\cdot(\mathbf y_1-\mathbf y_2+\mathbf e_i)},
\end{equation}
with
\[
S^{(i,\CDD)}_{\mathbf j,\mathbf k}
:=
\left\{(\mathbf y_1,\mathbf y_2)\in \CDD_{\mathbf j}\times \CDD_{\mathbf k}:
B^\top(\mathbf y_1-\mathbf y_2+\mathbf e_i)=0\right\}.
\]
\end{lem}
\begin{proof}
Recall that
\[
H_{\mathbf g}
=
\sum_{\Bxx\in \BFF_2^n}F_c(\Bxx)\ket{\Bxx}\bra{\Bxx}
=
\sum_{i=1}^m c_i H_{f_i},
\qquad
H_{f_i}=(-1)^{v_i}Z^{\mathbf b_i}.
\]
Hence
\[
\langle s_{\mathbf g}^{(\CDD)}\rangle
=
\frac{\bra{\rho^{(\CDD)}}H_{\mathbf g}\ket{\rho^{(\CDD)}}}{\braket{\rho^{(\CDD)}|\rho^{(\CDD)}}}.
\]
Using $H^{\otimes n}Z^{\mathbf b_i}H^{\otimes n}=X^{\mathbf b_i}$ and the formula of $\ket{\widetilde\rho^{(\CDD)}}$ in \eqref{eq:f2_nd_imperfect_tilde}, we obtain
\begin{align*}
&\bra{\rho^{(\CDD)}}H_{f_i}\ket{\rho^{(\CDD)}} \\
=&(-1)^{v_i}
\sum_{\mathbf j,\mathbf k\in T_l}
\frac{\overline{w_{\mathbf j}}w_{\mathbf k}}{\sqrt{|\CEE_{\mathbf j}|\,|\CEE_{\mathbf k}|}}
\sum_{\substack{\mathbf y_1\in \CDD_{\mathbf j}\\ \mathbf y_2\in \CDD_{\mathbf k}}}
(-1)^{\mathbf v\cdot(\mathbf y_1-\mathbf y_2)}
\braket{B^\top \mathbf y_1|X^{\mathbf b_i}|B^\top \mathbf y_2} \\
=&
\sum_{\mathbf j,\mathbf k\in T_l}
\frac{\overline{w_{\mathbf j}}w_{\mathbf k}}{\sqrt{|\CEE_{\mathbf j}|\,|\CEE_{\mathbf k}|}}
\sum_{\substack{\mathbf y_1\in \CDD_{\mathbf j}\\ \mathbf y_2\in \CDD_{\mathbf k}}}
(-1)^{\mathbf v\cdot(\mathbf y_1-\mathbf y_2+\mathbf e_i)}
\braket{B^\top \mathbf y_1|B^\top(\mathbf y_2-\mathbf e_i)}.
\end{align*}
Therefore
\[
\bra{\rho^{(\CDD)}}H_{\mathbf g}\ket{\rho^{(\CDD)}}
=
\mathbf w^\dagger A^{(\CDD)}\mathbf w,
\]
where $A^{(\CDD)}$ is given by \eqref{eq:f2_nd_AD}. Dividing by
$\braket{\rho^{(\CDD)}|\rho^{(\CDD)}}$
yields the claim.
\end{proof}

Replacing $\CDD_{\mathbf j}$ by $\CEE_{\mathbf j}$ in the construction above, we obtain the perfect matrix $A^{(\CEE)}$.

\begin{lem}\label{lem:f2_nd_average}
If $v_1,\ldots,v_m$ are chosen independently and uniformly at random from $\BFF_2$, then
\[
\Expect_{v_1,\ldots,v_m}A^{(\CEE)}=A^{(\mathbf g,l)},
\]
where  $A^{(\mathbf{g},l)}$ is a $|T_l|\times |T_l|$ matrix whose columns and rows are labeled by $T_l$,
and the $(\mathbf{j},\mathbf{k})$-entry of $A^{(\mathbf{g},l)}$ is
\begin{align}
A^{(\mathbf{g},l)}(\mathbf{j},\mathbf{k}) = \left\{
\begin{aligned}
&g_t\sqrt{j_t(m_t-j_t+1)}, && \text{ if } \mathbf{j}=\mathbf{k}+\mathbf{e}_t \text{ for some }t,\\
&g_t\sqrt{k_t(m_t-k_t+1)}, && \text{ if } \mathbf{j}=\mathbf{k}-\mathbf{e}_t \text{ for some }t,\\
&0, && \text{ otherwise. }
\end{aligned}
\right.
\end{align}
\end{lem}
\begin{proof}
For fixed $i,\mathbf j,\mathbf k$, split
\[
S^{(i,\CEE)}_{\mathbf j,\mathbf k}
=
S^{(i,\CEE,0)}_{\mathbf j,\mathbf k}
\cup
S^{(i,\CEE,1)}_{\mathbf j,\mathbf k},
\]
where
\[
S^{(i,\CEE,0)}_{\mathbf j,\mathbf k}
:=
\{(\mathbf y_1,\mathbf y_2)\in S^{(i,\CEE)}_{\mathbf j,\mathbf k}: \mathbf y_1-\mathbf y_2+\mathbf e_i=0\}
\]
and
\[
S^{(i,\CEE,1)}_{\mathbf j,\mathbf k}
:=
\{(\mathbf y_1,\mathbf y_2)\in S^{(i,\CEE)}_{\mathbf j,\mathbf k}: \mathbf y_1-\mathbf y_2+\mathbf e_i\neq 0\}.
\]
If $(\mathbf y_1,\mathbf y_2)\in S^{(i,\CEE,1)}_{\mathbf j,\mathbf k}$, then
\[
\Expect_{v_1,\ldots,v_m}(-1)^{\mathbf v\cdot(\mathbf y_1-\mathbf y_2+\mathbf e_i)}=0.
\]
Hence only the pairs in $S^{(i,\CEE,0)}_{\mathbf j,\mathbf k}$ contribute to the expectation.

If $S^{(i,\CEE,0)}_{\mathbf j,\mathbf k}\neq \emptyset$, then $\mathbf y_2=\mathbf y_1+\mathbf e_i$.
Thus $i\in S_t$ for some $t$, and necessarily $\mathbf j=\mathbf k+\mathbf e_t$ or
$\mathbf j=\mathbf k-\mathbf e_t$.
Assume for instance that $\mathbf j=\mathbf k+\mathbf e_t$.
Then for each fixed $i\in S_t$,
\[
\left|S^{(i,\CEE,0)}_{\mathbf k+\mathbf e_t,\mathbf k}\right|
=
\binom{m_t-1}{k_t}\prod_{s\neq t}\binom{m_s}{k_s}.
\]
Since $c_i=g_t$ on $S_t$, we obtain
\begin{align*}
&\Expect_{v_1,\ldots,v_m}A^{(\CEE)}(\mathbf k+\mathbf e_t,\mathbf k) \\
=&
\frac{g_t m_t\binom{m_t-1}{k_t}\prod_{s\neq t}\binom{m_s}{k_s}}
{\sqrt{\prod_{s=1}^N \binom{m_s}{k_s+\delta_{st}}\binom{m_s}{k_s}}}
=g_t\sqrt{(k_t+1)(m_t-k_t)} \\
=&A^{(\mathbf g,l)}(\mathbf k+\mathbf e_t,\mathbf k).
\end{align*}
The case $\mathbf j=\mathbf k-\mathbf e_t$ is the same by symmetry, and if
$\mathbf j-\mathbf k\notin\{\pm \mathbf e_1,\ldots,\pm \mathbf e_N\}$ then the relevant set is empty. Hence
$\Expect_{v_1,\ldots,v_m}A^{(\CEE)}=A^{(\mathbf g,l)}$.
\end{proof}

For fixed $i,\mathbf j,\mathbf k$, define
\[
T^{(i,\CFF)}_{\mathbf j,\mathbf k}
:=
\left\{(\mathbf y_1,\mathbf y_2)
\in
\CEE_{\mathbf j}\times \CFF_{\mathbf k}
\cup
\CFF_{\mathbf j}\times \CEE_{\mathbf k}:
\mathbf y_1-\mathbf y_2+\mathbf e_i=0\right\}.
\]

\begin{lem}\label{lem:f2_nd_decomposition}
If $v_1,\ldots,v_m$ are chosen independently and uniformly at random from $\BFF_2$, then
\[
\Expect_{v_1,\ldots,v_m}A^{(\CDD)}
=
A^{(\mathbf g,l)}-D^{(\CFF)},
\]
where $D^{(\CFF)}$ is the symmetric matrix indexed by $T_l$ whose $(\mathbf j,\mathbf k)$-entry is
\begin{equation}\label{eq:f2_nd_DF}
D^{(\CFF)}(\mathbf j,\mathbf k)
=
\frac{1}{\sqrt{|\CEE_{\mathbf j}|\,|\CEE_{\mathbf k}|}}
\sum_{i=1}^m c_i
\left|T^{(i,\CFF)}_{\mathbf j,\mathbf k}\right|.
\end{equation}
In particular, $D^{(\CFF)}(\mathbf j,\mathbf k)=0$ unless $\mathbf j-\mathbf k=\pm \mathbf e_t$ for some $t$.
\end{lem}
\begin{proof}
For each fixed $i,\mathbf j,\mathbf k$, we have the disjoint decomposition
\[
S^{(i,\CEE,0)}_{\mathbf j,\mathbf k}
=
S^{(i,\CDD,0)}_{\mathbf j,\mathbf k}
\cup
T^{(i,\CFF)}_{\mathbf j,\mathbf k}.
\]
Averaging over $\mathbf v$ keeps only the pairs satisfying
$\mathbf y_1-\mathbf y_2+\mathbf e_i=0$, so that
\[
\Expect_{v_1,\ldots,v_m}A^{(\CEE)}(\mathbf j,\mathbf k)
=
\Expect_{v_1,\ldots,v_m}A^{(\CDD)}(\mathbf j,\mathbf k)
+
D^{(\CFF)}(\mathbf j,\mathbf k).
\]
Hence, we get the result from Lemma~\ref{lem:f2_nd_average}.
\end{proof}

For $i\in S_t$ and $\mathbf j+\mathbf e_t\in T_l$, define
\[
\CEE_{\mathbf j}^{(0,i)}:=\{\mathbf y\in \CEE_{\mathbf j}: y_i=0\},
\qquad
\CEE_{\mathbf j+\mathbf e_t}^{(1,i)}:=\{\mathbf y\in \CEE_{\mathbf j+\mathbf e_t}: y_i=1\},
\]
\[
\CFF_{\mathbf j}^{(0,i)}:=\CFF_{\mathbf j}\cap \CEE_{\mathbf j}^{(0,i)},
\qquad
\CFF_{\mathbf j+\mathbf e_t}^{(1,i)}:=\CFF_{\mathbf j+\mathbf e_t}\cap \CEE_{\mathbf j+\mathbf e_t}^{(1,i)}.
\]
We then define the blockwise effective failure rates
\begin{equation}\label{eq:f2_nd_gamma0_appendix}
\widetilde\gamma_{\mathbf j,t}^{(0)}
:=
\frac{\sum_{i\in S_t}|\CFF_{\mathbf j}^{(0,i)}|}{(m_t-j_t)|\CEE_{\mathbf j}|},
\qquad
\widetilde\gamma_{\mathbf j+\mathbf e_t,t}^{(1)}
:=
\frac{\sum_{i\in S_t}|\CFF_{\mathbf j+\mathbf e_t}^{(1,i)}|}{(j_t+1)|\CEE_{\mathbf j+\mathbf e_t}|},
\end{equation}
for every admissible pair $(\mathbf j,t)$ with $\mathbf j+\mathbf e_t\in T_l$, and set
\begin{equation}\label{eq:f2_nd_gamma_max}
\widetilde\gamma_{\max}
:=
\max_{\substack{\mathbf j\in T_l,\ 1\le t\le N\\ \mathbf j+\mathbf e_t\in T_l}}
\frac{\widetilde\gamma_{\mathbf j,t}^{(0)}+\widetilde\gamma_{\mathbf j+\mathbf e_t,t}^{(1)}}{2}.
\end{equation}
These quantities are the block analogue of the failure rates in the one-dimensional analysis: they measure the failure proportion not on $\CEE_{\mathbf j}$ itself, but on pairs $(\mathbf y,i)$ along the edge
$\mathbf j\leftrightarrow \mathbf j+\mathbf e_t$.

\begin{lem}\label{lem:f2_nd_failure_bound}
With $\widetilde\gamma_{\max}$ defined in \eqref{eq:f2_nd_gamma_max}, we have
\[
\|D^{(\CFF)}\|
\le
2\widetilde\gamma_{\max}\sum_{t=1}^N g_t(m_t+1).
\]
\end{lem}
\begin{proof}
Fix $\mathbf j\in T_l$ and $t\in\{1,\ldots,N\}$ such that $\mathbf j+\mathbf e_t\in T_l$.
Only the coordinates $i\in S_t$ can contribute to
$D^{(\CFF)}(\mathbf j,\mathbf j+\mathbf e_t)$.
Moreover, the relation $\mathbf y_2=\mathbf y_1+\mathbf e_i$ gives
\[
\left|T^{(i,\CFF)}_{\mathbf j,\mathbf j+\mathbf e_t}\right|
\le
|\CFF_{\mathbf j}^{(0,i)}|+|\CFF_{\mathbf j+\mathbf e_t}^{(1,i)}|.
\]
Therefore, by \eqref{eq:f2_nd_DF},
\begin{align*}
&D^{(\CFF)}(\mathbf j,\mathbf j+\mathbf e_t)  
= 
\frac{g_t}{\sqrt{|\CEE_{\mathbf j}|\,|\CEE_{\mathbf j+\mathbf e_t}|}}
\sum_{i\in S_t}
\Big(|\CFF_{\mathbf j}^{(0,i)}|+|\CFF_{\mathbf j+\mathbf e_t}^{(1,i)}|\Big).
\end{align*}
Using \eqref{eq:f2_nd_gamma0_appendix} and the identity
\[
|\CEE_{\mathbf j+\mathbf e_t}|
=
\frac{m_t-j_t}{j_t+1}|\CEE_{\mathbf j}|,
\]
we obtain
\begin{align*}
&D^{(\CFF)}(\mathbf j,\mathbf j+\mathbf e_t) \\
\le&
\frac{g_t(m_t-j_t)|\CEE_{\mathbf j}|}{\sqrt{|\CEE_{\mathbf j}|\,|\CEE_{\mathbf j+\mathbf e_t}|}}
\widetilde\gamma_{\mathbf j,t}^{(0)}
+
\frac{g_t(j_t+1)|\CEE_{\mathbf j+\mathbf e_t}|}{\sqrt{|\CEE_{\mathbf j}|\,|\CEE_{\mathbf j+\mathbf e_t}|}}
\widetilde\gamma_{\mathbf j+\mathbf e_t,t}^{(1)} \\
=&
g_t\sqrt{(j_t+1)(m_t-j_t)}
\Big(\widetilde\gamma_{\mathbf j,t}^{(0)}+\widetilde\gamma_{\mathbf j+\mathbf e_t,t}^{(1)}\Big) \\
\le&
2g_t\sqrt{(j_t+1)(m_t-j_t)}\,\widetilde\gamma_{\max}
\le g_t(m_t+1)\widetilde\gamma_{\max},
\end{align*}
where in the last step we used $2\sqrt{ab}\le a+b$.
The same bound holds for $D^{(\CFF)}(\mathbf j+\mathbf e_t,\mathbf j)$ by symmetry.

Hence every row indexed by $\mathbf j$ has at most two nonzero neighbors in each block direction $t$, namely $\mathbf j\pm \mathbf e_t$, and so
\[
\sum_{\mathbf k\in T_l}\left|D^{(\CFF)}(\mathbf j,\mathbf k)\right|
\le
2\widetilde\gamma_{\max}\sum_{t=1}^N g_t(m_t+1).
\]
By Gershgorin's circle theorem, the row-sum bound now gives the claimed estimate for $\|D^{(\CFF)}\|$.
\end{proof}

\begin{mproof}[ Proof of Theorem \ref{thm:f2_nd_average}]
By Lemma~\ref{lem:f2_nd_general_expression} and the fact that
\eqref{eq:f2_nd_norm} is independent of $\mathbf v$, we have
\[
\Expect_{v_1,\ldots,v_m}\langle s_{\mathbf g}^{(\CDD)}\rangle
=
\frac{\Expect_{v_1,\ldots,v_m}(\mathbf w^\dagger A^{(\CDD)}\mathbf w)}{\braket{\rho^{(\CDD)}|\rho^{(\CDD)}}}.
\]
Moreover,
\[
\braket{\rho^{(\CDD)}|\rho^{(\CDD)}}
=
\sum_{\mathbf j\in T_l}|w_{\mathbf j}|^2(1-\gamma_{\mathbf j})
\ge
(1-\gamma_{\max})\|\mathbf w\|_2^2,
\]
and also
\[
\braket{\rho^{(\CDD)}|\rho^{(\CDD)}}
\le
\|\mathbf w\|_2^2.
\]
By Lemma~\ref{lem:f2_nd_decomposition},
\[
\Expect_{v_1,\ldots,v_m}A^{(\CDD)}
=
A^{(\mathbf g,l)}-D^{(\CFF)}.
\]
Therefore
\begin{align*}
&\Expect_{v_1,\ldots,v_m}\langle s_{\mathbf g}^{(\CDD)}\rangle  
= 
\frac{\mathbf w^\dagger A^{(\mathbf g,l)}\mathbf w}{\braket{\rho^{(\CDD)}|\rho^{(\CDD)}}}
-
\frac{\mathbf w^\dagger D^{(\CFF)}\mathbf w}{\braket{\rho^{(\CDD)}|\rho^{(\CDD)}}}.
\end{align*}
Since $\mathbf w\ge 0$ entrywise and $A^{(\mathbf g,l)}$ has nonnegative entries, the first term is nonnegative; hence
\[
\frac{\mathbf w^\dagger A^{(\mathbf g,l)}\mathbf w}{\braket{\rho^{(\CDD)}|\rho^{(\CDD)}}}
\ge
\frac{\mathbf w^\dagger A^{(\mathbf g,l)}\mathbf w}{\|\mathbf w\|_2^2}.
\]
For the second term, use
$|\mathbf w^\dagger D^{(\CFF)}\mathbf w|\le \|D^{(\CFF)}\|\,\|\mathbf w\|_2^2$
and Lemma~\ref{lem:f2_nd_failure_bound} to get
\[
-\frac{\mathbf w^\dagger D^{(\CFF)}\mathbf w}{\braket{\rho^{(\CDD)}|\rho^{(\CDD)}}}
\ge
-2\,\frac{\widetilde\gamma_{\max}}{1-\gamma_{\max}}
\sum_{t=1}^N g_t(m_t+1).
\]
Combining the two estimates proves the theorem.
\end{mproof}

\section{Detailed Proof for Block-Structured Hamiltonian DQI}\label{appen:C}

We first establish the following lemma to compute the coefficients $r_{\mathbf{y}^{(1)}, \dots, \mathbf{y}^{(N)}}$ (equivalently $r_{\mathbf{j}}$).

\begin{lem}\label{lem:block_bj_explicit}
Given a commuting Hamiltonian $H_{\mathbf g}
=
\sum_{t=1}^N g_t \sum_{a=1}^{m_t} P_{t,a}$, and a degree-$l$ polynomial 
$P(x)=\sum_{j=0}^l a_j x^j$. Then the coefficient $r_{\Byy^{(1)},\dots,\Byy^{(N)}}$ defined  in \eqref{eq:def_r} can be rewritten as 
\begin{align}
    r_{\Byy^{(1)},\dots,\Byy^{(N)}}
=
2^{-m}
\sum_{s_1=0}^{m_1}\cdots \sum_{s_N=0}^{m_N}
\left(
\prod_{t=1}^N K_{s_t}(|\Byy^{(t)}|;m_t)
\right)
P\!\left(
\sum_{t=1}^N g_t(m_t-2s_t)
\right),
\label{eq:block_bj_explicit_krawtchouk}
\end{align}
for $\sum^N_{t=1}|\Byy^{(t)}|\leq l$. Here, the 
 Krawtchouk polynomial $K_s(j;m)$ is defined as 
\[
K_s(j;m):=
\sum_{u=0}^{s}(-1)^u \binom{j}{u}\binom{m-j}{s-u},
\qquad
0\le s\le m.
\]

\end{lem}

\begin{proof}
First, using a basic identity from the Fourier analysis of Boolean functions, we have:
\[
\mathbf 1_{u\equiv y\ (\bmod 2)}
=
\frac12\sum_{\varepsilon\in\{\pm1\}}\varepsilon^{u+y}.
\]
Extending this to each block $t \in [N]$, we obtain:
\begin{align}
\label{eq:identity}
  \mathbf 1_{\bmu^{(t)}\equiv \Byy^{(t)}\ (\bmod 2)}
=
2^{-m_t}
\sum_{\sigma^{(t)}\in \{\pm1\}^{m_t}}
(\sigma^{(t)})^{\bmu^{(t)}+\Byy^{(t)}},  
\end{align}
where $\sigma^{(t)}=(\sigma^{(t)}_1,\dots,\sigma^{(t)}_{m_t})\in \{\pm1\}^{m_t}$.
For any $\nu^{(t)}\in \Z_{\ge 0}^{m_t}$, we denote
\[
(\sigma^{(t)})^{\nu^{(t)}}
:=
\prod_{a=1}^{m_t} (\sigma^{(t)}_a)^{(\nu^{(t)})_a}.
\]
Recall that $r_{\mathbf{y}^{(1)}, \dots, \mathbf{y}^{(N)}}$, as defined in \eqref{eq:def_r}, is given by:  
\[
 r_{\Byy^{(1)},\ldots,\Byy^{(N)}}
=
\sum_{j=0}^l a_j\, j!
\sum_{\substack{
\bmu^{(t)}\in \Z_{\ge 0}^{m_t}\ \forall t,\\
\sum_{t=1}^N |\bmu^{(t)}|=j,\\
\bmu^{(t)}\equiv \Byy^{(t)}\ (\mathrm{mod}\ 2)\ \forall t
}}
\frac{
\prod_{t=1}^N g_t^{|\bmu^{(t)}|}
}{
\prod_{t=1}^N \bmu^{(t)}!
}.
\]
Substituting the identity from \eqref{eq:identity} into this expression, we have:
\begin{align}
    r_{\Byy^{(1)},\dots,\Byy^{(N)}}
&=
2^{-m}
\sum_{\sigma^{(1)}\in \{\pm1\}^{m_1}}
\cdots
\sum_{\sigma^{(N)}\in \{\pm1\}^{m_N}}
\left(\prod_{t=1}^N (\sigma^{(t)})^{\Byy^{(t)}}\right)
\nonumber\\
&\qquad \times
\sum_{j=0}^l a_j\, j!
\sum_{\substack{
\bmu^{(t)}\in \Z_{\ge 0}^{m_t}\ \forall t,\\
\sum_{t=1}^N |\bmu^{(t)}|=j
}}
\frac{
\prod_{t=1}^N g_t^{|\bmu^{(t)}|}(\sigma^{(t)})^{\bmu^{(t)}}}
{\prod_{t=1}^N \bmu^{(t)}!}.
\label{eq:block_rxy_sigma}
\end{align}

By the multinomial theorem, the innermost sum satisfies
\[
j!
\sum_{\substack{
\bmu^{(t)}\in \Z_{\ge 0}^{m_t}\ \forall t,\\
\sum_{t=1}^N |\bmu^{(t)}|=j
}}
\frac{
\prod_{t=1}^N g_t^{|\bmu^{(t)}|}(\sigma^{(t)})^{\bmu^{(t)}}}
{\prod_{t=1}^N \bmu^{(t)}!}
=
\left(
\sum_{t=1}^N g_t \sum_{a=1}^{m_t} \sigma^{(t)}_a
\right)^j.
\]
Substituting this back into \eqref{eq:block_rxy_sigma} and recognizing the expansion of the polynomial $P$, we obtain:
\begin{align}
r_{\Byy^{(1)},\dots,\Byy^{(N)}}
&=
2^{-m}
\sum_{\sigma^{(1)}\in \{\pm1\}^{m_1}}
\cdots
\sum_{\sigma^{(N)}\in \{\pm1\}^{m_N}}
\left(\prod_{t=1}^N (\sigma^{(t)})^{\Byy^{(t)}}\right)
P\!\left(
\sum_{t=1}^N g_t \sum_{a=1}^{m_t} \sigma^{(t)}_a
\right).
\label{eq:block_rxy_fourier_formula}
\end{align}

Now, consider the vectors $(\mathbf{y}^{(1)}, \dots, \mathbf{y}^{(N)})$ with Hamming weights $|\mathbf{y}^{(t)}| = j_t$. For each $t \in [N]$, let $S_t' := \{ a : (\mathbf{y}^{(t)})_a = 1 \}$ be the set of indices with non-zero entries, where $|S_t'| = j_t$. For any $\sigma^{(t)} \in \{ \pm 1 \}^{m_t}$, let $p_t$ and $q_t$ denote the number of $-1$ entries inside and outside the support $S_t'$, respectively:
\[
p_t:=\#\{a\in S_t':\sigma^{(t)}_a=-1\},
\qquad
q_t:=\#\{a\notin S_t':\sigma^{(t)}_a=-1\}.
\]
It follows that 
\[
(\sigma^{(t)})^{\Byy^{(t)}} = (-1)^{p_t},
\qquad
\sum_{a=1}^{m_t}\sigma^{(t)}_a = m_t-2(p_t+q_t).
\]
The number of such configurations for $\sigma^{(t)}$ is given by the product of binomial coefficients $\binom{j_t}{p_t} \binom{m_t - j_t}{q_t}$. 
Hence, 
\begin{align*}
    r_{\Byy^{(1)},\dots,\Byy^{(N)}}
=&
2^{-m}
\sum_{p_1=0}^{|\Byy^{(1)}|}\sum_{q_1=0}^{m_1-|\Byy^{(1)}|}
\cdots
\sum_{p_N=0}^{|\Byy^{(N)}|}\sum_{q_N=0}^{m_N-|\Byy^{(N)}|}
(-1)^{\sum_{t=1}^N p_t}\\
&\times
\left(
\prod_{t=1}^N
\binom{|\Byy^{(t)}|}{p_t}\binom{m_t-|\Byy^{(t)}|}{q_t}
\right)
P\!\left(
\sum_{t=1}^N g_t\bigl(m_t-2(p_t+q_t)\bigr)
\right).
\end{align*}
Now, let $s_t=p_t+q_t$, and consider the Krawtchouk polynomial~\cite{krawtchouk1929generalisation,nikiforov1991classical} defined as follows
\begin{align*}
    K_{s_t}(|\Byy^{(t)}|;m_t)
    =\sum_{p_t=0}^{s_t}
(-1)^{p_t}
\binom{|\Byy^{(t)}|}{p_t}\binom{m_t-|\Byy^{(t)}|}{s_t-p_t}.
\end{align*}
Therefore, we obtain the result in \eqref{eq:block_bj_explicit_krawtchouk}.

\end{proof}

\begin{mproof}[Proof of complexity in Theorem \ref{thm:block_pauli_poly_state}]
Let us analyze the time complexity in Theorem \ref{thm:block_pauli_poly_state}. 

First, in the step of reference state preparation. 
We first need to prepare the state
$$ \sum_{\mathbf{j} \in T_l} \gamma_{\mathbf{j}} \ket{j_1} \ot \cdots \ot\ket{j_N},$$
where the coefficient $\gamma_{\mathbf{j}} := \frac{\left( \prod_{t=1}^N \binom{m_t}{j_t} \right)^{1/2} r_{\mathbf{j}}}{\mathcal{N}}$.
By Lemma \ref{lem:block_bj_explicit}, each coefficient $\gamma_{\mathbf{j}}$ can be computed classically in $\mathrm{poly}(m, N, \prod_{t=1}^N m_t)$ time. Consequently, the total time required to prepare this weight state is $O(|T_l| \cdot \mathrm{poly}(m, N, \prod_{t=1}^N m_t))$.

Hence, the full reference state
$$\ket{R^l(H_{\mathbf{g}})} = \sum_{\mathbf{j} \in T_l} \gamma_{\mathbf{j}} \ket{D_{m_1, j_1}} \otimes \dots \otimes \ket{D_{m_N, j_N}},$$
which is a superposition of products of Dicke states, can be prepared using $O(\sum_t m^2_t)$ gates by applying the techniques from \cite{bartschi2022short, wang2024quantum}. We note that this step could potentially be further optimized using the sparse Dicke state preparation methods introduced in \cite{khattar2025verifiable}.

Second, in the step of control-Pauli operation, the preparation of Bell state need $O(n)$ gates, and 
the control-Pauli part need $O(m)$ control-Pauli gates. 

Third, in the step of coherent Bell measurement and decoding, 
the coherent Bell measurement  need $O(m)$ gates, and the classical decoding  is assumed to be efficient. 

Finally, the total complexity is $\text{poly}(m, n, |T_l|)$.
In addition, as $|T_l|\leq \prod_{t=1}^N (m_t+1)$. Hence, if $N$ is a fixed constant, then $|T_l|\leq m^N$. 
Therefore, the total complexity is $\text{poly}(m, n)$.

\end{mproof}

\end{appendix}

\bibliographystyle{unsrt}
\bibliography{reference}{}
\end{document}